\newif\iflong
\newif\ifshort

\longtrue

\iflong 
\else
\shorttrue
\fi

\documentclass[sigconf]{aamas}

\usepackage{balance} %

\setcopyright{ifaamas}
\acmConference[AAMAS '23]{Proc.\@ of the 22nd International Conference
on Autonomous Agents and Multiagent Systems (AAMAS 2023)}{May 29 -- June 2, 2023}
{London, United Kingdom}{A.~Ricci, W.~Yeoh, N.~Agmon, B.~An (eds.)}
\copyrightyear{2023}
\acmYear{2023}
\acmDOI{}
\acmPrice{}
\acmISBN{}

\acmSubmissionID{305}

\usepackage{multirow}
\usepackage{booktabs}
\usepackage{diagbox}
\usepackage{subcaption}
\usepackage{wrapfig}

\newcommand{\mytitle}{Hedonic Games With Friends, Enemies, and Neutrals:\\ Resolving Open Questions and Fine-Grained Complexity}
\newcommand{\appendixtitle}{Supplementary Material for the Paper ``\mytitle''}

\title{\mytitle}

\author{Jiehua Chen}
\affiliation{
  \institution{TU Wien}
  \city{Vienna}
  \country{Austria}}
\email{jiehua.chen@tuwien.ac.at}

\author{Gergely Cs{\'a}ji}
\affiliation{
  \institution{KRTK KTI}
  \city{Budapest}
  \country{Hungary}}
\email{csaji.gergely@krtk.hu}

\author{Sanjukta Roy}
\affiliation{
  \institution{Pennsylvania State University}
  \city{}
  \country{United States}}
\email{sanjukta@psu.edu}

\author{Sofia Simola}
\affiliation{
  \institution{TU Wien}
  \city{Vienna}
  \country{Austria}}
\email{sofia.simola@tuwien.ac.at}

\usepackage[utf8]{inputenc}
\usepackage{amsmath,amsthm}
\usepackage{graphicx,dsfont,mathtools}

\usepackage{xspace}
\usepackage{tablefootnote}

\usepackage{xifthen}
\usepackage{paralist}

\usepackage{hyperref}
\hypersetup{%
backref=true, 
pagebackref=true, 
hypertexnames=true,
colorlinks=true,citecolor=green!35!black,linkcolor=red!60!black%
}

\usepackage{tikz}
\usetikzlibrary{decorations,arrows,arrows.meta,petri,topaths,backgrounds,shapes,positioning,fit,calc,decorations.pathreplacing,patterns,intersections}

\pgfdeclarelayer{bg}    %
\pgfdeclarelayer{fg}    %
\pgfsetlayers{bg,main,fg}  %

\newcommand{\ARW}[4][]{%
  \coordinate (p1) at (#4);
  \foreach \ang in {#2}{%
    \draw[#1] (p1) arc (\ang:#3) ;
    \coordinate (p2) at (p1);
    \path (p2) arc (\ang:#3) coordinate (p1);
  }
}

\tikzstyle{agent} = [draw=black, circle, fill=black,  inner sep=1.2pt]
\tikzstyle{agentU} = [draw=black, circle, fill=black,  inner sep=1.2pt, line width=0.8pt]
\tikzstyle{agentW} = [draw=orange, rectangle, fill=orange, inner sep=1.5pt, line width=0.8pt]

\tikzstyle{nn} = [draw=blue, circle, inner sep=1.2pt,fill=blue]
\tikzstyle{nnd} = [draw=blue, circle, inner sep=1.2pt,fill=gray]
\tikzstyle{pn} = [draw=gray, circle, inner sep=1pt,fill=gray]
\tikzstyle{pnn} = [draw=none, circle, inner sep=1pt]
\tikzstyle{sn} = [draw, rectangle, inner sep=1.8pt]
\tikzstyle{ssn} = [draw=green!50!black, rectangle, inner sep=1.8pt,fill=green!50!black]

\tikzstyle{nnn} = [draw=green!60!black, circle, inner sep=1.2pt,fill=green!60!black]
\tikzstyle{sett} = [draw, thick, purple]
\tikzstyle{ele} = [draw, thick, green!70!black]
\tikzstyle{gl} = [draw, gray]
\tikzstyle{fc} = [blue!60!black]
\tikzstyle{ec} = [red!60!black, dashed]

\tikzstyle{ip} = [draw=none, fill = red!20!white]

\tikzstyle{hiddenfc} = [fc, draw opacity=0.3]
\tikzstyle{hiddenec} = [ec, draw opacity=0.3]

\tikzstyle{privaten} = [draw=gray, inner sep=1.2pt, draw=green!60!black, text=green!80!black, opacity=0.8]
\tikzstyle{privatee} = [draw=gray, dashed, green!60!black, draw opacity=0.5]
\tikzstyle{grouplabel} = [inner sep=2pt, fill=white,inner sep=0.5pt]

\usepackage[textsize=tiny,textwidth=1.5cm,linecolor=green!70!black, backgroundcolor=green!10, bordercolor=black,disable]{todonotes}

\usepackage[noend,ruled,linesnumbered]{algorithm2e}

\SetAlFnt{\small}
\SetAlCapFnt{\small}
\SetAlCapNameFnt{\small} 
\SetAlCapHSkip{0pt}

\SetKw{KwAnd}{and}
\SetKw{KwOr}{or}
\SetKw{KwNo}{no}
\SetKw{KwYes}{yes}
\SetKw{KwT}{true}
\SetKw{KwF}{false}
\SetKw{KwST}{s.t.}

\SetKwComment{Mycomment}{$\triangleright$\ }{}

\usepackage[capitalize]{cleveref}

\newtheorem{lemma}{Lemma}
\newtheorem{proposition}{Proposition}
\newtheorem{example}{Example}
\newtheorem{theorem}{Theorem}
\newtheorem{claim}{Claim}[theorem]
\newtheorem{clm}{Claim}[subsection]
\newtheorem{observation}{Observation}

\theoremstyle{definition}
\newtheorem{definition}{Definition}

\crefname{table}{Table}{Tables}
\crefname{figure}{Figure}{Figures}
\crefname{theorem}{Theorem}{Theorems}
\crefname{corollary}{Corollary}{Corollaries}
\crefname{observation}{Observation}{Observations}
\crefname{lemma}{Lemma}{Lemmas}
\crefname{example}{Example}{Examples}
\crefname{reduction}{Reduction}{Reductions}
\crefname{construction}{Construction}{Constructions}
\crefname{subsection}{Subsection}{Subsections}
\crefname{section}{Section}{Sections}
\crefname{claim}{Claim}{Claims}
\crefname{clm}{Claim}{Claims}
\crefname{algorithm}{Algorithm}{Algorithm}
\crefname{definition}{Definition}{Definitions}

\newcommand{\decprob}[3]{
   \begin{center}%
     \begin{itemize}[d]
       \item[\textsc{#1}]
       \item[\textbf{Input:}]  #2\\[0.2ex]
       \item[\textbf{Question:}]  #3
     \end{itemize}
  \end{center}
}

\newcommand{\verif}{\textsc{CoreV}\xspace}
\newcommand{\sverif}{\textsc{StrictCoreV}\xspace}
\newcommand{\togethersverif}{\textsc{(S)CoreV}\xspace}

\newcommand{\xctg}{\textsc{E-X3C}\xspace}
\newcommand{\xct}{\textsc{X3C}\xspace}

\newcommand{\clique}{\textsc{Clique}\xspace}
\newcommand{\pxct}{\textsc{Planar-X3C}\xspace}

\newcommand{\NS}{\textsc{NashEx}\xspace}
\newcommand{\IndS}{\textsc{IndividEx}\xspace}

\newcommand{\appsymb}{$\star$}

\newcommand{\hx}{\ensuremath{\hat{x}}}
\newcommand{\hy}{\ensuremath{\hat{y}}}
\newcommand{\enn}{\ensuremath{\hat{n}}}
\newcommand{\emm}{\ensuremath{\hat{m}}}

\newcommand{\seqq}[1]{\ensuremath{\langle #1 \rangle}}

\newcommand{\setind}{\ensuremath{p}}
\newcommand{\sstar}{\ensuremath{\triangledown}}
\newcommand{\dummys}{\ensuremath{{\color{blue}d^{\sstar}\!(i)}}}
\newcommand{\dummyin}{\ensuremath{{\color{blue}d^{\textsf{in}}(i)}}}
\newcommand{\dummyinT}{\ensuremath{{\color{gray!60!black}d^{\textsf{in}}(i)_1}}}
\newcommand{\dummyout}{\ensuremath{{\color{blue}d^{\textsf{out}}(i)}}}
\newcommand{\Dummys}{\ensuremath{D}}

\newcommand{\myemph}[1]{{\color{green!25!black}\emph{#1}}}

\newcommand{\fas}{\ensuremath{\mathsf{f}}}
\newcommand{\maxdeg}{\ensuremath{\Delta}}
\newcommand{\maxfdeg}{\ensuremath{\Delta_{F}}}

\newcommand{\maxcoal}{\ensuremath{\kappa}}

\definecolor{darkgreen}{rgb}{0.01,0.6,0.1}
\definecolor{darkblue}{rgb}{0,0,0.4}
\definecolor{winered}{rgb}{0.6,0.1,0.1}
\definecolor{doncolor}{RGB}{78,154,0}
\definecolor{falsecolor}{RGB}{0,55,255}
\definecolor{truecolor}{RGB}{164,0,0}
\definecolor{lightblue}{rgb}{0.527,0.805,0.977}

\newcommand{\np}{\ensuremath{\mathsf{NP}}}
\newcommand{\npc}{$\np$-{\normalfont\text{c}}}
\newcommand{\nph}{$\np$-{\normalfont\text{h}}}

\newcommand{\conp}{\ensuremath{\mathsf{coNP}}}
\newcommand{\conph}{$\conp$-{\normalfont\text{h}}}
\newcommand{\conpc}{$\conp$-{\normalfont\text{c}}}

\newcommand{\cowone}{\ensuremath{\mathsf{coW}[1]}}

\newcommand{\cowoneh}{$\cowone$-h}

\newcommand{\pp}{$\mathsf{P}$}
\newcommand{\xp}{$\mathsf{XP}$}
\newcommand{\fpt}{$\mathsf{FPT}$}

\newcommand{\hardplanarsymmsymb}{$^\clubsuit$}
\newcommand{\hardplanarsymb}{$^\dagger$}
\newcommand{\harddagsymb}{$^\blacklozenge$}
\newcommand{\hardsymsymb}{$^\spadesuit$}
\newcommand{\hardsymballpara}{$^\ddagger$}
\newcommand{\easynash}{{\Large $^{\diamondsuit}$}}
\newcommand{\easyfennash}{$^{\circ}$}

\newcommand{\good}{\ensuremath{\mathsf{\color{blue!70!black}g}}}
\newcommand{\bad}{\ensuremath{\mathsf{\color{red!60!black}b}}}

\newcommand{\goodG}{\ensuremath{G^{\good}}}
\newcommand{\badG}{\ensuremath{G^{\bad}}}

\newcommand{\hG}{\hat{G}}
\newcommand{\hE}{\hat{E}}
\newcommand{\hV}{\hat{V}}

\newcommand{\hgoodG}{\ensuremath{\hat{G}^{\good}}}

\newcommand{\singles}{\ensuremath{V_{\mathsf{S}}}}
\newcommand{\nonsingles}{\ensuremath{V_{\mathsf{NS}}}}
\newcommand{\TB}{\ensuremath{T_B}}
\newcommand{\TBs}{\ensuremath{T_{B^*}}}
\newcommand{\Bs}{\ensuremath{B_{\mathsf{S}}}}
\newcommand{\Bns}{\ensuremath{B_{\mathsf{NS}}}}
\newcommand{\frPi}{\ensuremath{n_i(\Pi)}}
\newcommand{\frBns}{\ensuremath{n_i(\Bns)}}
\newcommand{\hBns}{\ensuremath{\hat{B}_{\mathsf{NS}}}}
\newcommand{\coloring}{\ensuremath{\chi}}

\newcommand{\sing}{\ensuremath{S}}
\newcommand{\Colors}{\ensuremath{C}}

\newcommand{\FEN}{\textsc{FEN}\textsuperscript{\textsc{s}}}

\newcommand{\FE}{\textsc{FE}\textsuperscript{\textsc{s}}}

\newcommand{\needtowrite}[1]{#1}
\newcommand{\needtoread}[1]{#1}

\tikzset{ttrue/.style={color=truecolor!50!white}}
\tikzset{ffalse/.style={color=falsecolor!50!white}}
\tikzset{dont/.style={color=doncolor!50!white}}
\tikzset{trueline/.style =   {line width= 3pt, ttrue}}
\tikzset{falseline/.style =   {line width= 3pt, ffalse}}
\tikzset{dontline/.style =   {line width= 3pt, dont}}
\def \xss {8ex}

\newcommand{\elementpic}{
  \node[] at (-1,0.3) (uim) {};
  \node[nn] at (0,0) (ui) {};
  \node[nn] at (.8,0) (uip) {};
  \node[] at (1.5,0.3) (uipp) {};
  \node[above left = 0pt and -1pt of ui] {$u_i$};
  \node[above = 0pt of uip] {$u_{i+1}$};
  \node[sn] at (-.5,-.6) (s1) {};
  \node[sn] at (.5,-.6) (s3) {};
  \node[sn] at (0,.6) (s2) {};

  \node[below=0pt of s1] {$w_j$};
  \node[above=0pt of s2] {$w_{q}$};
  \node[below=0pt of s3] {$w_{\setind}$};

  \path[draw] (ui) -- (s1);
  \path[draw] (ui) -- (s2);
  \path[draw] (ui) -- (s3);

  \draw[gl,dashed] (uim) to[bend right=10]  (ui);
  \draw[gl,dashed] (uip) to[bend left=4] (ui);
  \draw[gl,dashed] (uip) to[bend right=10] (uipp);

}

\newcommand{\setpic}{
  \node[sn] at (0,0) (wj) {};

  \node[nn] at (.5,1) (u1) {};
  \node[nn] at (1,0) (u2) {};
  \node[nn] at (-1,0) (u3) {};

  \foreach \i in {1,2,3} {
    \draw[] (u\i) -- (wj);
  }

  \foreach \i / \p / \a / \r / \n in {wj/below left/-1pt/-2pt/{w_j},
    u1/above left/-1pt/-2pt/{u_i}, u2/above right/-1pt/-3pt/{u_k}, u3/above left/-1pt/-3pt/{u_r}}{   
    \node[\p = \a and \r of \i] {$\n$};
  }
}

\newcommand{\trianglegadget}[1]{
   \foreach  \deg / \n / \rr in {-60/a/\rad, 180/b/\rad, 60/c/\rad}
   {
    \node[pn] at (\deg:\rr) (#1\n) {};
  }

    \foreach \n / \nn / \p / \l / \r  / \c in
    {a/{0}/below left/1/-5/black, c/{2}/right/0/1/black, b/{1}/below/0/0/black}{ 
    \node[\p = \l pt and \r pt of #1\n, text=\c, inner sep=0pt, fill=white] {$#1^{\nn}_j$};
  }
}
\newcommand{\fescoreverifdeltalong}{
  \conpc ($\maxdeg=4$)\hardplanarsymmsymb
}
\newcommand{\fescoreverifdelta}{
  \conpc\hardplanarsymmsymb
}
\newcommand{\fescoreverifdeltacite}{
 [T\ref{thm:verif-non-symmetric-four}]
}
\newcommand{\fenscoreverifdelta}{
  \needtoread{\conpc\hardsymsymb}
}
\newcommand{\fenscoreverifdeltacite}{
 [T\ref{thm:core_verify_neutals}]
}

\newcommand{\fescoreverifcoallong}{
  \needtoread{\cowoneh\hardsymsymb}, \xp
}

\newcommand{\fescoreverifcoalcite}{
 [T\ref{thm:W1h+XP_largest_coalition}]
}

\newcommand{\fenscoreverifcoal}{
  \needtoread{\conpc\hardsymsymb}
}
\newcommand{\fenscoreverifcoalcite}{
  [T\ref{thm:core_verify_neutals}]
}

\newcommand{\fescoreveriffas}{
 \needtoread{\conpc}
}
\newcommand{\fescoreveriffascite}{
 [T\ref{thm:fas_number1}]
}
\newcommand{\fenscoreveriffas}{
   \needtoread{\conpc}
}
\newcommand{\fenscoreveriffascite}{
  [T\ref{thm:verif-fas+deg+k}]
}

\newcommand{\fescoreverifcoaldelta}{
 {\fpt}
}
\newcommand{\fescoreverifcoaldeltacite}{
 [T\ref{thm:fpt_coalition+deg}]
}

\newcommand{\fenscoreverifcoaldelta}{
  \needtoread{\conpc\hardsymsymb\hardsymballpara}
}
\newcommand{\fenscoreverifcoaldeltacite}{
 [T\ref{thm:core_verify_neutals}]
}

\newcommand{\fescoreverifcoalfas}{
 \needtowrite{\fpt}
}
\newcommand{\fescoreverifcoalfascite}{
 [T\ref{thm:fe-core-fpt-k-delta}]
}
\newcommand{\fenscoreverifcoalfas}{
 \needtoread{\conpc\hardsymballpara}
}

\newcommand{\fenscoreverifcoalfascite}{
 [T\ref{thm:verif-fas+deg+k}]
}

\newcommand{\fescoreverifdeltafas}{
  \needtoread{\conpc}
}
\newcommand{\fescoreverifdeltafascite}{
  [T\ref{thm:sverif_cont_fas+delta}]
}

\newcommand{\fenscoreverifdeltafas}{
  \needtoread{\conpc\hardsymballpara}
}
\newcommand{\fenscoreverifdeltafascite}{
   [T\ref{thm:verif-fas+deg+k}]
 }

\newcommand{\appendixsymb}{$\star$}
\usepackage{etoolbox} %
\newcommand{\tocommentout}[1]{%
}

\newcommand{\appendixproofwithstatement}[3]{%
  \gappto{\appendixtext}{
    \subsection{Proof of \cref{#1}}\label{proof:#1}
    \noindent {\normalfont\emph{#2}}
    #3
  }
}

\newcommand{\appendixsection}[1]{%
  \gappto{\appendixtext}{
    \section{Additional Material for Section~\ref{#1}}
    \label{appsec:#1}
  }
}

\newcommand{\appendixalg}[4]{%
  #2
  \gappto{\appendixtext}{
    \subsection{Continuation of \cref{#1}}\label{alg:#1}%
    \noindent{\normalfont\emph{#3}}

    {#4}
    }
}

\begin{abstract}
  We investigate verification and existence problems for prominent stability concepts in hedonic games with friends, enemies, and optionally with neutrals~\cite{dimitrov2006simple,ohta2017core}.
  We resolve several (long-standing) open questions~\cite{woeginger2013core,ReyRotheSchSch2016wonderful,ohta2017core,BOSY19FEN-IS} and
  show that for friend-oriented preferences,
  under the \emph{friends and enemies} model, it is coNP-complete to verify whether a given agent partition is (strictly) core stable, %
  while under the \emph{friends, enemies, and neutrals} model, it is NP-complete to determine whether an individual stable partition exists. 
  We further look into natural restricted cases from the literature, such as
  when the friends and enemies relationships are symmetric, 
  when the initial coalitions have bounded size, 
  when the vertex degree in the friendship graph (resp.\ the union of friendship and enemy graph) is bounded,
  or when such graph is acyclic or close to being acyclic.
  We obtain a complete (parameterized) complexity picture regarding these cases.
\end{abstract}


\begin{document}

\maketitle

\section{Introduction}\label{sec:intro}

Hedonic games, introduced by %
Dr{\'e}ze and Greenberg~\cite{DG80hedonic}, are coalition formation games where each agent's preferences over possible coalitions (i.e., subsets of agents) depend only on the members in the respective coalitions.
The goal is to partition the agents into disjoint coalitions which are ``stable''. 
Typical stability concepts include \emph{(strict) core stability}, \emph{Nash stability}, and \emph{individual stability}~\cite{BaKoSo2001Hedonic,bogomolnaia2002stability,Ballester2004,gairing2010computing,SungDimAdditiveHG2010,peters2015simple,peters2016graphical,kerkmann2020hedonic}.
Briefly put, a partition is \myemph{core} stable if no subset~$S$ of agents can strictly improve by joining~$S$, and it is \myemph{strictly core} stable if no subset~$S$ of agents can weakly improve by joining~$S$ whereas at least one agent can strictly improve. 
The partition is \myemph{Nash} stable if it is \myemph{individually rational} (i.e., no agent prefers to be alone), and no agent envies another coalition (i.e., prefers to be in this coalition rather than her own).
It is \myemph{individually} stable if it is individually rational, and no agent
envies another coalition and this coalition is fine with accepting her.

The existence of a stable partition and the computational complexity of determining whether such partition exists depends on the representation of the preferences of each agent~\cite{woeginger2013core}.
To simplify the representation of the preferences, Dimitrov et al.~\cite{dimitrov2006simple} introduce the so-called hedonic games with friends and enemies, where there is a directed graph on the agents (the so-called \myemph{friendship graph}) such that an agent~$x$ considers another agent~$y$ a friend if there is an arc from~$x$ to $y$; otherwise, $x$ considers $y$ an enemy. 
Depending on whether more friends or fewer enemies are preferred, Dimitrov et al.\ distinguish between \myemph{friend-oriented} and \myemph{enemy-oriented} preferences. 
Under \myemph{friend-oriented} preferences, when comparing two coalitions, an agent prefers the one with more friends, and for the same number of friends, she prefers fewer enemies,
while under \myemph{enemy-oriented} preferences, an agent prefers the coalition with fewer enemies, and for the same number of enemies, she prefers more friends.  %
Recently, Brandt et al.~\cite{Brandt_Bullinger_Tappe_2022} show that it is NP-complete to determine the existence of a Nash stable partition under friend-oriented preferences.
Dimitrov et al.~\cite{dimitrov2006simple} show that 
under friend-oriented preferences, there is always a strictly core stable partition (which is hence core stable and individually stable),
and under the enemy-oriented preferences, a core stable partition always exists. 
However, the computational effort to find these partitions is different:
Under the friend-oriented preferences, the strongly connected components in the friendship graph form a strictly core stable partition and can be found in linear time,
whereas under the enemy-oriented preferences, it is NP-hard to find a core stable partition~\cite{SungDimitrov2007enemy,woeginger2013core} and beyond NP to find a strictly core stable partition~\cite{ReyRotheSchSch2016wonderful}.
One question that has remained open for a decade asks what the complexity of the core verification in the friend-oriented case is~\cite{woeginger2013core,ReyRotheSchSch2016wonderful,ohta2017core,BOSY19FEN-IS}; it was conjectured to be polynomial-time solvable by Woeginger~\cite{woeginger2013core}.

\looseness=-1
\begin{table*}[t]
  \renewcommand{\aboverulesep}{0pt} 
    \renewcommand{\belowrulesep}{0pt}
    \renewcommand{\arraystretch}{1.1}
    
  \caption{
   \FE-{\togethersverif} refers to the problems~\verif and \sverif in \FE. %
   See Section~\ref{sec:defi} for the definitions of $\maxdeg$, $\maxcoal$, and $\fas$.  %
   ``\hardplanarsymmsymb'' means hardness holds even for planar graphs (and for symmetric preferences, albeit with a larger, constant max degree~[T\ref{thm:verif-symmetric-eight}]).
   ``\hardplanarsymb'' (resp.\ ``\hardsymsymb'') means hardness holds even for planar graphs (resp.\ symmetric preferences) . 
   ``\harddagsymb'' (resp.\ ``\hardsymballpara'') means hardness holds even when the enemy graph is acyclic (resp.\ $\maxcoal+\maxdeg+\fas$ is a constant).
   ``\easynash'' (resp.\ ``\easyfennash'') means polynomial even if $\fas=2$ (resp.\ only the friendship graph is acyclic). $^\ast$~Brandt et al.~\cite{Brandt_Bullinger_Tappe_2022} recently showed it to be \nph. However, their reduction does not bound the maximum degree or the feedback arc set and it is not planar.}
 \begin{tabular}{@{}c@{\;}|@{\;} c@{\,}c@{\;} c  c@{\,}c@{\,} c  c@{\,}c@{\,} @{}c@{} | c@{\,}c@{} c c@{}c @{\,} c c@{}c@{}}
   \toprule
   & \multicolumn{2}{c}{\FE-{\togethersverif}} &&  \multicolumn{2}{c}{\FEN-{\verif}} && \multicolumn{2}{c}{\FEN-{\sverif}} && \multicolumn{2}{c}{\FE-{\NS}}&&  \multicolumn{2}{c}{\FEN-{\NS}}&&  \multicolumn{2}{c}{\FEN-{\IndS}}\\\midrule
   \diagbox[width=2.4cm, height=.6cm]{\raisebox{-2pt}{\scriptsize restrictions~~~~~~~~~~~~~~~~~~~~~~~~\hspace*{2cm}}}{\raisebox{2pt}{\scriptsize ~~~always exists?}}
   & \multicolumn{2}{c}{yes} & & \multicolumn{2}{c}{no} & & \multicolumn{2}{c}{no} & & \multicolumn{2}{c}{no} & & \multicolumn{2}{c}{no} & & \multicolumn{2}{c}{no}\\\cline{1-1}\cline{2-3}\cline{5-6}\cline{8-9}\cline{11-12}\cline{14-15}\cline{17-18} \\[-2ex]
   $\maxdeg$%
   & \fescoreverifdelta &  \fescoreverifdeltacite  && \fenscoreverifdelta  & \fenscoreverifdeltacite&& \fenscoreverifdelta  & \fenscoreverifdeltacite && {\npc\hardplanarsymb$^\ast$} & [T\ref{thm:fe-ns-hard}] && \needtoread{\npc\harddagsymb}  & [T\ref{thm:ns-deg_fas}] && {\npc\harddagsymb}  & [T\ref{thm:FEN-is-fas+delta}] 
   \\
   $\fas$%
  & \fescoreveriffas & \fescoreveriffascite  && \fenscoreveriffas &\fenscoreveriffascite && \fenscoreveriffas &\fenscoreveriffascite && \needtoread{\npc} & [T\ref{thm:fe-nash-fas-delta-nph}] && \needtoread{\npc\harddagsymb} & [T\ref{thm:ns-deg_fas}] && \needtoread{\npc\harddagsymb}  & [T\ref{thm:FEN-is-fas+delta}] \\
$\maxdeg+ \fas$ 
   &  \fescoreverifdeltafas & \fescoreverifdeltafascite &&\fenscoreverifdeltafas& \fenscoreverifdeltafascite&&\fenscoreverifdeltafas& \fenscoreverifdeltafascite && \needtoread{\npc} & [T\ref{thm:fe-nash-fas-delta-nph}] && \needtoread{\npc\harddagsymb}  & [T\ref{thm:ns-deg_fas}] && \needtoread{\npc\harddagsymb}  & [T\ref{thm:FEN-is-fas+delta}]\\
   $\maxcoal$%
  & \fescoreverifcoallong & \fescoreverifcoalcite && \fenscoreverifcoal  & \fenscoreverifcoalcite&& \fenscoreverifcoal  & \fenscoreverifcoalcite && -- & -- && -- & -- && -- & -- \\
$\fas+ \maxcoal$ 
& \fescoreverifcoalfas & \fescoreverifcoalfascite &&\fenscoreverifcoalfas & \fenscoreverifcoalfascite &&\fenscoreverifcoalfas & \fenscoreverifcoalfascite && -- & -- && -- & -- && -- & --\\  
$\maxdeg+\maxcoal$ 
& \fescoreverifcoaldelta& \fescoreverifcoaldeltacite &&  \fenscoreverifcoaldelta &\fenscoreverifcoaldeltacite &&  \fenscoreverifcoaldelta &\fenscoreverifcoaldeltacite  &&  -- & -- && -- & -- && -- & --\\

   symm.\ & \fescoreverifdelta & \fescoreverifdeltacite & & \fenscoreverifdelta  & \fenscoreverifdeltacite & &\fenscoreverifdelta  & \fenscoreverifdeltacite & &
           \needtoread{\pp} & [O\ref{thm:sym-ns-is}] && \needtoread{\pp} & [O\ref{thm:sym-ns-is}] && \needtoread{\pp} & [O\ref{thm:sym-ns-is}]\\
   DAG & \needtoread{\pp} & [P\ref{obs:FE-acyclic}] && \needtoread{\pp} &  [P\ref{prop:dag-friends-core-P}] & & \needtoread{\conpc} &  [T\ref{thm:verif-fas+deg+k}] & & \needtoread{\pp\easynash} & [T\ref{thm:fe-nash-fas-delta-nph}] && \needtoread{\pp} & [T\ref{thm:ns_is_acyclic}] && \needtoread{\pp\easyfennash} & [T\ref{thm:ns_is_acyclic}]\\\bottomrule
   
 \end{tabular}
\label{results}
\end{table*}
Ota et al.~\cite{ohta2017core} extend the model of Dimitrov et al.\ by also allowing agents to be \emph{neutral} to other agents who do not impact the preferences, and show that the same approach of Dimitrov et al.\ gives rise to a core stable partition under friend-oriented preferences.
They leave open the complexity of verifying core stable partitions.
Barrot et al.~\cite{BOSY19FEN-IS} show that this model may not admit individually stable partitions and leave open the complexity of determining whether one exists.
As far as we know, Nash stability has not been studied in the context with neutrals.

Both models, with or without neutrals, are a restriction of \emph{hedonic games with additive preferences} where it is NP-complete to decide whether a Nash stable or individually stable partition exists~\cite{SungDimAdditiveHG2010}, and it is $\Sigma^{\text{p}}_2$-complete to decide whether a core stable or strictly core stable partition exists~\cite{Woeginger2013AHGcore,Peters2017AHG_SC}.

In this paper, we focus on the friend-oriented model and resolve long-standing open questions by showing that all mentioned problems whose complexity was unknown are in fact intractable (either \conp- or \np-complete). In particular, we refute Woeginger's conjecture~\cite{woeginger2013core,ReyRotheSchSch2016wonderful} and show that verifying core stable partitions is \emph{not} polynomial-time solvable unless \pp{}$=${}\np.
To understand the true causes of the intractability results and to explore the line between easy and hard cases, we further look into interesting restricted scenarios such as planar or acyclic graphs, and natural parameters such as maximum degree~$\maxdeg$ and feedback arc set number~$\fas$ of the input graph, and also the size~$\maxcoal$ of the largest coalition in a given partition.
We analyze and obtain a complete picture of fine-grained complexity of both the verification and existence problems with respect to the four stability concepts and under friend-oriented preferences.
Our results are given in \cref{results}. We summarize our main contributions as follows.
\begin{compactitem}[--]
  \item First and foremost, we establish that it is \conp-complete to decide whether a given partition is core stable or strictly core stable, even in the case without neutrals (see \cref{thm:verif-non-symmetric-four}),  
  and it is \np-complete to decide whether an individually stable partition exists in the case with neutrals (see \cref{thm:FEN-is-fas+delta}).

   The first result has both theoretical and practical significance:
  (1) The reduction is based on a novel friendship gadget, which may be of independent interest for other hardness reductions for hedonic games; 
  (2) It not only showcases a rare complexity situation where verification is much harder than searching, but it can also be served as a complexity barrier against manipulation; e.g., when an agent or a subset of agents want to know if it is beneficial to maintain the status quo rather than to deviate, they essentially need to solve the \conp-hard verification question.

\item Second, we show that assuming the friends and enemies relationship graph to be acyclic (DAG) almost always ensures polynomial-time solvability.
The strict core verification problem with neutrals is the only exception.
Moreover, we obtain complexity dichotomies with regards to the distance to being a DAG, the so-called feedback arc set number~$\fas$.
We note that DAGs or relationship graphs with small~$\fas$ occurs, for instance, for authors when the friendships are based on the popularity of authors.  A prominent author is followed by many other authors, whereas, this relation is often asymmetric and ordering the authors according to their popularity can yield a small $\fas$.

\item Third, strengthening known and own results, we show that assuming the relationship graph to be planar (e.g., when the agents are located on the plane) or sparse (i.e., the maximum degree~$\maxdeg$ in the relationship graph is bounded since each agent typically only knows a few other agents) does not lower the complexity. %
\item Finally, for the verification problem where a partition is given,
we show that under the friends and enemies model, if the initial coalitions have small constant sizes, then the problem can be solved in polynomial time, i.e., an \xp\ algorithm wrt.~$\maxcoal$, but this parameter alone cannot yield fixed-parameter (FPT) algorithms under standard complexity theoretic assumptions.
Combining with $\fas$ or $\maxdeg$, we obtain FPT algorithms. %
Note that the algorithm for the combined parameter~$(\maxcoal,\fas)$ is based on a reduction to \textsc{Directed Subgraph Isomorphism} where the pattern graph is of size $O(\maxcoal^2)$. Our crucial observation further reduces it to the case where the pattern graph is indeed a directed in-tree, enabling us to design an algorithm with desired running time.
The problem is much harder when neutrals are present; both core verification problems remain \conp-hard even if $\maxcoal+\maxdeg+\fas$ is a constant.
\end{compactitem}

\smallskip
\noindent \textbf{Paper structure.} In \cref{sec:defi}, we define the model and relevant concepts, the central problems, and parameters. In \cref{sec:fe,sec:fen}, we consider the model without neutrals and with neutrals, respectively. We conclude in \cref{sec:conclude}.
\ifshort
Due to space constraints, proofs of the results and additional materials marked with (\appendixsymb) are deferred to the full version of the paper \cite{full_version}.
\else
Proofs of the results and additional materials marked with (\appendixsymb) are deferred to the appendix.
\fi
\section{Basic definitions and fundamentals} %
\label{sec:defi}

Given an integer~$t$, let \myemph{$[t]$} denote the set~$\{1,2,\ldots,t\}$. Given a directed graph~$G$ and a vertex~$v$, the sets~$N^+_G(v)$ and $N^-_G(v)$ denote the out- and in-neighborhood of~$v$.
An instance of \myemph{\textsc{Hedonic Games}} consists of a set~$V$ of agents and for each agent a preference order (with possibly ties) over non-empty agent subsets, called \myemph{coalitions}, which contains her.
In this paper, we focus on a natural and simple variant of \textsc{Hedonic Games} where each agent regards every other agent either good (i.e., a friend), or bad (i.e., an enemy), or neutral such that agents' preferences are friends oriented.
Formally, we are given a set of agents~$V$ and two directed graphs on~$V$, called \myemph{friendship} graph~$\goodG$ and \myemph{enemy} graph~$\badG$ with \emph{disjoint} arc sets, such that an agent~$i$ regards another agent~$j$ as friend (resp.\ enemy) whenever $\goodG$ (resp.\ $\badG$) contains the arc~$(i,j)$; $i$ considers $j$ neutral if neither $\goodG$ nor $\badG$ contains $(i, j)$.

For each agent~$i\in V$, the preference order~$\succeq_i$ of $i$ is derived as follows:
For two coalitions~$S$ and $T$ containing~$i$, agent $i$ \myemph{(strictly) prefers}~$S$ to~$T$, written as \myemph{$S \succ_i T$}, if 
\begin{inparaenum}[(i)]
  \item either $|N^+_{\goodG}(i)\cap S| > |N^+_{\goodG}(i)\cap T|$, 
  \item or $|N^+_{\goodG}(i)\cap S| = |N^+_{\goodG}(i)\cap T|$ and $|N^+_{\badG}(i)\cap S| < |N^+_{\badG}(i)\cap T|$. 
\end{inparaenum}
Agent~$i$ is \myemph{indifferent} between $S$ and $T$, written as \myemph{$S \sim_i T$},
if  $|N^+_{\goodG}(i)\cap S| = |N^+_{\goodG}(i)\cap T|$ and $|N^+_{\badG}(i)\cap S| = |N^+_{\badG}(i)\cap T|$. 
Agent~$i$ \myemph{weakly prefers}~$S$ to~$T$ if $S\succ_i T$ or $S\sim_i T$. Note that the number of neutral agents in the coalition does not affect agent's preferences regarding that coalition.%

We call $V$ the \myemph{grand} coalition. A \myemph{coalition structure}~$\Pi$ of $V$ is a partition of $V$ into disjoint coalitions, i.e., the coalitions $V'$ in~$\Pi$ are pairwise disjoint and $\bigcup_{V'\in \Pi} V' = V$.
We will use coalition structure and partition interchangeably. 
Given a coalition structure~$\Pi$ of $V$ and an agent~$i\in V$, let \myemph{$\Pi(i)$} denote the coalition which contains~$i$.
A coalition~$W$ is \myemph{strictly blocking} a coalition structure~$\Pi$ if every agent~$i \in W$ strictly prefers~$W$ to $\Pi(i)$, and it is \myemph{weakly blocking}~$\Pi$ if every agent~$i \in W$ weakly prefers~$W$ to $\Pi(i)$ and at least one agent~$i\in W$ strictly prefers~$W$ to $\Pi(i)$.

We use \myemph{\FEN}\ to denote the \textsc{Hedonic Games} variant with friends, enemies, and neutrals, and use \myemph{\FE}\ to denote the restricted variant of \FEN\ where no agent is neutral to any other agent, i.e., $(i,j)\in E(\goodG)\cup E(\badG)$ holds for all distinct agents $i$ and $j$.
Note that the superscript \textsuperscript{\textsc{s}} refers to simple and is added to distinguish from the abbreviation~\textsc{FEN} used in the literature~\cite{kerkmann2020hedonic}.
For \FE, we follow the convention in the literature and only specify the friendship relation.
Due to this, in the remainder of the paper, we assume that an \FE\ instance consists of the friendship graph only.

\smallskip
\noindent \textbf{(Strictly) core stable coalition structures, Nash and individual stability.}
Let $\Pi$ be a coalition structure.
We say that $\Pi$ is \myemph{core stable} (resp.\ \myemph{strictly core stable}) if no coalition is strictly (resp.\ weakly) blocking~$\Pi$.
Clearly, by definition, a strictly core stable partition is also a core stable one.
We call $\Pi$ \myemph{Nash stable} if it is \myemph{individually rational} (i.e., no agent prefers to be alone)
and no agent \myemph{envies} another coalition (i.e., no agent~$x$ and coalition~$C$ in $\Pi$ exist such that $x$ prefers $C\cup \{x\}$ to her own).
It is \myemph{individually stable} if it is individually rational,
and no agent~$x$ and coalition $C\in  \Pi$ form a \myemph{blocking tuple}
(i.e., $x$ envies $C$ and each agent~$j\in C$ weakly prefers $C\cup \{x\}$ to $C$).

\begin{example}\label{ex:prelim}
  The graph on the left (with blue arcs only) is an instance of \FE, where each arc specifies the friendship relation.  
  The coalition structure, derived from the strongly connected components, $\Pi_1=\{\{1,2\},\{3\}, \{4\}\}$ is strictly core stable, but not Nash stable since $3$ wants to join~$\{4\}$.
  Indeed, there is no Nash stable solution.
  The graph on the right (blue arcs indicating friends while red arcs enemies) is an instance of \FEN.
  The coalition structure $\Pi_2=\{\{1,2\},\{3,4\}\}$ is strictly core stable and Nash stable.

  \def \xx {1.7}
  
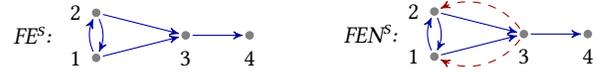
{\centering
    \begin{tikzpicture}[>=stealth',shorten <= 1pt, shorten >= 1pt]
    \foreach \i / \j / \n in {-.2/-0.3/1, -.2/.3/2, 0.5/0/3, 1/0/4} {
      \node[pn] at (\i*\xx, \j) (\n) {};
    }
    \foreach \n / \p / \l in {1/left/1, 2/left/1, 3/below/1, 4/below/1} {
      \node[\p = \l pt of \n] {$\n$}; 
    }
    \foreach \s / \t in {2/3, 1/3, 3/4} {
      \draw[->, fc] (\s) -- (\t);
    }
    \foreach \s / \t in {1/2, 2/1} {
      \draw[->, fc] (\s) edge[bend left=25] (\t);
    }
    \node at (-0.7*\xx, 0) {\FE:};
  \end{tikzpicture}~~\;\qquad\;~~
  \begin{tikzpicture}[>=stealth',shorten <= 1pt, shorten >= 1pt]
      \foreach \i / \j / \n in {-.2/-0.3/1, -.2/.3/2, 0.5/0/3, 1/0/4} {
      \node[pn] at (\i*\xx, \j) (\n) {};
    }
    \foreach \n / \p / \l in {1/left/1, 2/left/1, 3/below/1, 4/below/1} {
      \node[\p = \l pt of \n] {$\n$}; 
    }

    \foreach \s / \t in {2/3, 1/3, 3/4} {
      \draw[->, fc] (\s) -- (\t);
    }
    \foreach \s / \t in {1/2, 2/1} {
      \draw[->, fc] (\s) edge[bend left=25] (\t);
    }

    \foreach \s / \t / \aa in {3/2/40, 3/1/-40} {
      \draw[->, ec] (\s) edge[bend right=\aa] (\t);
    }
    \node at (-0.7*\xx, 0) {\FEN:};
  \end{tikzpicture}
 \par}
\end{example}

\noindent The following relation is known from the literature~\cite{bogomolnaia2002stability,dimitrov2006simple,ohta2017core}. %
\begin{proposition}
  \begin{compactenum}[(i)]
    \item Every strictly core stable coalition structure is individually stable and core stable.
    \item Nash stability implies individual stability.
    \item For \FE, a strictly core stable coalition structure always exists and it can be found in linear time.
    \item For \FEN, a core stable coalition structure always exists and it can be found in linear time. 
  \end{compactenum}
\end{proposition}

\smallskip

\noindent \textbf{Central problems.}
We are interested in the following core verification problems.
\decprob{\FEN-\verif~({\normalfont \text{resp.}} \FE-\verif)}{
 An \FEN\ instance~$(V, \goodG, \badG)$ (resp.\ \FE\ instance $(V, \goodG)$), and a coalition structure~$\Pi$ on~$V$.
}{
  Is $\Pi$ core stable?
}
We define \FEN-\sverif\ and \FE-\sverif\ accordingly when we instead ask whether $\Pi$ is strictly core stable. 
All four problems are contained in \conp\ since checking whether a coalition is blocking a coalition structure can be done in polynomial time. 

By definition, it is straightforward that verifying Nash stability or individual stability is polynomially solvable. Hence, we look into the existence questions, which are contained in NP. 
\decprob{\FEN-\NS\ ({\normalfont \text{resp.}} \FEN-\IndS)}{
 An \FEN\ instance~$I=(V, \goodG, \badG)$. %
}{
Does $I$ admits a Nash stable (resp.\ individually stable) coalition structure?
}
We define \FE-\NS\ accordingly for the \FE\ case. %
We assume basic knowledge of parameterized complexity and refer to the following textbooks~\cite{Nie06,CyFoKoLoMaPiPiSa2015} for more details.
\smallskip

\noindent \textbf{Graph structures and parameters.}
We investigate the (parameterized) complexity of the above problems and focus on restricted instances.
Given an instance~$I=(V, \goodG, \badG)$, we define the following parameters:
\begin{compactitem}[--]
  \item Max degree~$\maxdeg$: For \FEN, it is defined as $\max_{i\in V} |N^+_{\goodG+\badG}(i)\cup N^-_{\goodG+\badG}(i)|$, while for \FE, it is defined as $\max_{i\in V} |N^+_{\goodG}(i)\cup N^-_{\goodG}(i)|$ since $\goodG+\badG$ is a complete digraph.
  \item Max coalition size~$\maxcoal$: It is defined as the size of the largest coalition in the coalition structure from the input.
  \item Feedback arc set number~$\fas$: For \FEN, $\fas$ is the smallest number of arcs deleting which makes $\goodG+\badG$ acyclic, while for \FE, $\fas$ is the smallest number of arcs deleting which makes $\goodG$ acyclic,
\end{compactitem}
We say that $I$ has \myemph{symmetric preferences} if each arc in $\goodG$ and $\badG$ is bi-directional (see the arcs $(1,2)$ and $(2,1)$ in \cref{ex:prelim}).
It contains \myemph{acyclic} graph (\myemph{DAG}) if  the union $\goodG+\badG$ is acyclic for the \FEN\ model and $\goodG$ is acyclic for the \FE\ model, respectively.

\section{The Friends and Enemies Model}\label{sec:fe}
\appendixsection{sec:fe}
In this section, we focus on the \FE\ model~\cite{dimitrov2006simple}.
First of all, we settle the complexity of problems regarding (strict) core verification and Nash existence, and show that they are intractable and remain so even for very restricted cases such as sparse graphs and symmetric preferences.
For the hardness reductions, we use the following NP-complete problem: %
\decprob{\pxct}
{A $3\enn$-element set~$\mathcal{X}=[3\enn]$ and a collection~$\mathcal{C}=(C_1,\ldots,$ $C_{\emm})$ of $3$-element subsets of~$X$ such that each element~$i\in X$ appears in either two or three members in~$\mathcal{C}$ and that the associated \emph{element-linked graph} is planar.}{Does~$\mathcal{C}$ contain an \myemph{exact cover} for~$X$, i.e., a subcollection~$\mathcal{K} \subseteq \mathcal{C}$ such that each element of~$X$ occurs in exactly one member of~$\mathcal{K}$?}
Herein, given a \pxct{} instance~$I=(X,\mathcal{C})$, the associated \myemph{element-linked graph} of~$I$ is a graph~\myemph{$G(I)=(U \uplus W, E)$} on two partite vertex sets~$U=\{u_i\mid i \in X\}$ and $W=\{w_j \mid C_j\in \mathcal{C}\}$ such that
$E=\{\{u_i,w_j\}\mid i \in C_j\}\cup \{\{u_i,u_{i+1}\} \mid i\in [3\enn-1]\}$ $\cup \{\{u_1,u_{3\enn}\}\}$ is planar.
We call the vertices in~$U$ and $W$ the \myemph{element-vertices} and the \myemph{set-vertices}, respectively.
We also call the cycle induced by the element-vertices the \myemph{element-cycle}.
For notational convenience, for each element~$i\in [3\enn]$, let \myemph{$\mathcal{C}(i)$} $\coloneqq \{C_j\in \mathcal{C} \mid i\in C_j\}$ denote the sets which contain element~$i$.

Dyer and Frieze~\cite{DF86} show that the NP-completeness of \pxct\ remains even if the planar embedding of the associated element-linked graph satisfies the following: %
\begin{align*}
&  \text{for all } u_i\in U \text{ there are at most two vertices } w_j \text{ such that }
 i\in C_j\\
 & \text{ and both lie \emph{inside} or \emph{outside} of the element-cycle}.\tag{$\heartsuit$}\label{cond:at-most-two-3sets}
\end{align*}
Hence, for notational convenience and based on this planar embedding, we partition~$\mathcal{C}$ into two disjoint subcollections: 
\myemph{$\mathcal{C}^{\textsf{out}}$} $\coloneqq \{C_j\in \mathcal{C}\mid w_j$  lies outside of the element-cycle  $\}$ and
\myemph{$\mathcal{C}^{\textsf{in}}$} $\coloneqq \mathcal{C}\setminus \mathcal{C}^{\textsf{out}}$.
In the hardness proofs, we will utilize this fact to construct appropriate gadgets which do not exceed the desired maximum vertex degree.

\newcommand{\clmcoreplanardegfourforward}{%
  If $I$ admits an exact cover, then $\Pi$ is not core stable (and hence not strictly core stable).%
}
   
\newcommand{\clmcoreplanardegfourbackward}{%
  If $\Pi$ admits a weakly blocking coalition, then $I$ admits an exact cover.
}

\newcommand{\thmcorefour}{%
  \FE-\verif\ and \FE-\sverif\ are \conp-complete even for planar friendship graphs and $\maxdeg=4$.
}
\begin{theorem}\label{thm:verif-non-symmetric-four}
  \thmcorefour
\end{theorem}

\begin{proof}[Proof sketch]%
  We show hardness for both problems via the same reduction. %
  Let $I=([3\enn], \mathcal{C})$ denote an instance of \pxct\ with
  $\mathcal{C}=\{C_1,\ldots,C_{\emm}\}$, 
  and let $G(I)=(U\cup W, E)$ denote the associated element-linked planar graph.
  Recall that there exists a planar embedding of $G(I)$ which satisfies \eqref{cond:at-most-two-3sets} such that $\mathcal{C}^{\textsf{out}}$ and $\mathcal{C}^{\textsf{in}}$ partition the set family~$\mathcal{C}$ into two disjoint subfamilies according to this embedding.
  
  For brevity's sake, define \myemph{$L=27\enn-1$}; we note that the desired blocking coalition will be of size $L+1$.
  For each element~$i\in [3\enn]$, create three \myemph{element agents}~$x_i, s_i, t_i$,
  and a set of~$L+1$ private \myemph{friendship} agents~$x_i^z$, $z\in \{0,\ldots,L\}$; these $L+1$ agents comprise the friendship gadget of element~$i$.
  The number of the friendship agents will ensure that we indeed have an exact cover.
  For each set~$C_j\in \mathcal{C}$ and each element~$i\in C_j$, create two \myemph{set agents}~$c_j^i$ and $d_j^i$.
  To connect the elements with the sets, for each element~$i\in [3\enn]$, create two groups of connection gadgets (one for each side of the element-cycle) with a total of eight agents called~$a^z_i, b_i^z$, $z\in \{0,1,2,3\}$,
  which serve as selector agents.
  We remark that agents~$a_i^3$ and $b_i^3$ serve as connectors and will be friends with the agents corresponding to the sets which contain~$i$.
  This completes the construction of the agents.
  In total, we have
  agent set~$V\coloneqq  \{x_i,x_i^0,\ldots,x_i^{L},s_i,t_i,a_i^z,b_i^z \mid i\in [3\enn], z\in \{0,1,2,3\}\} \cup \{c^i_j,d^i_j,c^k_j,d^k_j,c^{r}_j,d^{r}_j\mid C_j\in \mathcal{C} \text{ with } C_j=\{i,k,r\} \}$.
  Next, we describe the friendship graph~$\goodG$; its planar embedding is based on the planar embedding of~$G(I)$.
 
  \begin{figure}[t!]
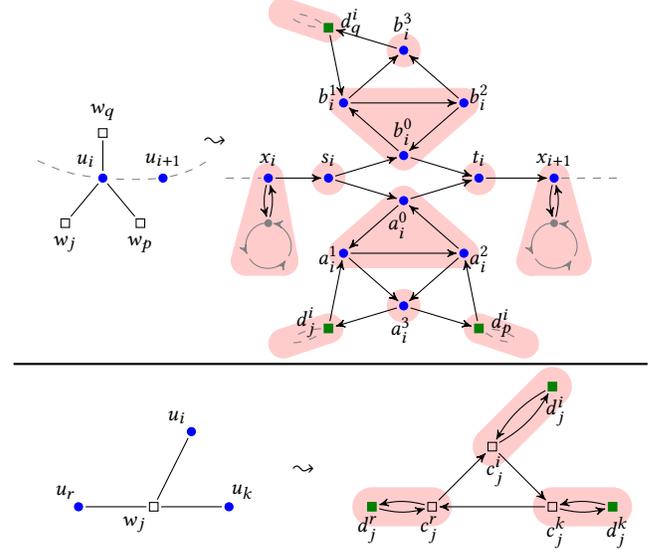

    \centering
  \begin{tikzpicture}[scale=1,every node/.style={scale=0.9},>=stealth', shorten <= 1pt, shorten >= 1pt]
  \input{tikz-hypergraph}%
  \elementpic
  \node at (1.5,0.5) {$\leadsto$};
  \begin{scope}[shift={(2,0)}]
    \node[nn] at (0.2,0) (xi) {};
    \node[pn] at (0.2,-.6) (pxi) {};
    \node[nn] at (1,0) (si) {};
    \node[nn] at (3,0) (ti) {};
    \node[nn] at (4,0) (xip) {};
    \node[pn] at (4,-0.6) (pxip) {};

    \node[] at (-.5,0) (uim) {};
    \node[] at (5,0) (uipp) {};

    \node[] at (-.1,-1.1) (uxi1) {};
    \node[] at (.4,-1.1) (uxi2) {};
    \node[] at (3.7,-1.1) (uxi1p) {};
    \node[] at (4.3,-1.1) (uxi2p) {};

    \node[] at (0.4,-2.2) (uji) {};
    \node[] at (3.6,-2.2) (uri) {};
    \node[] at (0.4,2.2) (uki) {};

    \foreach \x / \y / \n / \typ in {
      2/-.3/ai0/nn,
      1.2/-1/ai1/nn, 2.8/-1/ai2/nn,
      2/.3/bi0/nn,
      2/-1.7/ai3/nn,
      1/-2/dji/ssn,
      3/-2/dri/ssn,
      1.2/1/bi1/nn, 2.8/1/bi2/nn,
      2/1.7/bi3/nn,
      1/2/dki/ssn%
    } {
      \node[\typ] at (\x, \y) (\n) {}; 
    }
    
    \draw[gl,dashed] (uim) -- (xi);
    \draw[gl,dashed] (xip) -- (uipp);
    \draw[gl,dashed] (dji)  to[bend right=15] (uji);
    \draw[gl,dashed] (uji)  to[bend right=15] (dji);
    \draw[gl,dashed] (dki)  to[bend right=15] (uki);
    \draw[gl,dashed] (uki)  to[bend right=15] (dki);
    \draw[gl,dashed] (dri)  to[bend right=15] (uri);
    \draw[gl,dashed] (uri)  to[bend right=15] (dri);
    
    \node[above = 0pt of xi] {$x_i$};
    \node[above = 0pt of xip] {$x_{i+1}$};
    \node[above = 0pt of si] {$s_{i}$};
    \node[above = 0pt of ti] {$t_{i}$};

    \foreach \s / \t in {xi/si, si/ai0, si/bi0, ai0/ti, bi0/ti, ai1/ai2, ti/xip, ai1/ai3,ai2/ai3,ai0/ai1,ai2/ai0,ai3/dji,ai3/dri,dji/ai1,dri/ai2,
      bi0/bi1,bi1/bi2,bi2/bi0,bi1/bi3,bi2/bi3,bi3/dki,dki/bi1} {
       \draw[->] (\s) -- (\t);
     }

     \foreach \s / \t / \d in {pxi/xi/15,xip/pxip/15} {
       \draw[->] (\s) to[bend right=\d] (\t);
       \draw[->] (\t) to[bend right=\d] (\s);
     }

     \foreach \i in {pxi,pxip} {
       \ARW[gray,->]{90:208,208:326,326:444}{.3}{\i}
     }
   
     \foreach \i / \p / \a / \r / \n in {ai0/below right/0pt/-10pt/{a_i^0},
       ai1/below left/-7pt/-2pt/{a_i^1},
       ai2/below right/-7pt/-2pt/{a_i^2},
       ai3/below right/0pt/-9pt/{a_i^3},
       dji/above left/-8pt/0pt/{d_j^i},
       dri/above right/-8pt/0pt/{d_{\setind}^i},
       bi1/above left/-7pt/-2pt/{b_i^1},
       bi2/above right/-7pt/-2pt/{b_i^2},
       bi0/above right/0pt/-8pt/{b_i^0},  bi3/above right/0pt/-8pt/{b_i^3},
       dki/above right/-8pt/0pt/{d_{q}^i},} {
       \node[\p = \a and \r of \i] {$\n$};
     }

     \begin{pgfonlayer}{bg}
      \node[ip, shape=circle, minimum size=5mm] at (si) {};
      \node[ip, shape=circle, minimum size=5mm] at (ti) {};
      \node[ip, shape=circle, minimum size=5mm] at (ai3) {};
      \node[ip, shape=circle, minimum size=5mm] at (bi3) {};
      \draw[ip] \hedgeiii{ai1}{ai0}{ai2}{2mm};
      \draw[ip] \hedgeiii{bi2}{bi0}{bi1}{2mm};
      \draw[ip] \hedgem{xi}{pxi}{uxi2, uxi1}{2mm};
      \draw[ip] \hedgem{uxi1p}{pxip}{xip, uxi2p}{2mm};
      \draw[ip] \hedgeii{uxi2}{xi}{2mm};
      \draw[ip] \hedgeii{uxi1p}{xip}{2mm};
      \draw[ip] \hedgeii{uji}{dji}{2mm};
      \draw[ip] \hedgeii{uri}{dri}{2mm};
      \draw[ip] \hedgeii{uki}{dki}{2mm};
     \end{pgfonlayer}
  \end{scope}
\end{tikzpicture}
\begin{tikzpicture}[scale=1,every node/.style={scale=0.9}, >=stealth', shorten <= 1pt, shorten >= 1pt]
  \input{tikz-hypergraph}
  \draw (-1.9,1.9) edge[thick] (6.6,1.9);

  \setpic

   \node at (2,0.5) {$\leadsto$};
 \begin{scope}[shift={(4.5,0)}]

   \foreach \x / \y / \n / \typ in {
      0/.8/cji/sn, .8/0/cjk/sn, -.8/0/cjr/sn, .8/1.6/dji/ssn, 1.6/0/djk/ssn,-1.6/0/djr/ssn%
    } {
      \node[\typ] at (\x,\y) (\n) {};
    }
    \foreach \i / \p / \a / \r / \n in {
      cji/below right/-1pt/-7pt/{c_j^i},
      dji/below right/-1pt/-7pt/{d_j^i},
      cjk/below right/-1pt/-7pt/{c_j^k},
      djk/below right/-1pt/-7pt/{d_j^k},
      cjr/below left/-1pt/-7pt/{c_j^r},
      djr/below left/-1pt/-7pt/{d_j^r}%
    }{
      \node[\p = \a and \r of \i] {$\n$};
    }
    \foreach \s in {i,k,r} {
      \draw[->] (cj\s) to[bend right=15] (dj\s);
      \draw[->] (dj\s) to[bend right=15] (cj\s);
    }

    \foreach \s/\t in {i/k,k/r,r/i} { 
      \draw[->] (cj\s) -- (cj\t);
    }

     \begin{pgfonlayer}{bg}
      \draw[ip] \hedgeii{cji}{dji}{2.6mm};
      \draw[ip] \hedgeii{cjk}{djk}{2.6mm};
      \draw[ip] \hedgeii{cjr}{djr}{2.6mm};
     \end{pgfonlayer}
     
  \end{scope}
\end{tikzpicture}
\caption{Gadgets for  \cref{thm:verif-non-symmetric-four}. The red areas indicate the coalitions in the initial partition~$\Pi$.
  Upper part: An element-gadget, where $\mathcal{C}(i)=\{C_j,C_{\setind}, C_q\}$ such that $C_j,C_p\in \mathcal{C}^{\textsf{out}}$ (indicated by the dashed line).
  The private friendship gadget for~$x_i$ (resp.\ $x_{i+1}$) is depicted as a gray directed cycle. 
  Lower part: A set-gadget, where $C_j=\{i,k,r\}$. The red areas indicate the coalitions in the initial partition~$\Pi$.}\label{fig:deg4-element-set}
\end{figure}
\begin{itemize}[--]
    \item %
    Starting from the planar embedding of $G(I)$, we replace each element vertex~$u_i\in U$ with the corresponding element agent~$x_i$ and their friendship agents~$x_i^z$, $z\in \{0,\ldots,L\}$, with bidirectional arcs $(x_i,x_i^0)$, $(x_i^0,x_i)$ and directed cycle~$(x_i^z,x_i^{z+1})$ ($z\in \{0,\ldots,L\}$, $z+1$ taken modulo $L+1$).    
    For each edge~$\{u_i,u_{i+1}\}$  on the element-cycle,
    we replace it with the arcs~$(x_i,s_i)$, $(s_i,a_i^0)$, $(s_i,b_i^0)$, $(a_i^0,t_i)$, $(b_i^0, t_i)$, $(t_i,x_{i+1})$; let ${i+1}=1$ if $i=3\enn$. 
    Further, for each $(i,z)\in [3\enn]\times \{0,1,2\}$,
    we add the arcs $(a_i^z,a_i^{z+1})$, $(b_i^z,b_i^{z+1})$ ($z+1$ taken modulo $3$),
    $(a_i^1,a_i^3)$, $(a_i^2,a_i^3)$, and
    $(b_i^1,b_i^3)$, $(b_i^2,b_i^3)$.

    \item For each set~$C_j \in \mathcal{C}$ with $C_j=\{i,k,r\}$ and $i<k<r$, we replace the corresponding set vertex in the planar embedding with a directed subgraph consisting of the following arcs:
    $(c_j^i,d_j^i)$, 
    $(d_j^i,c_j^i)$, 
    $(c_j^k,d_j^k)$, 
    $(d_j^k,c_j^k)$, 
    $(c_j^{r},d_j^{r})$, 
    $(d_j^{{r}},c_j^{r})$,
    $(c_j^{i}, c_j^{k})$,
    $(c_j^{k}, c_j^{{r}})$,
    $(c_j^{{r}}, c_j^{i})$.

    \item For each edge~$\{u_i,w_j\}\in E(G(I))$, we replace this edge according to where the corresponding set vertex lies in the planar embedding.
    If $C_j \in \mathcal{C}^{\textsf{out}}$, then we add the arcs~$(a_i^3, d_j^i)$, $(d_j^i,a_i^{z})$, where $z$ is deterministically fixed to either $1$ or $2$ so as to maintain the planarity.
    Analogously, if $C_j\in \mathcal{C}^{\textsf{in}}$, then we add arcs~$(b_i^3, d_j^i)$, $(d_j^i,b_i^{z})$; again $z$ is deterministically fixed to either $1$ or $2$ so as to maintain the planarity.
    Note that by \eqref{cond:at-most-two-3sets},
    adding arcs $(d_j^i,v_i^z)$ ($v\in \{a,b\}$, $z\in [2]$) preserves planarity.
  \end{itemize}

 \noindent See \cref{fig:deg4-element-set} for an illustration.
  To complete the construction, define the initial coalition structure
    $\Pi\coloneqq
     \{\{x_i,x_i^0,x_i^1,\ldots,x^{L}\}, \{s_i\}, \{t_i\}\mid i\in [3\enn]\}\cup
      \{\{v_i^0,v_i^1,v_i^2\}, \{v_i^3\} \mid (i,v)\in [3\enn]\times \{a,b\}\}\cup
      \{\{c_j^i,d_j^i\}\mid \{u_i,w_j\}\in E(G(I))\}$.
The general idea for the reduction is as follows: For each element~$i\in [3\enn]$ we created an element-gadget (consisting of several element agents). The element-gadgets are ``connected'' via appropriate selector gadgets so they correspond to the element-cycle in the associated element-linked graph~$G(I)$.
This ensures that a weakly blocking coalition will need to contain dedicated agents which correspond to all elements. 
For each set~$C_j\in \mathcal{C}$ we created a set-gadget (again consisting of several set agents).
We connected an element-gadget to a set-gadget via a communication gadget if and only if the corresponding element is contained in the corresponding set. 
This ensures that an agent in a set-gadget is in a blocking coalition if and only if the element agents ``contained'' in the set are in the blocking coalition as well.
Note that this already gives us a covering for the elements.
To have an exact cover, we introduced a novel friendship gadget which can never participate in any blocking coalition but shall ensure that its associated agent will never join a blocking coalition that is too large. This sets an upper limit on the size of a blocking coalition.

      \appendixalg{thm:verif-non-symmetric-four}{One can verify that the graph is planar.
        \ifshort
        The proofs for having maximum degree~$4$ and the correctness are deferred to the full version.
        \else
        The proofs for having maximum degree~$4$ and the correctness are deferred to the appendix.        
        \fi}{\thmcorefour}{ %
    \smallskip

    \noindent To see why $\maxdeg=4$, we observe the following.
  \begin{observation}\label{obs:prefs}
    The agents have the following in- and out-neighbors in~$\goodG$. 
    \begin{compactenum}[(i)]
      \item\label{obs:Prefs-elem1} For each~$i\in [3\enn]$, it holds that
      $N^+(x_i)=\{x_i, t_{i-1}\}$ (let $t_{i-1}=t_{3\enn}$ if $i=1$),
      $N^-(x_i)=\{x_i^0, s_i\}$,
      $N^+(x^0_i)=\{x_i, x_i^1\}$,
      $N^-(x^0_i)=\{x_i, x_i^{L}\}$,
      and for each $z\in [L]$ it holds that $N^+(x^z_i)=\{x_i^{z+1}\}$ and
      $N^-(x^z_i) = \{x_i^{z-1}\}$ (where $z+1$ and $z-1$ are taken modulo $L+1$).
      \item\label{obs:Prefs-elem2} For each~$i\in [3\enn]$, it holds that
      $N^+(s_i)=N^-(t_i)=\{a_i^0,b_i^0\}$,
      $N^-(s_i)=\{x_i\}$,
      $N^+(t_i)=\{x_{i+1}\}$ (let $x_{i+1}=x_1$ if $i=3\enn$).
      \item\label{obs:Prefs-connect}
      For each $(i,v,z)\in [3\enn]\times \{a,b\}\times \{1,2\}$, the following holds.
      $N^+(v^0_i) =\{t_i, v_i^1\}$, 
      $N^-(v^0_i)=\{s_i, v_i^2\}$
      $N^+(v^z_i)=\{v_i^{z+1}, v_i^{z+2}\}$ (where $z+2$ is taken modulo $4$),
      $N^-(v^3_i)=\{v_i^1,v_i^2\}$.
      $N^-(v^z_i)$ consists of agent~$v^{z-1}_i$ and at most one agent from $\{d_j^i\mid C_j\in \mathcal{C}\wedge i\in C_j\}$.
      $N^+(v_i^3)$ consists of at most two friends from  $\{d_j^i\mid C_j\in \mathcal{C}\wedge i\in C_j\}$.
      \item\label{obs:Prefs-set} For each $C_j\in \mathcal{C}$ and each $i\in C_j$, the following holds.
      $N^+(c_j^i)$ consists of agent~$d_j^i$ and an agent from $\{c_j^{k}\mid k\in C_j\wedge k\neq i\}$, and 
      $N^-(c_j^i)$ consists of agent~$d_j^i$ and another agent from $\{c_j^{k}\mid k\in C_j\wedge k\neq i\}$.
      $N^+(d_j^i)$ consists of agent~$c_j^i$ and an agent from $\{a_i^1, a_i^2, b_i^1, b_i^2\}$,
      and $N^-(d_j^i)$ consists of agent~$c_j^i$ and an agent from $\{a_i^3,b_i^3\}$.
    \end{compactenum}
  \end{observation}

  Before we continue with the correctness proof, we observe the following for~$\Pi$.
  \begin{observation}\label{obs:Pi}
    \begin{compactenum}[(i)]
      \item\label{obs:Pi-elem} For each~$i\in [3\enn]$, element agent~$x_i$ has one friend~$x_i^0$ and $L$ enemies~$x_i^z$, $z\in [L]$, in $\Pi(x_i)$,
      while both $s_i$ and $t_i$ have no friends and no enemies in their respective initial coalitions.
      \item\label{obs:Pi-private} For each~$i\in [3\enn]$, agent~$x_i^0$ has two friends $x_i$ and $x_i^1$, and $L-1$ enemies in~$\Pi(x_i^0)$, while each agent~$x_i^z$, $z\in [L]$, has
      exactly one friend and $L$ enemies in~$\Pi(x_i^z)$.
      \item\label{obs:Pi-connect} For each $(i,v,z)\in [3\enn]\times \{a,b\}\times \{0,1,2\}$, agent~$v^z_i$ has one enemy and one friend in $\Pi(v^z_i)$,
      while agent~$v^3_i$ has no friends and no enemies in~$\Pi(v_i^3)$.
      \item\label{obs:Pi-set} For each $i\in [3\enn]$ and $j\in [\emm]$ with $i\in C_j$, both~$c_j^i$ and $d_j^z$ have each one friend and no enemies in their respective initial coalitions.
    \end{compactenum}
  \end{observation}

  It remains to show the correctness of the construction.
  We prove this via the following two claims.

  \begin{clm}\label{claim:planar-deg4-forward}
    \clmcoreplanardegfourforward
  \end{clm}

  {
    \begin{proof}[Proof of \cref{claim:planar-deg4-forward}]
    \renewcommand{\qedsymbol}{$\diamond$}
    Let $\mathcal{K}$ be an exact cover of~$I$.
    Partition $\mathcal{K}$ into two subsets as follows:
     \begin{align*}
       \mathcal{K^{\textsf{out}}}\coloneqq \mathcal{K} \cap \mathcal{C}^{\textsf{out}} \text{ and }
       \mathcal{K^{\textsf{in}}}\coloneqq \mathcal{K}\setminus \mathcal{K}^{\textsf{in}}.
     \end{align*}
    We claim that the following coalition
    \begin{align*}
      V'  \coloneqq
      &\{x_i,s_i,t_i,\mid i\in [3\enn]\}\cup \{c_j^i,d_j^i \mid i\in [3\enn] \wedge C_j\in \mathcal{K}\cap \mathcal{C}(i) \} \cup\\
      & \{a_i^z\mid i \in \bigcup_{C_j\in \mathcal{K}^{\textsf{out}}}{C_j}, z\in \{0,1,2, 3\}\}\cup\\
      & \{b_i^z\mid
        i \in \bigcup_{C_j\in \mathcal{K}^{\textsf{in}}}{C_j},z\in \{0,1,2, 3\}\}.
    \end{align*}
    is blocking~$\Pi$. Note that $|V'|=27\enn=L+1$.
    \begin{itemize}[--]
      \item For each $i\in [3\enn]$, 
      element agent~$x_i$ has exactly one friend and $|V'|-2=L-1$ enemies in~$V'$,
      agents~$s_i$ and $t_i$ each have one friend in $V'$.
      Hence, by \cref{obs:Pi}\eqref{obs:Pi-elem}, we know that $x_i,s_i,t_i$ all prefer $V'$ to their coalition under~$\Pi$.
      \item For each $i \in\bigcup_{C_j\in \mathcal{K}^{\textsf{out}}}{C_j}$,
      agent~$a_i^0$ has two friends in~$V'$, namely $a_i^1$ and $t_i$,
      agent~$a_i^1$ has two friends in~$V'$, namely $a_i^2$ and $a_i^3$,
      agent~$a_i^2$ has two friends in~$V'$, namely $a_i^0$ and $a_i^3$,
      while agent~$a_i^3$ has at least one friend in~$V'$, namely $d_j^i$ with $C_j\in \mathcal{K}^{\textsf{out}}$.
      By \cref{obs:Pi}\eqref{obs:Pi-connect}, each agent~$a_i^z$, $i \in \bigcup_{C_j\in \mathcal{K}^{\textsf{out}}}{C_j}$ and $z\in \{0,1,2,3\}$ strictly prefers $V'$ to her initial coalition~$\Pi(a_i^z)$.
      \item Analogously, we infer the same for $b_i^z$, $i \in \bigcup_{C_j\in \mathcal{K}^{\textsf{in}}}{C_j}$ and $z\in \{0,1,2,3\}$.
      \item For each~$C_j\in \mathcal{K}$ and each $i\in C_j$, agent~$c_j^i$ has two friends in~$V'$: $d_j^i$ and one agent from $\{c_j^{i'}\mid i'\in C_j\setminus \{i\}\}$,
      and agent~$d_j^i$ also has two agents in~$V'$, where one is $c_j^i$ and the other is either $a_i^z$ or $b_i^z$ ($z$ is fixed to either $1$ or $2$), depending on whether $C_j\in \mathcal{K}^{\textsf{out}}$.
      By \cref{obs:Pi}\eqref{obs:Pi-set}, we obtain that both $c_j^i$ and $d_j^i$, $C_j\in \mathcal{K}$ strictly prefer $V'$ to their respective coalition under~$\Pi$.
    \end{itemize}
     Summarizing, we show that $V'$ indeed strictly blocks~$\Pi$.
  \end{proof}
  }

  \begin{clm}%
    \label{claim:planar-deg4-backward}
    \clmcoreplanardegfourbackward
  \end{clm}
  {\begin{proof}[Proof of \cref{claim:planar-deg4-backward}]
    \renewcommand{\qedsymbol}{$\diamond$}
    Let $V'$ be a weakly blocking coalition of~$\Pi$.
    We claim that $\mathcal{K}\coloneqq \{C_j\mid c_j^i \in V'\text{ for some }i \in C_j\}$ is an exact cover for~$I$.
    We prove this in three steps.

    \smallskip
    \noindent \textbf{Step 1.} 
    We aim to show that if $V'$ contains an element agent, then $\mathcal{K}$ is a set cover.
    First, we observe that each agent~$x_i^z$, $(i,z)\in [3\enn]\times \{0,\ldots,L\}$,
    has the maximum number of friends that she can get from her initial coalition.
    If it would hold that $x_i^z\in V'$ for some $(i,z)\in [3\enn]\times \{0,\ldots,L\}$,
    then by \cref{obs:Pi}\eqref{obs:Pi-private} and \cref{obs:prefs}\eqref{obs:Prefs-elem1},
    all agents from $\{x_i,x_i^0,\ldots,x_i^L \}$ must also be in~$V'$, implying
    that $\Pi(x_i^0)\subseteq V'$.
    However, since $x_i^0$ has in total only two friends and $V'\neq \Pi(x_i^0)$, she will be worse off in a larger coalition for which she has the same number of friends.
    Hence,  the following holds:
    \begin{align}
      \{x_i^z\mid (i,z)\in [3\enn]\times \{0,\ldots,L\}\} \cap V'=\emptyset.
      \label{eq:no-xi^0}
    \end{align}
    Next, consider an arbitrary~$i\in [3\enn]$.
    By the above and by \cref{obs:prefs}\eqref{obs:Prefs-elem1},
    it holds that
    \begin{align}
      \text{if } x_i \in V', \text{ then } s_i\in V', \label{eq:xi->si}
    \end{align}
    since $x_i$ has one friend in $\Pi(x_i)$  (see \cref{obs:Pi}\eqref{obs:Pi-elem}).
    Further, since $\Pi(t_i)$ is a singleton, by \cref{obs:prefs}\eqref{obs:Prefs-elem2}, it holds that  
     \begin{align}
      \text{if } t_i \in V', \text{ then } x_{i+1}\in V', \label{eq:ti->xi+1}
    \end{align}
    where we let $x_{i+1}=x_1$ if $i=3\enn$.
    Similarly, it holds that  
     \begin{align}
      \text{if } s_i \in V', \text{ then } \{a_{i}^0,b^0_i\}\cap V' \neq \emptyset, \label{eq:si->aiorbi}
    \end{align}
    Since $a_i^0$ (resp.\ $b_i^0$) initially has one friend and one enemy but there is not two-cycle which contains $a_i^0$, it follows that if $a_i^0\in V'$
    (resp.\ $b_i^0\in V'$), then $V'$ must contain two friends for her.
    By \cref{obs:prefs}\eqref{obs:Prefs-connect}, we infer that for each $v\in \{a,b\}$
     \begin{align}
       \text{if } v_i^0 \in V', \text{ then } \{t_i,v^1_i\}\subseteq V'.
       \label{eq:aiorbi->tiaibi}
     \end{align}
     Similarly, we obtain that
     \begin{multline}\label{eq:all-y}
       \text{ for all } (i,v)\in [3\enn]\times \{a,b\},\\
       \text{ if } \{v_i^0,v_i^1,v_i^2\}\cap V'\neq \emptyset, 
      \text{ then } \{t_i,v_i^0,v_i^1,v_i^2,v_i^3\}\subseteq V'.
    \end{multline}
    By applying \eqref{eq:xi->si}--\eqref{eq:all-y} repeatedly,
    we also obtain that 
    \begin{align}\label{eq:all-elem}
      \nonumber & \text{if } \{x_i,s_i,t_i \mid i\in [3\enn]\}\cap V' \neq \emptyset, \text{ then for all } i \in [3\enn] \text{ it holds that }\\
      &~~ \{x_i,s_i,t_i,a_i^0, a_i^1,a_i^2,a_i^3\} \subseteq V'
       \text{ or } \{x_i,s_i,t_i,b_i^0,b_i^1,b_i^2,b_i^3\}\subseteq V'.
    \end{align}    
    By the above and by the preferences of $a^3_i$ and $b^3_i$,
    we have that if $a^3_i\in V'$ (resp.\ $b^3_i\in V'$), then there must be a set~$C_j\in \mathcal{C}^{\textsf{out}}$ (resp.\ $C_j\in \mathcal{C}^{\textsf{in}}$) with $i\in C_j$ such that
    $d_j^i\in V'$.
    Let $d_j^i$ be such an agent in $V'$.
    Since $d_j^i$ initially has exactly one friend and no enemy (see \cref{obs:Pi}\eqref{obs:Pi-set}),
    for $d_j^i$ to be in a blocking coalition,
    she must have at least two friends, since there are no other two-cycle that
    contains $d_j^i$ and she has at least one enemy in~$V'$, namely, $a_i^3$ (resp.\ $b_i^3$).
    This shows that
    \begin{multline}
      \text{for all } (i,v)\in [3\enn]\times \{a,b\},
      \text{if } v_i^3 \in V',\\
       \text{~~ then there exists a } C_j \text{ with } i \in C_j \text{ such that }
        \{c_j^i,d_j^i\} \subseteq V'. %
        \label{eq:set-cover}
    \end{multline}
  Together with \eqref{eq:all-elem}, this implies that if $V'$ contains at least one element agent, then $\mathcal{K}$ is a set cover.

    \smallskip
    \noindent\textbf{Step 2.}
    Next, we show that if $V'$ contains at least one element agent, then $\mathcal{K}$ is an exact cover.
    To this end, assume that $x_i\in V'$ for some $i\in [\enn]$.
    Since $|C_j|=3$ for all $j\in [\emm]$ and there are $3\enn$ elements, we first infer by Step 1 that $|\mathcal{K}|\ge \enn$.
    Let $p_j^i\in V'$ be an arbitrary set agent for some $p\in \{c,d\}$.
    Observe that agent $p_j^i$ initially has exactly one friend but no enemy, so she needs at least two friends in $V'$ to be blocking since she has at least one enemy in $V'$, namely the element agent.
    By construction (see \cref{obs:prefs}\eqref{obs:Prefs-set}), we infer that all friends in~$N^+(p_j^i)$ have to be in~$V'$, including~$d_j^i$.
    Similarly, $N^+(d_j^i)\subseteq V'$, i.e., $\{c_j^i,a_i^1\}\subseteq V'$ if $C_j\in \mathcal{C}^{\textsf{out}}$, and $\{c_j^i,a_i^1\}\subseteq V'$ otherwise.
    Consequently, $c_j^k\in V'$ holds for all $k\in C_j$.
    Using the above reasoning repeatedly, we obtain that 
    \begin{align}
      \nonumber
      &\text{if } \{c_j^i,d_j^i\} \cap V' \neq \emptyset \text{ for some } i \in C_j, 
         \text{ then for all } k \in C_j,\\
       & \text{it holds that } \{c_j^{k}, d_j^{k}\}\subseteq V' \text{ and }
          \{a_k^1,a_k^2,b_k^1,b_k^2\}\cap V'\neq \emptyset.
        \label{eq:all-in-or-all-out}
    \end{align}
    Now, suppose, for the sake of contradiction, that $V'$ contains some element agent~$x_i$ and $|\mathcal{K}| > \enn$.
    Then, by \eqref{eq:all-elem}, \eqref{eq:set-cover}, and \eqref{eq:all-in-or-all-out},
    we infer that $|V'| \ge (3+4)\cdot 3\enn + 2 \cdot 3 \cdot (\enn +1) = 27\enn +6$.
    Then, $x_i$ will have at least $27\enn +6-2=L+3$ enemies in $V'$ and will not weakly prefer $V'$ over her initial coalition~$\Pi(x_i)$ (see \cref{obs:Pi}\eqref{obs:Pi-elem}), a contradiction.

    \smallskip
    \noindent\textbf{Step 3.} Finally, to show that $\mathcal{K}$ is an exact cover, it suffices to show that $V'$ contains an element agent. %
    Suppose, for the sake of contradiction, that $V'$ does not contain any element agent.
    Then, by the contra-positive of \eqref{eq:all-elem},
    $V'\cap \{x_i,s_i,t_i\mid i \in [3\enn]\}=\emptyset$.
    By the contra-positive of \eqref{eq:all-y},
    we infer that $\{v^0_i,v^1_i,v^2_i\mid (i,v)\in [3\enn]\times\{a,b\}\}\cap V'=\emptyset$.
    By the contra-positive of \eqref{eq:all-in-or-all-out},
    we infer that $\{c_j^i,d_j^i\mid j\in [\emm], i\in C_j\}\cap V' = \emptyset$.
    Finally, by the contra-positive of \eqref{eq:set-cover}, we finally infer that
    $\{v_i^3\mid (v,i)\in \{a,b\}\times [3\enn]\}\cap V'=\emptyset$.
    Together with \eqref{eq:no-xi^0}, this implies that $V'=\emptyset$ and will not block~$\Pi$, a contradiction. %
  \end{proof}
  }
  \noindent By \cref{claim:planar-deg4-forward}, if $I$ admits an exact cover, then $\Pi$ is not core stable and hence not strictly core stable.
  If $\Pi$ is not core stable, then it admits a strictly blocking coalition, which is also weakly blocking.
  Hence, by \cref{claim:planar-deg4-backward}, $I$ does not admit an exact cover.
  Together, %
  this proves \cref{thm:verif-non-symmetric-four}.}
\end{proof}

Next, we show that restricting the preferences to be symmetric does not lower the complexity. The ideas of the reduction are fairly similar to that for \cref{thm:verif-non-symmetric-four}.

\newcommand{\thmcoresymmetric}{%
  \FE-\verif\ and \FE-\sverif\ are \conp-com\-plete even if the preferences are symmetric and the friendship graph is planar and $\maxdeg=8$.
}
\begin{theorem}[\appendixsymb]\label{thm:verif-symmetric-eight}
  \thmcoresymmetric
\end{theorem}

\appendixproofwithstatement{thm:verif-symmetric-eight}{\thmcoresymmetric}{
  \begin{proof}%
  Again, we show both hardness via the same reduction.
  Let $I=([3\enn], \mathcal{C})$ denote an instance of \pxct\ with
  $\mathcal{C}=\{C_1,\ldots,C_{\emm}\}$
  and let $G(I)=(U\cup W, E)$ denote the associated element-linked planar graph.
  Recall that there exists an planar embedding of $G(I)$ which satisfies \eqref{cond:at-most-two-3sets} such that $\mathcal{C}^{\textsf{out}}$ and $\mathcal{C}^{\textsf{in}}$ partition~$\mathcal{C}$ into two disjoint subfamilies according to this embedding.

  Again, for brevity's sake, define \myemph{$L=39\enn-2$} and note that the size of the desired blocking coalition will be $L+2$.
  Similarly to the reduction for \cref{thm:verif-non-symmetric-four}, we derive a planar embedding of the preference graph from the planar embedding of~$G(I)$.
  First of all, we introduce the agents.
  For each element~$i\in [3\enn]$, we create six \myemph{element agents}~$x_i, \hx_i,y_i,\hy_i, s_i, t_i$,
  and $2(L+1)$ private \myemph{friendship} agents~$x_i^z$, $y^z_i$, $z\in \{0,\ldots,L\}$.
  For each set~$C_j\in \mathcal{C}$ and each element~$i\in C_j$, we create two \myemph{set agents}~$c_j^i$ and $d_j^i$.
  Further, for each element~$i\in [3\enn]$, we create two groups of connection gadgets (one for each side of the element-cycle) with a total of ten agents called~$a^z_i, b_i^z$, $z\in \{0,\ldots,4\}$,
  which serve as the selector agents.
  We remark that $a_i^4$ and $b_i^4$ serve as connectors and will be friends with the set agents.
  To avoid local and undesired blocking coalitions,
  for each element and each side of the element-cycle,
  we introduce up to two dummy agents on a side if there are less than two sets on that side containing the element.
  Formally, for each element~$i\in [3\enn]$ and each side~$\sstar \in \{{\textsf{out}}, {\textsf{in}}\}$,
  we introduce $2-|\mathcal{C}^{\sstar}\cap \mathcal{C}(i)|$ agents~$\dummys_z$, $z\in [2-|\mathcal{C}^{\sstar}\cap \mathcal{C}(i)|]$.

  \noindent This completes the construction of the agents.
  In total, we have
  \begin{align*}
    V\coloneqq
    & \{x_i,\hx_i,y_i,\hy_i,x_i^z,y_i^z,s_i,t_i\mid
      (i,z) \in [3\enn]\times \{0,\ldots,L\}\} \cup \\
 &    \{a_i^z,b_i^z\mid (i,z)\in [3\enn]\times \{0,\ldots,4\}\} \cup \\
    &  \{c^i_j,d^i_j,c^k_j,d^k_j,c^{r}_j,d^{r}_j\mid C_j \in \mathcal{C} \text{ with } C_j=\{i,k,r\} \} \cup  \myemph{\Dummys},  \text{ where }
  \end{align*}
   \myemph{$\Dummys$} $\coloneqq \{\dummys_z\mid i \in [3\enn], \sstar\in \{\textsf{in},\textsf{out}\}, z\in [2-|\mathcal{C}^{\sstar}\cap \mathcal{C}(i)|]\}$.
  Next, we describe the friendship preferences from the planar embedding of~$G(I)$.
  Since the friendship preferences are symmetric, in the following, we assume that the corresponding friendship graph~$\goodG$ is undirected.

    \begin{figure*}[t!]
    \centering
  \begin{tikzpicture}
  \input{tikz-hypergraph}%
[scale=1,every node/.style={scale=0.9},>=stealth', shorten <= 1pt, shorten >= 1pt]
 \elementpic
 \node at (3,0) {$\leadsto$};
  \begin{scope}[shift={(8,0)}]
    \node[nn] at (0,0.5) (xi) {};
    \node[nn] at (-1,0.5) (hxi) {};
    \node[nn] at (-0.5,1) (xi0) {};
    \node[nn] at (4,0.5) (xip) {};

    \node[nn] at (0,-0.5) (yi) {};
    \node[nn] at (-1,-0.5) (hyi) {};
    \node[nn] at (-0.5,-1) (yi0) {};
    \node[nn] at (4,-0.5) (yip) {};
    
    \node[nn] at (1,0) (si) {};
    \node[nn] at (3,0) (ti) {};
    \node[nn] at (-2,0) (tip) {};

    \node[] at (-3,0.3) (ub) {};
    \node[] at (-3,-0.3) (ua) {};
    \node[] at (5,0.5) (uhx) {};
    \node[] at (4.5,1) (ux0) {};
    \node[] at (5,-0.5) (uhy) {};
    \node[] at (4.5,-1) (uy0) {};
    
    \node[] at (-1,2.1) (uxF1) {};
    \node[] at (0,2.1) (uxF2) {};
    
    \node[] at (-1,-2.1) (uyF1) {};
    \node[] at (0,-2.1) (uyF2) {};

    \foreach \x / \y / \n / \typ in {
      2/-.3/ai0/nn,
      2/-0.9/ai1/nn,
      2.8/-1.5/ai2/nn,
      2/.3/bi0/nn,
      1.2/-1.5/ai3/nn,
      2/-2.1/ai4/nn,
      1/-2.8/dji/ssn,
      3/-2.8/dpi/ssn,
      3/2.8/dum/nnd,
      2/0.9/bi1/nn,
      2.8/1.5/bi2/nn,
      1.2/1.5/bi3/nn,
      2/2.1/bi4/nn,
      1/2.8/dqi/ssn%
    } {
      \node[\typ] at (\x, \y) (\n) {}; 
    }
    
    \draw[gl,dashed] (ua) -- (tip);
    \draw[gl,dashed] (ub) -- (tip);
    \draw[gl,dashed] (uhx) -- (xip);
    \draw[gl,dashed] (ux0) -- (xip);
    \draw[gl,dashed] (uhy) -- (yip);
    \draw[gl,dashed] (uy0) -- (yip);
    
    \node[above = -2pt of xi] {$\hx_i$};
    \node[above = 0pt of hxi] {$x_i$};
    \node[above = 0pt of xi0] {$x^0_i$};
    \node[above left= -4pt of xip] {$x_{i+1}$};
    \node[below = 0pt of yi] {$\hy_i$};
    \node[below = 0pt of hyi] {$y_i$};
    \node[below left = -4pt of yip] {$y_{i+1}$};
    \node[below = -2pt of yi0] {$y^0_i$};
    \node[above = 0pt of si] {$s_{i}$};
    \node[above = 0pt of ti] {$t_{i}$};
    \node[above = 0pt of tip] {$t_{i-1}$};

    \foreach \s / \t in {xi/si, yi/si, si/ai0, si/bi0, ai0/ti, bi0/ti, ai1/ai2, ti/xip, ti/yip, tip/hxi, tip/hyi, ai1/ai3,ai2/ai3,ai0/ai1,ai2/ai4,ai3/ai4,ai4/dji,ai4/dpi,
      bi0/bi1,bi1/bi2,bi2/bi4,bi1/bi3,bi2/bi3,bi4/bi3,bi4/dqi,bi4/dum,
      xi/xi0,xi/hxi,xi0/hxi,yi/yi0,yi/hyi,yi0/hyi} {
       \draw[-] (\s) -- (\t);
     }

     \foreach \i / \d in {xi0/-0.6,yi0/0.6} {
       \ARW[gray,-]{90:208,208:326,326:444}{\d}{\i}
     }
   
     \foreach \i / \p / \a / \r / \n in {ai0/below right/-7pt/0pt/{a_i^0},
       ai1/above left/-9pt/0pt/{a_i^1},
       ai2/below right/-7pt/-2pt/{a_i^3},
       ai3/below left/-7pt/-2pt/{a_i^2},
       ai4/below right/0pt/-9pt/{a_i^4},
       dji/above left/-8pt/0pt/{d_j^i},
       dpi/above right/-8pt/0pt/{d_{p}^i},
       bi1/below left/-9pt/0pt/{b_i^1},
       bi2/above right/-7pt/-2pt/{b_i^3},
       bi0/above right/-7pt/0pt/{b_i^0},
       bi3/above left/-7pt/-2pt/{b_i^2},
       bi4/above right/0pt/-8pt/{b_i^4},
       dum/above right/-8pt/0pt/{\dummyinT},
       dqi/above left/-8pt/0pt/{d_{q}^i},} {
       \node[\p = \a and \r of \i] {$\n$};
     }

      \begin{pgfonlayer}{bg}
      \draw[ip] \hedgem{hxi}{uxF1}{uxF2,xi}{3mm};
      \draw[ip] \hedgem{yi}{uyF2}{uyF1,hyi}{3mm};
      \draw[ip] \hedgem{si}{bi0}{ti,ai0}{2mm};
      \draw[ip] \hedgeiii{bi1}{bi3}{bi2}{2mm};
      \draw[ip] \hedgeiii{ai1}{ai2}{ai3}{2mm};
      \draw[ip] \hedgeiii{dji}{ai4}{dpi}{2mm};
      \draw[ip] \hedgeiii{dum}{bi4}{dqi}{2mm};
    \end{pgfonlayer}
  \end{scope}
\end{tikzpicture}
\caption{Illustration of an element gadget used in the proof for \cref{thm:verif-symmetric-eight},
  where $\mathcal{C}(i)=\{C_j,C_{\setind}, C_q\}$ such that in the planar embedding, $C_j,C_p\in \mathcal{C}^{\textsf{out}}$ (as indicated by the dashed line).
  For brevity's sake, the private friendship gadget for $x^0_i$ (resp.\ $y^0_{i}$) is depicted as a gray circle. The red areas indicate the individual coalitions in the initial partition~$\Pi$.}\label{fig:deg8-elem}
\end{figure*}  \begin{itemize}[--]
     \item  %
      Starting from the planar embedding~$G(I)$, we replace each element vertex~$u_i\in U$ with the corresponding element agents~$x_i$, $\hx_i$, $y_i$, $\hy_i$ and their private friendship agents~$x_i^z$ and $y_i^z$, $z\in \{0,\ldots,L\}$.
      Agents $x_i$, $\hx_i$, and $x_i^0$ (resp.\ $y_i$, $\hy_i$, and $y_i^0$) are mutual friends.
      That is, they form a triangle in $\goodG$.
      The $L+1$ friendship agents~$x_i^z$ (resp.\ $y_i^z$), $z\in \{0,\ldots,L\}$, form a cycle in~$G$.
      
      For each edge~$\{u_i,u_{i+1}\}$ on the element-cycle,
      we replace it with the edges~$\{\hx_i,s_i\}$, $\{\hy_i,s_i\}$, $\{s_i,a_i^0\}$, $\{s_i,b_i^0\}$, $\{a_i^0,t_i\}$, $\{b_i^0, t_i\}$, $\{t_i,x_{i+1}\}$, and $\{t_i,y_{i+1}\}$, where let $i+1=1$ if $i=3\enn$. 
      Further, for each $i\in [3\enn]$ and each $v\in \{a,b\}$, we add edges
      so that $v_i^0$ and $v_i^1$ are mutual friends,
      agents~$v_i^1$, $v_i^2$, and $v_i^3$ form a triangle,
      and agents~$v_i^2$, $v_i^3$, and $v_i^4$ form another triangle.

      Clearly, since $G(I)$ has a planar embedding, the replacement and creation of the new edges so far maintain planarity.
      \item The gadget for the sets are almost the same as the one given in the proof for \cref{thm:verif-non-symmetric-four}, except that we have edges instead of arcs.
      More precisely, for each set~$C_j \in \mathcal{C}$ with $C_j=\{i,k,r\}$,
      we add edges so that $c^i_j,c_j^k,c_j^r$ form a triangle,
      while $c_j^i$ and $d_j^i$ (resp.\ $c_j^k$ and $d_j^k$, and  $c_j^r$ and $d_j^r$) induce an edge.
      \item  For each edge~$\{u_i,w_j\}\in E(G(I))$, we replace this edge according to where the corresponding set vertex lies in the planar embedding.
      If $C_j \in \mathcal{C}^{\textsf{out}}$, then add an edge~$\{d_j^i, a_i^4\}$ to~$\goodG$.
      Analogously, if $C_j\in \mathcal{C}^{\textsf{in}}$, then add an edge~$\{d_j^i, b_i^4\}$. %
      \item Finally, for each~$i\in [3\enn]$, add the remaining edges~$\{\dummyout_z,a_i^4\}$, $z\in [2-|\mathcal{C}^{\textsf{out}}\cap \mathcal{C}(i)|]$,
      and edges~$\{\dummyin_z,a_i^4\}$, $z\in [2-|\mathcal{C}^{\textsf{in}}\cap \mathcal{C}(i)|]$.
      Note that since at most two sets containing the same element are embedded on the same side of the element-cycle, we ensure that after this final step both $a_i^4$ and $b_i^4$ have exactly four mutual friends.
  \end{itemize}

  This completes the construction of the friendship graph~$\goodG$. See \cref{fig:deg8-elem,fig:deg8-set} for an illustration. 
  It is straightforward to verify that the graph is planar.

\begin{figure}[t!]
  \centering
\begin{tikzpicture}[scale=1,every node/.style={scale=0.9}, >=stealth', shorten <= 1pt, shorten >= 1pt]
  \def \xsc {0.6}
  \def \xss {0.8}
\input{tikz-hypergraph}
  \setpic
   
  \node at (2,0.5) {$\leadsto$};
  \begin{scope}[shift={(4.6, 0)}]
    \foreach \x / \y / \n / \typ in {
      0/1/cji/sn, 1/0/cjk/sn, -1/0/cjr/sn, 1/2/dji/ssn, 2/0/djk/ssn,-2/0/djr/ssn%
    } {
      \node[\typ] at (\x*0.6,\y*0.6) (\n) {};
    }
    \node at (1.5*\xss,2.5*\xsc) (uji) {};
    \node at (2.5*\xss,0*\xsc) (ujk) {};
    \node at (-2.5*\xss,0*\xsc) (ujr) {}; 
    
    \foreach \i / \p / \a / \r / \n in {
      cji/below right/-1pt/-7pt/{c_j^i},
      dji/below right/-1pt/-7pt/{d_j^i},
      cjk/below right/-1pt/-7pt/{c_j^k},
      djk/below right/-1pt/-7pt/{d_j^k},
      cjr/below left/-1pt/-7pt/{c_j^r},
      djr/below left/-1pt/-7pt/{d_j^r}%
    }{
      \node[\p = \a and \r of \i] {$\n$};
    }
    \foreach \s in {i,k,r} {
      \draw[-] (cj\s) to (dj\s);
    }

    \foreach \s/\t in {i/k,k/r,r/i} {
      \draw[-] (cj\s) -- (cj\t);
    }
    
    \draw[gl,dashed] (uji) -- (dji);    
    \draw[gl,dashed] (ujk) -- (djk);    
    \draw[gl,dashed] (ujr) -- (djr);    

   \begin{pgfonlayer}{bg}
      \draw[ip] \hedgeiii{cji}{cjk}{cjr}{3mm};
      \draw[ip] \hedgeii{uji}{dji}{3mm};
      \draw[ip] \hedgeii{ujk}{djk}{3mm};
      \draw[ip] \hedgeii{ujr}{djr}{3mm};
    \end{pgfonlayer}

  \end{scope}
  
\end{tikzpicture}
\caption{Illustration of a set gadget used in the proof for \cref{thm:verif-symmetric-eight}, where $C_j=\{i,k,r\}$. The red areas indicate the individual coalitions in the initial partition~$\Pi$.}\label{fig:deg8-set}
\end{figure}
  To see why the graph has maximum vertex degree eight, we observe the following:
  \begin{observation}\label{obs:sym-prefs}
    The agents have the following friends in~$\goodG$.
    \begin{compactenum}[(i)]
      \item\label{obs:sym-Prefs-elem1} For each~$(i,v,z)\in [3\enn] \times \{x,y\}\times [L]$, 
      agent~$v_i$ has three friends, $\hat{v}_i$ and $v_i^0$, and $t_{i-1}$ (let $t_{i-1}=t_{3\enn}$ if $i=1$),
      agent~$\hat{v}_i$ has three friends, $v_i$ and $v_i^0$, and $s_i$,
      agent~$v_i^0$ has four friends, $v_i$, $\hat{v}_i$, $v_i^1$, and $v_i^{L-2}$,
      and agent~$v_i^z$ has two friends, $v_i^{z-1}$ and $v_i^{z+1}$ ($z+1$ taken modulo $L-1$).
      \item \label{obs:sym-Prefs-elem2} For each~$(i,v)\in [3\enn] \times \{a,b\}$, 
      agent~$v_i^0$ has three friends, $s_i$, $t_i$, and $v_i^1$,
      agent~$v_i^1$ has three friends, $v_i^0$, $v_i^2$, and $v_i^3$,
      agent~$v_i^2$ has three friends, $v_i^1$, $v_i^3$, and $v_i^4$,
      agent~$v_i^3$ has three friends, $v_i^1$, $v_i^2$, and $v_i^4$,
      while agent $v_i^4$ has exactly four friends, agents~$v_i^2$ and $v_i^3$, the set agents~$d_j^i$ which are on the same side of the element-cycle,
      and agents~$\dummys_z$, $z\in [2-|\mathcal{C}^{\star}\cap \mathcal{C}(i)|]$,
      and $\star=\textsf{out}$ if $v=a$; $\star=\textsf{in}$ otherwise.
      Each dummy agent from~$D$ has exactly one friend.

      \item\label{obs:sym-Prefs-set} For each $\{u_i,w_j\}\in E(G(I))$, %
      agent $c_j^i$ has three friends, $d_j^i$ and the two other set agents in $\{c_j^{i'}\mid i'\in C_j\setminus \{i\}\}$,
      while agent~$d_j^i$ has two friends, $c_j^i$ and $v_i^4$ (with $v$ being either $a$ or $b$).
    \end{compactenum}
  \end{observation}  
  To complete the construction, we define the initial coalition structure as follows.
  \begin{align*}
    \Pi\coloneqq
    & \{\{x_i,\hx_i,x_i^0,\ldots,x^{L}\}, \{y_i,\hy_i,y_i^0,\ldots,y^{L}\},
     \{s_i,t_i,a_i^0,b_i^0\} \mid i\in [3\enn]\}\cup\\
    &    \{\{v_i^1,v_i^2,v_i^3\} \mid (i,v)\in [3\enn]\times \{a,b\}\} \cup\\
    & \{\{a_i^4,d_j^i,\dummyout_z\mid
      i\in [3\enn] \wedge C_j\in \mathcal{C}^{\textsf{out}}\cap \mathcal{C}(i) \wedge\\
    & ~~~~~~~~\quad\qquad\qquad\qquad z\in [2-|\mathcal{C}^{\textsf{out}}\cap \mathcal{C}(i)|] \} \cup\\
    & \{\{b_i^4,d_j^i,\dummyin_z\mid   i\in [3\enn] \wedge C_j\in \mathcal{C}^{\textsf{in}}\cap \mathcal{C}(i) \wedge\\
    & ~~\quad\qquad\qquad\qquad z\in [2-|\mathcal{C}^{\textsf{in}}\cap \mathcal{C}(i)|] \} \cup\\
    &  \{\{c_j^i,c_j^k,c_j^r \mid C_j\in \mathcal{C} \text{ with } C_j=\{i,k,r\} \}.
  \end{align*}
  Note that the coalition of each agent~$a_i^4$ (resp.\ $b_i^4$), $i\in [3\enn]$, consists of exactly three agents.
  Before we continue with the correctness proof, we observe the following for~$\Pi$.
  \begin{observation}\label{obs:sym-Pi}
    \begin{compactenum}[(i)]
      \item\label{obs:sym-Pi-elem} For each~$(i,v)\in [3\enn] \times \{x,y\}$,
      element agent~$v_i$ has two friends~$v_i^0$ and $\hat{v}_i$, and $L$ enemies~$x_i^z$, $z\in [L]$, in $\Pi(v_i)$,
      element agent~$\hat{v}_i$ has two friends~$v_i^0$ and ${v}_i$, and $L$ enemies~$x_i^z$, $z\in [L]$, in $\Pi(\hat{v}_i)$,      
      while both $s_i$ and $t_i$ have two friends and one enemy in their joint initial coalition.
      \item\label{obs:sym-Pi-private} For each~$(i,v)\in [3\enn] \times \{x,y\}$, 
      private agent~$v_i^0$ has four friends~$v_i$, $\hat{v}_i$, $v_i^1$, and $v_i^{L}$,
      and $L-2$ enemies in~$\Pi(v_i^0)$, while each agent~$v_i^z$, $z\in [L]$, has
      exactly two friends and $L$ enemies in~$\Pi(v_i^z)$.
      \item\label{obs:sym-Pi-connect} For each $(i,v,z)\in [3\enn]\times \{a,b\}\times \{0,1,2\}$,
      agent~$v_i^0$ has two friends ($s_i$ and $t_i$) and one enemy in $\Pi(v_i^0)$,
      agent~$v^z_i$ has two friends and no enemies in $\Pi(v^z_i)$,
      and agent~$v^4_i$ has exactly two friends but no enemies in~$\Pi(v^4_i)$.
      The two friends of~$v^4_i$ can be either two dummy agents, or one dummy agent and one set agent, or two set agents.
      Each dummy agent from~$\Dummys$ has exactly one friend and one enemy in her initial coalition.
      \item\label{obs:sym-Pi-set} For each $i\in [n]$ and $j\in [3\emm]$ with $i\in C_j$, agent $c_j^i$ has $d_j^i$ and one agent from $\{c_j^k\mid k\in C_j\setminus\{i\}\}$
      as friends and no enemies in her initial coalition~$\Pi(c_j^i)$,
      and $d_j^i$ has exactly one friend~$c_j^i$ and one enemy in her initial coalition~$\Pi(d_j^i)$.
    \end{compactenum}
  \end{observation}
  It remains to show that $I$ admits an exact cover if and only if $\Pi$ is \emph{not} core stable (resp.\ strictly core stable).
  We prove this via the following claims.

  \begin{clm}\label{claim:planar-sym-forward}
    If $I$ admits an exact cover, then $\Pi$ is not core stable.
  \end{clm}
  \begin{proof}[Proof of \cref{claim:planar-sym-forward}]
    \renewcommand{\qedsymbol}{$\diamond$}
    Let $\mathcal{K}$ be an exact cover of~$I$.
    Partition $\mathcal{K}$ into two subsets as follows:
     \begin{align*}
       \mathcal{K^{\textsf{out}}}\coloneqq \mathcal{K} \cap \mathcal{C}^{\textsf{out}} \text{ and }
       \mathcal{K^{\textsf{in}}}\coloneqq \mathcal{K}\setminus \mathcal{K}^{\textsf{in}}.
     \end{align*}
    We claim that the following coalition
    \begin{align*}
      V'  \coloneqq
      &\{x_i,\hx_i,y_i,\hy_{i},s_i,t_i,c_j^i,d_j^i \mid i\in [3\enn]\wedge C_j\in \mathcal{K}\cap \mathcal{C}(i) \} \cup\\
      & \{a_i^z\mid i \in \bigcup_{C_j\in \mathcal{K}^{\textsf{out}}}{C_j}, z\in \{0,\ldots,4\}\}\cup \\
      &\{b_i^z\mid
        \bigcup_{C_j\in \mathcal{K}^{\textsf{in}}}{C_j},z\in \{0,\ldots,4\}\}.
    \end{align*}
    is blocking~$\Pi$. Note that $|V'|=39\enn=L+2$.
    \begin{itemize}[--]
      \item For each $(i,v)\in [3\enn]\times \{x,y\}$, 
      both element agents~$v_i$ and $\hat{v}_i$ have exactly two friends and $|V'|-3=L-1$ enemies in~$V'$,
      agents~$s_i$ and $t_i$ each have three friends in $V'$.
      Hence, by \cref{obs:sym-Pi}\eqref{obs:sym-Pi-elem},
      agents~$x_i,s_i,t_i$ all prefer $V'$ to their coalition under~$\Pi$.
      \item For each $i \in\bigcup_{C_j\in \mathcal{K}^{\textsf{out}}}{C_j}$,
      agent~$a_i^0$ has three friends in~$V'$, namely $s_i$, $t_i$, and $a_i^1$,
      agent~$a_i^1$ has three friends in~$V'$, namely $a_i^0$, $a_i^2$, and $a_i^3$,
      agent~$a_i^2$ has three friends in~$V'$, namely $a_i^1$ and $a_i^3$, and $a_i^4$,
      agent~$a_i^3$ has three friends in~$V'$, namely $a_i^1$ and $a_i^2$, and $a_i^4$,
      while agent~$a_i^4$ has at least three friends in~$V'$, namely $a_i^2$, $a_i^3$, and $d_j^i$ with $C_j\in \mathcal{K}^{\textsf{out}}$.
      By \cref{obs:sym-Pi}\eqref{obs:sym-Pi-connect}, we have that each agent~$a_i^z$, $i \in \bigcup_{C_j\in \mathcal{K}^{\textsf{out}}}{C_j}$ and $z\in \{0,\ldots,4\}$, strictly prefers $V'$ to her initial coalition~$\Pi(a_i^z)$.
      \item Likewise, we infer the same for $b_i^z$, $i \in \bigcup_{C_j\in \mathcal{K}^{\textsf{in}}}{C_j}$ and $z\in \{0,\ldots,4\}$.
      \item For each~$C_j\in \mathcal{K}$ and each $i\in C_j$, agent~$c_j^i$ has three friends in~$V'$, namely $\{c_j^{i'}\mid i'\in C_j\setminus \{i'\}\}$ and $d_j^i$,
      and agent $d_j^i$ has two friends in~$V'$: $c_j^i$ and either $a_i^4$ (if $C_j\in \mathcal{K}^{\textsf{out}}$) or $b_i^4$ (if $C_j\in \mathcal{K}^{\textsf{in}}$).
      By \cref{obs:sym-Pi}\eqref{obs:sym-Pi-set}, we obtain that both $c_j^i$ and $d_j^i$, $C_j\in \mathcal{K}$, strictly prefer $V'$ to their respective coalition under~$\Pi$.
    \end{itemize}
    Summarizing, we have shown that $V'$ is indeed blocking~$\Pi$.
    \end{proof}
    Next, we observe several properties that a blocking coalition of $\Pi$ needs to satisfy.
    \begin{clm}\label{claim:blocking}
      Every weakly blocking coalition~$V'$ of $\Pi$ satisfies:
      \begin{compactenum}[(i)]
        \item\label{block-dummy-private} $(\Dummys\cup \{v_i^z\mid (i,z)\in [3\enn]\times \{0,\ldots,L\}\})\cap V'=\emptyset$.
        \item\label{block-elem} If there exists an $i\in [3\enn]$ and a $v\in \{x,y\}$ such that $\{v_{i},\hat{v}_{i}\}\cap V' \neq \emptyset$, then
        $\{s_{i'},t_{i'}\mid i'\in [3\enn]\}\subseteq V'$, and 
        for each $i'\in [3\enn]$ we have that ``$\{x_{i'},\hat{x}_{i'}\}\subseteq V'$ or $\{y_{i'},\hat{y}_{i'}\}\subseteq V'$'' and
        ``$\{a^0_{i'},a_{i'}^1\}\subseteq V'$ or $\{b^0_{i'},b^1_{i'}\}\subseteq V'$''.
        \item\label{block-select} If there exists an $i\in [3\enn]$ such that $\{s_i,t_i\}\cap V' \neq \emptyset$, then
         $\{s_{i'},t_{i'}\mid i'\in [3\enn]\}\subseteq V'$, and 
        for each $i'\in [3\enn]$ we have that ``$\{x_{i'},\hat{x}_{i'}\}\subseteq V'$ or $\{y_{i'},\hat{y}_{i'}\}\subseteq V'$'' and
        ``$\{a^0_{i'},a_{i'}^1\}\subseteq V'$ or $\{b^0_{i'},b^1_{i'}\}\subseteq V'$''.
        \item\label{block-connect} For all $(i,v)\in [3\enn]\times \{a,b\}$, if $\{v_i^0, v_i^1,v_i^2,v_i^3\}\cap V' \neq \emptyset$,
        then $\{v_i^z\mid z\in \{0,\ldots,4\}\}\cup \{s_i,t_i\}\subseteq V'$.
        \item\label{block-setcover} For all $(i,v)\in [3\enn]\times \{a,b\}$, if $v_i^4 \in V'$,
        then $\{v_i^2,v_i^3\}\cap V' \neq \emptyset$ and there exists a set~$C_j\in \mathcal{C}$ such that $\{c^i_j,d^i_j\}\subseteq V'$.
        \item\label{block-exact} If $\{c_j^i,d_j^i\}\cap V' \neq \emptyset$ for some $j\in [\emm]$ and $i \in C_j$, then for all $k\in C_j$ it holds that $\{c_j^k,d_j^k\}\subseteq V'$
        and $\{a^4_k,b_k^4\}\cap V'\neq \emptyset$.
      \end{compactenum}
    \end{clm}
    \begin{proof}[Proof of \cref{claim:blocking}]
      \renewcommand{\qedsymbol}{$\diamond$}
      All statements can be shown using reasoning similar to the one for \eqref{eq:no-xi^0}--\eqref{%
        eq:all-in-or-all-out}.
      We prove them here for the sake of completeness.

      First, it is straightforward to see that no dummy agent from~$D$ will participate in a weakly blocking coalition since each dummy agent only weakly prefers to be in a coalition~$W$ (other than her initially coalition) that consists of herself and the other agent~$v_i^4$ with $(i,v)\in [3\enn]\times \{a,b\}$ but this agent~$v_i^4$ strictly prefers her initial coalition (where she has two friends) to~$W$.
      Second, by a similar reasoning, we can show that the no private agent~$v_i^z$, $(i,z)\in [3\enn]\times \{0,\ldots,L\}$ is in~$V'$.
      This yields Statement~\eqref{block-dummy-private}.
     Combining the second part of Statement~\eqref{block-dummy-private} with \cref{obs:sym-prefs}\eqref{obs:sym-Prefs-elem1} and \cref{obs:sym-Pi}\eqref{obs:sym-Pi-elem}, we infer the following:
      \begin{multline}
        \text{For all } (i',v)\in [3\enn]\times \{x,y\} \text{ it holds that if } \\
        \{v_{i'},\hat{v}_{i'}\} \cap V' \neq \emptyset,
        \text{ then } \{v_{i'},\hat{v}_{i'},t_{i'-1},s_{i'}\}\subseteq V',\label{eq:xi}
      \end{multline}
      where $t_{i'-1}=t_{3\enn}$ if $i'=1$.

      Now, to show Statement~\eqref{block-elem}, assume that there exists an $i\in [3\enn]$ and a $v\in \{x,y\}$, such that $\{v_{i},\hat{v}_{i}\}\cap V' \neq \emptyset$.
      In the following, we let $i-1=3\enn$ if $i=1$, and $i+1=1$ if $i=3\enn$.
      By \eqref{eq:xi}, we have that $\{v_i,\hat{v}_i,t_{i-1},s_i\}\subseteq V'$.
      This implies that $t_{i-1}$ and $s_i$ each have at least two enemies in $V'$.
      Since both $t_{i-1}$ and $s_i$ have exactly two friends and one enemy in their respective initial coalitions (see \cref{obs:sym-Pi}\eqref{obs:sym-Pi-elem}), %
      for $V'$ to be a weakly blocking coalition, both $t_{i-1}$ and $s_i$ must have at least three friends in~$V'$.
      By their preferences, it follows that
      \begin{align}
       |\{a^0_{i-1},b^0_{i-1}, x_{i}, y_{i}\}\cap V'| \ge 3 \text{ and }         |\{a^0_{i},b^0_{i}, \hx_{i}, \hy_{i}\}\cap V'| \ge 3.\label{eq:a0}
      \end{align}
      By \eqref{eq:xi} it follows that
      \begin{align}
        \{x_{i},\hx_{i},t_{i-1},s_{i}\} \subseteq V' \text{ or }  \{y_{i},\hy_{i},t_{i-1},s_{i}\} \subseteq V'.\label{eq:st}
      \end{align}
      Moreover, since each agent from $\{a^0_{i-1},b^0_{i-1}, a^0_{i}, b^0_{i}\}$ initially has two friends and one enemy in her respective coalition
      and has at least two enemies in~$V'$,
      it follows that each agent in~$\{a^0_{i-1},b^0_{i-1}\}\cap V'$ (resp.\ $\{a^0_{i}, b^0_{i}\}\cap V'$) must have all their three friends in $V'$ (cf.\ \eqref{eq:a0}).
      By construction, it follows that
      \begin{align}
\nonumber        \{a^0_{i-1},a^1_{i-1}\} \subseteq V \text{ or } &  \{b^0_{i-1},b^1_{i-1}\} \subseteq V, \text{ and }\\
 \nonumber       \{a^0_{i},a^1_{i}\} \subseteq V \text{ or } & \{b^0_{i},b^1_{i}\} \subseteq V, \text{ and }\\
        \{s_{i-1},t_{i-1},s_i,t_i\} \subseteq & V'.\label{eq:a1}
      \end{align}
      Since both $s_{i-1}$ and $t_i$ have two friends and one enemy in their respective initial coalition
      and have now at least two enemies in~$V'$,
      in order for them to join the blocking coalition~$V'$, 
      they must each have at least three friends in~$V'$:
        \begin{align}
        \nonumber  |\{a^0_{i-1},b^0_{i-1}, \hx_{i-1}, \hy_{i-1}\}\cap V'| \ge 3 \text{ and }\\
       |\{a^0_{i},b^0_{i}, x_{i+1}, y_{i+1}\}\cap V'| \ge 3.\label{eq:a0-2}
        \end{align}
        By \eqref{eq:xi}, this implies that
        \begin{align}
          \nonumber
          \{x_{i-1},\hx_{i-1}\}\subseteq V'\text{ or } & \{y_{i-1},\hy_{i-1}\}\subseteq V',
          \text{ and }\\
          \{x_{i+1},\hx_{i+1}\}\subseteq V'\text{ or } & \{y_{i+1},\hy_{i+1}\}\subseteq V'.\label{eq:xii}
        \end{align}
        
       By \eqref{eq:a0}--\eqref{eq:xii} and by repeatedly using the aforementioned reasoning, we conclude Statement~\eqref{block-elem}.

      Statement~\eqref{block-select} follows by a similar reasoning.
      Assume that there exists an $i\in [3\enn]$ such that $\{s_{i},t_{i}\}\cap V' \neq \emptyset$.
      Then, by \cref{obs:sym-Pi}\eqref{obs:sym-Pi-elem},
      we immediately find that $|\{\hx_i,\hy_i,a_i^0,b_i^0\}\cap V'|\ge 3$ (if $s_i\in V'$)
      or $|\{x_{i+1},y_{i+1},a_i^0,b_i^0\}\cap V'|\ge 3$ (if $t_i\in V'$).
      In any case, we obtain that there exists an $i'\in [3\enn]$ such that $\{x_{i'},y_{i'},\hx_{i'},\hy_{i'}\}\cap V'\neq \emptyset$,
      which by Statement~\eqref{block-elem}, implies the conclusion in Statement~\eqref{block-select}.      

       To show Statement~\eqref{block-connect}, let us assume that $v_i^z\in V'$ holds for some~$(v,i,z)\in \{a,b\}\times [3\enn] \times \{0,1,2,3\}$.

       We distinguish between three cases.
       If $z=1$, then since $v_i^1$ initially has two friends and no enemies (see \cref{obs:sym-Pi}\eqref{obs:sym-Pi-connect}),
       in order to weakly block, either $V'$ has exactly two friends and no enemies or at least three friends for~$v_i^1$.
       The former case implies that $V'=\{v_i^0,v_i^1,v_i^2\}$ or $V'=\{v_i^0,v_i^1,v_i^3\}$ since $V'=\neq \Pi(v_i^1)=\{v_i^1,v_i^2,v_i^3\}$.
       However, $v_i^0$ will not weakly prefer $V'$ to her initial coalition~$\Pi(v^0_i)$ where she has two friends.
       This means that $V'$ must contains three friends for~$v_i^0$, i.e., $\{v_i^0,v_i^1,v_i^2,v_i^3\}\subseteq V'$.

       Similarly, we infer that $\{v_i^1,v_i^2,v_i^3,v_i^4\}\subseteq V'$ when $z\in \{2,3\}$. 

       Finlay, if $z=0$, then since $v_i^0$ initially has two friends and one enemies (see \cref{obs:sym-Pi}\eqref{obs:sym-Pi-connect}),
       in order to weakly block, either $V'$ has exactly two friends and at most one enemy or at least three friends for~$v_i^0$.
       The former case implies that $v_i^1\in V'$.
       However, by the above, we also infer that $|V'|\ge 5$, a contradiction. 
       This means that $V'$ must have at least three friends for~$v_i^0$, i.e., $\{v_i^0,v_i^1,s_i,t_i\}\subseteq V'$.      
      These jointly imply that $\{v_i^{z'}\mid z'\in \{0,\ldots, 4\}\}\cup \{s_i,t_i\}\subseteq V'$, as desired.

      For Statement~\eqref{block-setcover}, we observe that $v_i^4$ initially has no enemy and exactly two friends.
      If $v_i^4\in V'$, then it has to get at least one of $\{ v_i^2,v_i^3\}$ (it has to get at least two friends, but if only two, then they cannot be the same). Then $v_i^2$ or $v_i^3$ also has to get another friend. If $v_i^1\in V'$, then $v_i^4$ has at least one enemy in $V'$. Otherwise $V'=\{ v_i^2,v_i^3,v_i^4\}$ and no agent strictly improves or if there are some others (but not $v_i^1$) in $V'$, then $v_i^2$ and $v_i^3$ will be worse off. So $v_i^4$ must get at least three friends in a weakly blocking coalition.
      This implies that $\{v_i^2,v_i^3\}\cap V' \neq \emptyset$ and there exists a set~$C_j\in \mathcal{C}$ such that $d_j^i\in V'$.
      Since $d_j^i$ initially has exactly one friend and no enemy,
      for $V'$ to be blocking,
      agent~$d_j^i$ must have at least two friends, since it has strictly more enemies.
      By construction, it follows that $c_j^i\in V'$, as desired.
      
      Finally, for Statement~\eqref{block-exact}, assume that $p_j^i\in V'$ for some $j\in [\emm]$, $i\in V'$, and $p\in \{c,d\}$.
      Let $C_j=\{i,k,r\}$.
      Observe that if $p = c$, then $p_j^i$ initially has two friends but no enemy,
      and if $p=d$, then she initially has one friend and one enemy. 
      First, assume that $p=c$.
      Then, if $p_j^i$ would only get two friends in $V'$ (but not the same two),
      then she cannot have any enemy in~$V'$, i.e., either $V'=\{p_j^i,d_j^i,c_j^k\}$ or $V'=\{p_j^i,d_j^i,c_j^r\}$.
      However, neither $c_j^k$ nor $c_j^r$ weakly prefers $V'$ to her initial coalition where she has two friends.
      This means that $p_j^i$ must have three friends in~$V'$, i.e., $\{c_j^i,c_j^k,c_j^r,d_j^i\}\subseteq V'$.
       Analogously, we infer that $\{c_j^i,c_j^k,c_j^r,d_j^i,d_j^k,d_j^r\}\subseteq V'$.

       Next, assume that $p=d$.
       Then, $p_j^i$ would get only one friend in~$V'$, then this friend cannot be $a_i^4$ or $b_i^4$ as by Statement~\eqref{block-setcover}, she would have at lest two enemies and would not weakly prefer $V'$ to her initial coalition.
       This means that $c_j^i\in V'$, and by previous reasoning that $p_j^i$ needs to get at least two friends in~$V'$ as otherwise she would not weakly prefer $V'$ to her initial coalition.
       That is, $\{c_j^i,d_j^i,a_i^4\}\subseteq V'$ if $C_j\in \mathcal{C}^{\textsf{out}}$ or
       $\{c_j^i,d_j^i,b_i^4\}\subseteq V'$ if $C_j\in \mathcal{C}^{\textsf{in}}$.
       Combined, we obtain Statement~\eqref{block-exact}.
    \end{proof}
    Now, we are ready to show the following.
    \begin{clm}
      \label{claim:sym-back}
      If $V'$ admits a weakly blocking coalition, then $I$ admits an exact cover.
    \end{clm}
    \begin{proof}[Proof of \cref{claim:sym-back}]
      \renewcommand{\qedsymbol}{$\diamond$}
      Let $V'$ be a blocking coalition of~$\Pi$.
      We claim that $\mathcal{K}\coloneqq \{C_j\mid c_j^i\in V' \text{ for some } i\in C_j\}$ is an exact cover for~$I$.
      First of all, we show that $\{x_i,y_i,\hx_i,\hy_i\}\cap V'\neq \emptyset$ for some~$i\in [3\enn]$.
      Suppose, for the sake of contradiction, that $\{\{x_i,y_i,\hx_i,\hy_i\} \mid i\in [3\enn]\}\cap V'= \emptyset$.
      Then, by the contra-positives of \cref{claim:blocking}\eqref{block-elem}--\eqref{block-select},
      it follows that $\{v_i,\hat{v}_i,s_i,t_i\mid (i,v)\in [3\enn]\times \{x,y\}\}\cap V' = \emptyset$.
      Hence, by the contra-positive of \cref{claim:blocking}\eqref{block-connect},
      it follows that $\{v_i^z\mid (i,v,z)\in [3\enn]\times \{a,b\}\times \{0,\ldots,3\}\}\cap V'= \emptyset$.
      By the contra-positive of \cref{claim:blocking}\eqref{block-setcover},
      we infer that $\{v_i^4\mid i\in [3\enn]\}\cap V'=\emptyset$.
      By the contra-positive of \cref{claim:blocking}\eqref{block-exact}, we infer that $\{c_j^k,d_j^k\mid j\in [\emm]\text{ and }k\in C_j\}\cap V'=\emptyset$,
      leaving $V'\cap D\neq \emptyset$ or $v_i^z\in V'$ for some $(v,i,z)\in \{x,y\}\times [3\enn]\times \{0,\ldots,L\}$ since $V'$ cannot be empty,
      a contradiction to \cref{claim:blocking}\eqref{block-dummy-private}.

      Now, we have shown that $\{x_i,y_i,\hx_i,\hy_i\}\cap V'\neq \emptyset$ for some~$i\in [3\enn]$.
      By \cref{claim:blocking}\eqref{block-elem}--\eqref{block-setcover}, we obtain that $\mathcal{K}$ is indeed a set cover.

      It remains to show that $\mathcal{K}$ is an exact cover, i.e., $|\mathcal{K}| \le \enn$.
      Since for all $i\in [3\enn]$, agent~$x_i$ initially has two friends and $L=39\enn-2$ enemies in~$\Pi(x_i)$,
      and has two friends in $V'$ (see \cref{claim:blocking}\eqref{block-dummy-private}),
      it follows that $x_i$ must have at most~$L=39\enn-2$ enemies in $V'$ too,
      implying that $|V'|\le 39\enn+1$.
      By \cref{claim:blocking}\eqref{block-elem}--\eqref{block-connect}, we obtain that
      $|\{x_i,y_i,\hx_i,\hy_i,s_i,t_i,a_i^0,\ldots,a_i^4,b_i^0,\ldots,b_i^4\}\subseteq V'| \ge 11$ holds for all $i\in [3\enn]$.
      By \cref{claim:blocking}\eqref{block-exact}, we obtain that
      $|\{c_j^k,d_j^k \mid j\in \mathcal{K}, i\in C_j\}\cap V'| \ge 6|\mathcal{K}|$.
      This implies that $39\enn+1\ge |V'|\ge 33\enn+6|\mathcal{K}|$, resulting in $|\mathcal{K}|\le \enn+\frac{1}{6}$, so $|\mathcal{K}|\le \enn$, since it is an integer, as desired.   
      \end{proof}
   Together with \cref{claim:planar-sym-forward,claim:sym-back}, we obtain the correctness proof for our construction.
 \end{proof}
}

 Next, strengthening the result by \citet{Brandt_Bullinger_Tappe_2022}, we show that finding a Nash stable solution is \np-hard even if the friendship graph has constant maximum degree and is planar.

 \newcommand{\fenashhard}{%
      \FE-\NS\ is \np-complete even if the friendship graph is planar and $\maxdeg=9$.
 }
 \begin{theorem}[\appendixsymb]\label{thm:fe-ns-hard}
  \fenashhard %
 \end{theorem}

 \appendixproofwithstatement{thm:fe-ns-hard}{\fenashhard}{
 \begin{proof}
   To show hardness, we reduce from \pxct. %
   Let $I=([3\enn], \mathcal{C})$ be an instance of \pxct\ with $\mathcal{C}=\{C_1. \ldots, C_{\emm}\}$.
   We create an instance of \FE-\NS\ as follows.
   \begin{compactitem}[--]
     \item For each element~$i\in [3\enn]$, create an \myemph{element} agent $x_i$. Define $X=\{x_i\mid i\in [3\enn]\}$. 
     \item For each set~$C_j\in \mathcal{C}$, create a \myemph{set} agent~$s_j$ and a set of $9$ \myemph{enforcers}, called $s_j^0$,   $s_j^1$, $s_j^2$,  $t_j^0, \ldots, t^5_j$.
     For each element $i\in C_j$, agent~$s_j$ and $x_i$ are mutual friends.
     Additionally $s_j$ considers the agents from $\{s^0_j,s^1_j,s^2_j\}$ friends.
     Agent~$s_j^0$ (resp.\ $s_j^1$ and $s_j^2$) considers agent~$s_j^1$ (resp.\ $s_j^2$ and $s_j^0$) a friend.
     Agent~$s_j^0$ also considers agent $t^0_j$ a friend.
     Agent~$t_j^0$ (resp.\ $t_j^1$ and $s_j^2$) considers agent~$s_j^1$ (resp.\ $s_j^2$ and $s_j^0$) a friend.
     Agent~$t_j^2$ also considers agent $t^3_j$ a friend.
     Agents~$t_j^3$ and $t_j^4$ are mutual friends.
     Agent~$t_j^3$ also considers agent $t^5_j$ a friend.
     Summarizing, the friendship graph consists of the following arc set $\{(x_i, s_j), (s_j, x_i)\mid i\in C_j, i\in [3\enn], j\in [\emm]\}\cup \{(s_j,s_j^z), (s_j^z, s_j^{z+1}), (t_j^z, t_j^{z+1})\mid 0\le z \le 2, j\in [\emm]\}\cup \{(s_j^0, t_j^{0}), (t_j^2, t_j^3), (t_j^3, t_j^4),(t_j^4, t_j^3), (t_j^4, t_j^5)\}$, where superscript $z+1$ is taken modulo $3$.
   \end{compactitem}

   \begin{figure}
    \def \rad {.5}
    \def \radd {1}
    \def \ra {1.9}
    \centering
    \begin{tikzpicture}[scale=1,every node/.style={scale=0.9},>=stealth', shorten <= 2pt, shorten >= 2pt]
      
      \foreach  \x / \y / \n in {.3/0/s, -.8/-.5/x1, -.8/0/x2, -.8/.5/x3}
        {
          \node[pn] at (\x, \y) (\n) {};
        }
        \foreach \n / \nn / \p / \l / \r  / \c in
        {s/{s_j}/below right/4/-4/black, x1/{x_i}/left/1/0/black, x2/{x_p}/left/1/0/black, x3/{x_q}/left/1/0/black}{ 
          \node[\p = \l pt and \r pt of \n, text=\c, inner sep=.5pt, fill=white] {${\nn}$};
        }
      \begin{scope}[xshift=1.5cm]

        \trianglegadget{s}

      \end{scope}
      
      \begin{scope}[xshift=3.2cm]
        
        \trianglegadget{t}

      \end{scope}

      \begin{scope}[xshift=4.4cm]
        
        \foreach  \x / \y / \n in {0/0/t3, .8/0/t4, 1.5/0/t5}
        {
          \node[pn] at (\x, \y) (\n) {};
        }

        \foreach \n / \nn / \p / \l / \r  / \c in
        {t3/{3}/below/2/0/black, t4/{4}/below/2/0/black, t5/{5}/below/2/0/black}{ 
          \node[\p = \l pt and \r pt of \n, text=\c, inner sep=.5pt, fill=white] {$t^{\nn}_j$};
        }
      \end{scope}

      \begin{pgfonlayer}{bg}
      \foreach \s / \t / \aa / \type in {s/x1/10/fc,s/x2/10/fc,s/x3/10/fc, x1/s/10/fc,x2/s/10/fc,x3/s/10/fc, s/sa/35/fc,s/sc/-35/fc, s/sb/0/fc, sa/sb/0/fc,sb/sc/0/fc, sc/sa/0/fc, sa/ta/0/fc, ta/tb/0/fc, tb/tc/0/fc, tc/ta/0/fc, tc/t3/0/fc, t3/t4/30/fc, t4/t3/30/fc, t4/t5/0/fc} {
         \draw[->, \type] (\s) edge[bend right = \aa] (\t);
       }  
      \end{pgfonlayer}
      
    \end{tikzpicture}
    \caption{Illustration for the set and element gadgets for \cref{thm:fe-ns-hard}, where $C_j=\{i, p, q\}$.}\label{fig:fe-ns}
  \end{figure}
  
  This completes the construction of the instance.
  See \cref{fig:fe-ns} for an illustration.
  Clearly, the set agents have the highest degree, which is $9$.
  It is also straightforward that the corresponding friendship graph is planar since $I$ has a planar embedding.  
  It remains to show the correctness, i.e., $I$ admits an exact cover if and only if the constructed instance admits a Nash stable partition.

 For the ``only if'' part, let $\mathcal{K}$ be an exact cover.
 We claim that partition~$\Pi =\{\{s_j,x_i,x_p,x_q\}$, $\{s_j^0,s_j^1,s_j^2\}\mid C_j\!\in\!\mathcal{K}$  with $C_j=\{i,p,q\}\}$ $\cup \{\{s_j,s_j^0,s_j^1,s_j^2\}\mid C_j\!\notin\! \mathcal{K}\}$ $\cup\{\{t_j^0,t_j^1,t_j^2\}\, \{t_j^3,t_j^4\}, \{t_j^5\}\mid j\!\in\! [\emm]\}$ is Nash stable.

 We first consider the enforcers and set agents and let $j\in [\emm]$.
 Observe that agent~$t_j^5$ will not deviate since she does not have any friends and she is alone in~$\Pi$.
 Agent~$t_j^4$ will not deviate either since she has two friends and is in a coalition with one friend and no enemies. 
 Similarly, no agent from $\{t_j^0, t_j^1, t_j^3, s_j^1, s_j^2\}$ will deviate since each of them has one friend and is in a coalition which contains her only friend.
 Consequently, agent~$t_j^2$ (resp.\ $s_j^0$) will not deviate since she has two friends and is in a coalition with one friend and one enemy (resp.\ at most two enemies), but the other coalition which contains her other friend contains one enemy (resp.\ two enemies) for her.
 Agent~$s_j$ will not deviate either since she has in total six friends and is always in a coalition with exactly three friends and no enemies.

 Finally, we consider the element agents and let $i\in [3\enn]$.
 Agent~$x_i$ has at most three friends and is in a coalition with one friend and two enemies.
 Note that her other remaining set friends do not correspond to sets from the exact cover and are in two different coalitions of size four.
 This means that each of these coalitions contains one friend and three enemies for~$x_i$.
 Hence, agent~$x_i$ will not deviate either.

 For the ``if'' part, let $\Pi$ be a Nash stable partition.
 We claim that $\mathcal{K}=\{C_j\in \mathcal{C} \mid \{x_i, x_p, x_q\}\subseteq \Pi(s_j) \text{ with } C_j=\{i,p,q\}\}$ is an exact cover.
 Before we show this, we observe some useful properties of~$\Pi$.
 \begin{clm}\label{clm:either-cover-or-not}
   \begin{compactenum}[(i)]
     \item \label{fe-ns-enforcer} 
     For each~$j\in [\emm]$, it holds that $\{s_j^0,s_j^1,s_j^2\}\subseteq \Pi(s_j^0)$ and $\Pi(s_j^0)$ contains at most two enemies for~$s_j^0$. 
      \item \label{fe-ns-set} It holds that $\Pi(s_j)\neq \Pi(s_{\ell})$ for all set agents~$s_j, s_{\ell}$, $j\neq \ell \in [\emm]$. %
   \end{compactenum}
 \end{clm}
 \begin{proof}[Proof of \cref{clm:either-cover-or-not}]
    \renewcommand{\qedsymbol}{$\diamond$}
     For Statement~\eqref{fe-ns-enforcer},
     we first observe that $\Pi(t_j^5)=\{t_j^5\}$ since she has no friends (i.e., only enemies).
    Then, in order to not envy $\Pi(t_j^5)$ where she has one friend, agent~$t_j^4$ must be in a coalition with one friend and no enemies (since she has in total two friends).
    This is only possible if $\Pi(t_j^4)=\{t_j^3, t_j^4\}=\Pi(t_j^3)$.
    Similarly, in order to not envy~$\Pi(t_j^3)$ where she has one friend, agent~$t_j^2$ must be in a coalition with one friend and at most one enemy.
    This means $t_j^1\in \Pi(t_j^0)$.
    Now, since $t_j^1$ considers $t_j^0$ an enemy, to avoid~$t_j^1$ from preferring
    to be alone, agent~$t_j^1$ must have her only friend in her coalition, implying that  $t_j^2\in \Pi(t_j^0)$.
    In other words, we have $\Pi(t_j^2)=\{t_j^0, t_j^1, t_j^2\}=\Pi(t_j^0)$.
    Finally, to not envy $\Pi(t_j^0)$ where she has one friend, agent~$s_j^0$ must be in a coalition with one friend and at most two enemies.
    Analogously, we have $\{s_j^0,s_j^1,s_j^2\}\subseteq \Pi(s_j^0)$, as desired.

    For Statement~\eqref{fe-ns-set}, suppose for a contradiction that $\Pi(s_j)=\Pi(s_\ell)$ holds for two distinct set-agents~$s_j, s_{\ell}$, $j\neq \ell \in [\emm]$.
    Without loss of generality assume that $C_j\neq C_{\ell}$. Then, neither of $s_j^0$ and $s_{\ell}^0$ can be in $\Pi(s_j)$ as they would get at least three enemies and they would envy $\Pi (t_j^0)$ or $\Pi (t_{\ell}^0)$, respectively, a contradiction.

  However, this implies that $s_j$ and $s_{\ell}$ must have three friends, so both of them gets all element agent friends. Hence, $s_j$ and $s_{\ell}$ have at least two enemies in $\Pi(s_j)$ (namely $s_{\ell}$/$s_j$ and at least one element agent). It follows that $s_j$ envies $\Pi(s_j^0)$ since it contains three friends and at most one enemy for her (see \cref{clm:either-cover-or-not}\eqref{fe-ns-enforcer}),
  a contradiction.
  \end{proof}

  To continue, \cref{clm:either-cover-or-not}\eqref{fe-ns-enforcer} implies that $\mathcal{K}$ is a set packing, i.e., no two sets in~$\mathcal{K}$  intersects.
  To complete the proof, we show that $\mathcal{K}$ is also a set cover, i.e., $\cup_{C\in \mathcal{K}}C=[3\enn]$.
  Consider an arbitrary element~$i\in [3\enn]$.
  Since $x_i$ considers the three corresponding set agents as friends, it follows that $x_i\in \Pi(s_j)$ for some set $C_j$ with $C_j=\{i,p,q\}$.
  It suffices to show that $\{x_i, x_p, x_q\}\subseteq \Pi(s_j)$. 
  Note that $\Pi(s_i)\cap \Pi(s_j^0)=\emptyset$ as otherwise $s_j^0$ will have at least three enemies (including $x_i, s_j, s_j^2$) in her coalition $\Pi(s_j^0)$, a contradiction to \cref{clm:either-cover-or-not}\eqref{fe-ns-enforcer}.
  Now, to not envy $\Pi(s_j^0)$, agent~$s_j$ must be in a coalition with at least three friends.
  This is only possible if all her element friends are in her coalition, as desired. %
\end{proof}
}
\subsection{Algorithms and refined complexity for FE}\label{sec:refined}

We start with some simple polynomial-time algorithms for core verification.
\newcommand{\propfecorealg}{%
For each of the following cases, \FE-\verif and \FE-\sverif are polynomial-time solvable.
\begin{compactenum}[(i)]
  \item The friendship graph is a acyclic.
  \item $\maxdeg=2$.
  \item The preferences are symmetric and $\maxdeg=4$. %
\end{compactenum}
}
\begin{proposition}[\appendixsymb]
\label{obs:FE-acyclic}
\propfecorealg
\end{proposition}
\appendixproofwithstatement{obs:FE-acyclic}{\propfecorealg}{
  \begin{proof}
  We consider the acyclic friendship graph case separately.
When the associated friendship graph is acyclic, every strongly connected component is a singleton.
Hence, the only coalition structure which is core stable is the one where each agent is in a singleton coalition. Hence, verifying (strictly) stable core is polynomial-time solvable.

In the remaining two case, the underlying directed or undirected graph $G$ has maximum degree at most $2$. Therefore it is a disjoint union of path and cycles. %

Now, if there is a blocking coalition, such that it induces a disconnected graph, then taking only one of the connected components must also be a blocking coalition, since those agent in it get the same number of friends and less enemies. 

Now, it is easy to see that a disjoint union of paths and cycles on $n$ agents has at most $n^2$ connected subgraphs, so we can check all of them whether they are blocking in polynomial-time. 
\end{proof}
}

While both verification problems are trivial if the friendship graph is acyclic, we show that interestingly even one feedback arc makes the problem intractable.
\newcommand{\corefashard}{%
  \FE-\verif\ and \FE-\sverif\ are \conp-com\-plete even if $\fas=1$, and each agent has at most $3$ friends. 
}
\begin{theorem}[\appendixsymb]\label{thm:fas_number1}
\corefashard %
\end{theorem}
\appendixproofwithstatement{thm:fas_number1}{\corefashard}{
  \begin{proof}
  To prove the theorem, we first consider \FE-\verif\ and provide a polynomial reduction from the \np-complete \pxct problem. Later we show that the same reduction works for \FE-\sverif\ with slight modifications.
  Let $I=([3\enn],\mathcal{C})$ denote an instance of \pxct, where $\mathcal{C}=\{C_1,\dots,C_{\emm }\}$.
  We construct an instance of \FE-\verif\ as follows.

\begin{itemize}[--]
    \item For each element $i\in [3n]$, create an element agent $a_i$.
    \item For each set $C_j$, $j\in [\emm ]$, create a set agent $c_j$.
    \item Create $4\enn+3$ special agents, called~$s_0,s_1, \dots, s_{4\enn+2}$.
    The special agents~$s_1, \dots, s_{3\enn}$ will act as the element selectors,
    while agent~$s_0$ will act as the set selector. The remaining agents are dummies that enforce a maximum size for the blocking coalition.
    \item Finally, create $3\enn-1$ additional agents, named $x_1, \dots, x_{3\enn-1}$.

  \end{itemize}
 
\noindent The friendship graph has the arcs: %

$\{(c_j,s_0) \mid j\in [\emm]\} $ 
    $\bigcup ~ \{(s_i,a_i) \mid i\in [3\enn]\}$
     $\bigcup~ \{(a_i,c_j) \mid$ for all $i, j\in [3\enn] \times [\emm ]$ such that $i\in C_j\}\ \bigcup$ 
     $\{(s_z,s_{z+1}) \mid$ $z \in \{0\} \cup [4\enn+2]\} \bigcup~ \{(s_i,x_i), (x_i,s_{i+1}), (x_i,s_{3\enn+1}) \mid $ $i \in [3\enn-1]\}$, where $z+1$ is taken modulo $4\enn+3$.

Before specifying the initial partition, we analyze the maximum number of friends an agent has and the feedback arc set number.
First, we observe that the friendship graph has one feedback arc: $(s_0, s_1)$.
Without it, the order 
$(s_1,x_1,s_2,x_2, \dots, s_{3\enn-1},x_{3\enn-1},$ $s_{3\enn}, \dots,$ $ s_{4\enn+2},
a_1,\dots, a_{3\enn}$, $c_1, \dots, c_{\emm}, s_0)$
is a topological order of the vertices, so the remaining digraph is acyclic.

It is also straightforward to see that each element agent~$a_i,$ $ i \in [3\enn]$ has three friends,
each set agent $c_j, j \in [\emm]$ has one friend,
each special agent $s_z$, $z\in \{0, 3\enn+1,\dots,4\enn+2\}$ has one friend,
each additional agent~$x_{i}$, $i\in [3\enn-1]$ and
agent~$s_{3\enn}$ have two friends,
and each remaining special agent from $\{s_1, \dots, s_{3\enn-1}\}$ has three friends.
Summarizing, every agent has at most three friends.

To complete the construction, we define the initial partition as follows: $\Pi = \{\{s_0, \dots, s_{4\enn+2}, x_1, \dots, x_{3\enn-1}\}\} \cup \{\{c_j\} \mid j \in [\emm]\} \cup \{\{a_i\} \mid i \in [3\enn]\}$.

\begin{figure}
       \centering
       \begin{tikzpicture}[scale=1,every node/.style={scale=0.9}, >=stealth', shorten <= 2pt, shorten >= 2pt]

         \begin{scope}[xshift=-1.8cm, yshift=-1cm]
           \foreach  \x / \y / \n  / \st in
           {1.2/0/a1/pn, 1.2/-.5/a2/pn, 1.2/-1.2/a3np/pnn, 1.2/-1.5/a3n/pn,
             -1.5/0/x1/pn, -1.5/-.5/x2/pn, -1.5/-1/x3np/pnn, -1.5/-1.5/x3n/pn,
             0/0.5/s1/pn, 0/0/s2/pn, 0/-0.52/s3np1/pnn,0/-0.65/s3np/pnn,0/-0.75/s3np3/pnn, 0/-1.25/s3n/pn}
           {
             \node[\st] at (\x*0.8, \y) (\n) {};
           }
         \end{scope}

         \begin{scope}[xshift=.5cm, yshift=-.7cm]
           
           \foreach  \x / \y / \n / \st in
           {0/0/c1/pn, 0/-.5/c2/pn, 0/-1/cmp/pnn, 0/-1.5/cm/pn}
           {
             \node[\st] at (\x, \y) (\n) {};
           }
         \end{scope}

         \begin{scope}[xshift=2cm, yshift=-.7cm]
             \node[pn] at (0,-1) (s0) {};
         \end{scope}

         \begin{scope}[xshift=-1.8cm, yshift=-3cm]
    
           \foreach  \x / \y / \n / \st in
           {  0/0/s3n1/pn, 1/0/s3n2/pn, 1.55/0/smpp1/pnn,1.75/0/smpp/pnn,1.95/0/smpp3/pnn, 2.5/0/s4n1/pn, 3.5/0/s4n2/pn}
           {
             \node[\st] at (\x, \y) (\n) {};
           }
         \end{scope}

         \foreach \n / \nn / \p / \l / \r  / \c in
         { x1/{x_1}/above left/1/-1/black, x2/{x_2}/above left/1/-1/black, x3n/{x_{3\enn-1}}/above left/0/-1/black,
           a1/{a_1}/below/1/0/black, a2/{a_2}/below/0/0/black, a3n/{a_{3\enn}}/below/1/0/black, 
            c1/{c_1}/above right/0/0/black, c2/{c_2}/above right/0/0/black, cm/{c_{\emm}}/above right/0/0/black,
             s1/{s_{1}}/above right/2/0/black, s2/{s_{2}}/above right/0/0/black,
             s3n/{s_{3\enn}}/above right/0/0/black, s3n1/{s_{3\enn+1}}/below right/0/0/black, s4n2/{s_{4\enn+2}}/below right/0/0/black,  s3n2/{s_{3\enn+2}}/below right/0/0/black, s4n1/{s_{4\enn+1}}/below right/0/0/black,
            s0/{s_{0}}/above right/1/0/black}{
           \node[\p = \l pt and \r pt of \n, text=\c, inner sep=.5pt, fill=white] {\small ${\nn}$};
         }

         \node at (smpp) {$\ldots$};
         \node[rotate=90] at (cmp) {$\ldots$};
         \node[rotate=90, xshift=4pt] at (a3np) {$\ldots$};
         \node[rotate=90,xshift=0pt] at (x3np) {$\ldots$};
         \node[rotate=90,xshift=0pt] at (s3np) {$\ldots$};
         \begin{pgfonlayer}{bg}
           \foreach \s / \t / \aa / \type in {
             s1/x1/0/fc, x1/s2/0/fc, s2/x2/0/fc, x2/s3np/0/fc, s3np/x3n/0/fc,
              x3n/s3n/0/fc, 
              x1/s3n1/0/fc, x2/s3n1/0/fc, x3n/s3n1/0/fc,
              s1/s2/0/fc, s2/s3np1/0/fc, s3np3/s3n/0/fc, s3n/s3n1/0/fc, s3n1/s3n2/0/fc, s3n2/smpp1/0/fc, smpp3/s4n1/0/fc, s4n1/s4n2/0/fc,
             c1/s0/0/fc,c2/s0/0/fc,cm/s0/0/fc} {
             \draw[->, \type] (\s) edge[bend right = \aa] (\t);
           }

            \foreach \i in {1,2,3n} { 
             \draw[->, fc] (s\i) edge[] (a\i);
           }

 \draw[yellow, very thick, rounded corners] (s0.east) .. controls ($(s0)+(1.2, 0.5)$) and ($(s0)+(1.2, .8)$)  .. ($(c1)+(0.4, 0.4)$) .. controls ($(c1)+(0, 0.4)$) and ($(s1)+(1.2, .4)$)  ..  (s1.east);

           \draw[->, fc, rounded corners] (s0.east) .. controls ($(s0)+(1.2, 0.5)$) and ($(s0)+(1.2, .8)$)  .. ($(c1)+(0.4, 0.4)$) .. controls ($(c1)+(0, 0.4)$) and ($(s1)+(1.2, .4)$)  ..  (s1.east);
           
         \draw[->, fc, rounded corners] (s4n2.west) .. controls ($(s4n2)+(2, 0)$) and ($(s0)+(2, 0)$)   .. (s0.east);

           \draw[->] (a1) edge (c2);
           \draw[->] (a1) edge ($(c2)!0.5!(cmp)$);
           \draw[->] (a1) edge ($(c2)!0.9!(cmp)$);

           \draw[->] (a2) edge (c1);
           \draw[->] (a2) edge (cmp);
           
           \draw[->] (a3n) edge (cmp);
           \draw[->] (a3n) edge ($(cmp)!0.5!(cm)$);
           \draw[->] (a3n) edge (cm);
         \end{pgfonlayer}
     
       \end{tikzpicture}\caption{The construction in \cref{thm:fas_number1} where \fas = 1 and the highlighted arc is the only feedback arc.}   \label{fig:fas_number1}
     \end{figure}

We observe the following for the initial partition~$\Pi$.
\begin{observation} \label{obs:fas} %
   \begin{compactenum}[(i)]
  \item \label{obs:fas,1i} For every $i \in [3\enn]$ and $j \in [\emm]$,  $a_i$ and $c_j$ have zero friends and zero enemies,
  \item \label{obs:fas,1ii}For every $z \in [3\enn - 1]$, $s_z$ and $x_z$ each have two friends and $7\enn -1$ enemies,
  \item  \label{obs:fas,1iii}For every $z \in \{0,3\enn, \dots, 4\enn+2\}$, $s_z$ has 1 friend and $7\enn $ enemies.
  \end{compactenum}
\end{observation}

It remains to show the correctness, i.e., $I$ admits an exact cover if and only if $\Pi$ is not core stable, i.e., $\Pi$ is strictly blocked by some coalition.
\begin{clm}
\label{claim:fas2}
If $\Pi$ is strictly blocked by some coalition, then $I$ admits an exact cover.
\end{clm} 

\begin{proof}[Proof of \cref{claim:fas2}]
 \renewcommand{\qedsymbol}{$\diamond$}
Let $P$ be a strictly blocking coalition of $\Pi$. 
We first show that $\{ s_0,s_1, \dots, s_{3\enn}\}\subseteq P$ and for every $i \in [3\enn - 1]$, $x_i \notin P$.

We can assume that $P$ is an inclusionwise minimal blocking coalition, therefore the graph induced by $P$ is strongly connected. 
This means that $s_0$ and $s_1$ must be in $P$. Since $P$ is blocking, each agent must have at least the same number of friends as in $\Pi$. Since $s_1$ has two friends in $\Pi$ by \cref{obs:fas}\eqref{obs:fas,1ii}, it must hold that $x_1$ or $s_2\in P$, but if $x_1\in P$, then also $s_2\in P$ by \cref{obs:fas}\eqref{obs:fas,1ii}. By repeating this argument, we have that $\{ s_0,s_1, \dots, s_{3\enn}\}\subseteq P$. 

If there is $i \in [3\enn - 1]$ such that $x_i\in P$, then by \cref{obs:fas}\eqref{obs:fas,1ii}, $x_i$ must get both her  friends, so $s_{3\enn+1}\in P$. Then $\{s_{3\enn+2}, \dots, s_{4\enn+2}\} \subset P$ by \cref{obs:fas}\eqref{obs:fas,1iii}, because for  every $z \in \{3\enn+1,\dots ,4\enn+2\}$, the only friend of $s_z$ is $s_{z+1}$ (here we take $z + 1$ modulo $4\enn + 3$). But then the size of $P$ is at least $(4\enn+3)+(3\enn-1)=7\enn+2$, since $\{s_0, \dots, s_{4\enn+2}\} \subseteq P$ and each $s_z, z \in [3\enn -1]$ must have a friend other than $s_{z+1}$, which can be either $x_z$ or $a_z$. Then, $s_0$ cannot strictly improve by  \cref{obs:fas}\eqref{obs:fas,1iii}, since she would have at least $7\enn$ enemies. This concludes that for no $i \in [3\enn - 1], x_i \in P$.

As no agent $x_i$ is in $P$ for any $i \in [3\enn -1]$, and each of $\{s_1, \dots,$ $s_{3\enn - 1}\}$ must have 2 friends to join $P$, $\{a_1, \dots, a_{3\enn-1}\} \subset P$. Since every $a_i$, $i \in [3\enn - 1]$ has to obtain a friend in $P$, at least $\enn$ set agents have to be in $P$. Furthermore, if $s_{3\enn+1}\in P$, then as we have seen, for every $z \in \{3\enn + 1, \dots, 4\enn + 2\}$, $s_z \in P$. Thus the size of $P$ would be at least $(4\enn+3)+(3\enn-1)+\enn =8\enn+2>7\enn+2$, so $s_0$ would disimprove, a contradiction. Therefore $s_{3\enn+1}\notin P$ and $s_{3\enn}$ must obtain $a_{3\enn}$ as a friend, hence $a_{3\enn}\in P$. The agent $s_0$ can only improve if the size of $P$ is strictly smaller than $7\enn+2$ by \cref{obs:fas}\eqref{obs:fas,1iii}. Thus only $\enn$ set agents can be included, so these set agents must form an exact cover.
\end{proof}
\begin{clm}
\label{claim:fas3}
If $I$ admits an exact cover, then there is a strictly blocking coalition $P$.
\end{clm} 
\begin{proof}[Proof of claim \ref{claim:fas3}] \renewcommand{\qedsymbol}{$\diamond$}
  Let $\mathcal{K}$ be an exact cover. Then, let $P = \{a_i \mid i \in [3\enn]\} \cup \{c_j \mid C_j \in \mathcal{K}\} \cup \{s_0, \dots, s_{3\enn}\}$. We will show that $P$ is a blocking coalition.
  The element agents and the set agents from $P$ each have one friend,
  so they strictly improve by observation \cref{obs:fas}\eqref{obs:fas,1i}.
  Finally, the special agents~$\{s_0, \dots, s_{3\enn}\}$ have the same number of friends, but since the size of the coalition is only $(3\enn+1)+3\enn+\enn =7\enn+1<7\enn+2$, each of them has fewer enemies, so they strictly improve.
\end{proof}

Claims \ref{claim:fas2} and \ref{claim:fas3} prove the correctness of the construction.

For \sverif\, the only difference in the construction is that we have one less special agent (i.e., we delete $s_{4\enn +2}$ and add the edge $(s_{4\enn + 1}, s_0$). The initial coalition for $s_0$ is $\{s_0,\dots,s_{4\enn+1}, x_1,\dots,x_{3\enn-1}\}$.
This means that $\{s_0, \dots, s_{3\enn}\}$ all have one less enemy originally in $\Pi$.  

Hence, we replace \cref{obs:fas}\eqref{obs:fas,1ii} - \eqref{obs:fas,1iii} by following:

\begin{compactenum}[({i}i')]
\item  \label{obs:fas,1ii_S}For every $z \in [3\enn - 1]$, $s_z$ and $x_z$ each have two friends and $7\enn -2$ enemies.
\item  \label{obs:fas,1iii_S} For every $z \in \{0,3\enn, \dots, 4\enn+2\}$, $s_z$ has 1 friend and $7\enn -1$ enemies.
\end{compactenum}

For the sake of completeness, we give a proof that a weakly blocking coalition implies the existence of an exact cover.

\begin{clm}
\label{claim:fas4}
If $\Pi$ is weakly blocked by some coalition, then $I$ admits an exact cover.
\end{clm}
\begin{proof}[Proof of claim \ref{claim:fas4}] \renewcommand{\qedsymbol}{$\diamond$}
Let $P$ be a weakly blocking coalition of $\Pi$. 
We can again assume that $P$ is an inclusionwise minimal blocking coalition, therefore the graph induced by $P$ is strongly connected.  

This means that $s_0$ and $s_1$ must be in $P$. Since $s_1$ has two friends in $\Pi$ by \cref{obs:fas}(i\ref{obs:fas,1ii_S}'), it must hold that $x_1$ or $s_2\in P$, but if $x_1\in P$, then also $s_2\in P$ by \cref{obs:fas}(i\ref{obs:fas,1ii_S}'). Continuing using this argument, we can see that $\{ s_0,s_1, \dots, s_{3\enn}\}\subset P$. 

If there is $i \in [3\enn - 1]$ such that $x_i\in P$, then by \cref{obs:fas}(i\ref{obs:fas,1ii_S}'), $x_i$ must get both her  friends, so $s_{3\enn+1}\in P$. As for every $z \in \{3\enn+1,\dots ,4\enn +1\}$ the only friend of $s_z$ is $s_{z+1}$, we have that $s_{3\enn+2}, \dots, s_{4\enn+1}$ are also in $P$ by  \cref{obs:fas}(i\ref{obs:fas,1iii_S}'). Then, the size of $P$ is at least $(4\enn+2)+(3\enn-1)=7\enn+1$, since all $\{s_0, \dots, s_{4\enn - 1}\} \subseteq P$ and for every $z \in \{1, \dots, 3\enn-1\}, s_z$ must have a friend other than $s_{z+1}$, which can be either $x_z$ or $a_z$. But then, $s_0$ must have at least $7\enn$ enemies, so for $s_0$ to weakly improve, there cannot be any more agents inside $P$. But this would mean that no agent strictly improves, a contradiction. Thus there is no $i \in [3\enn - 1]$ such that $x_i \in P$,

Therefore, for each of $\{s_1, \dots, s_{3\enn - 1}\}$ to have 2 friends, we must have that $\{a_1, \dots, a_{3\enn-1}\} \subset P$. Hence, since they all must get a friend to weakly improve, at least $\enn$ set agents have to be in $P$. Furthermore, if $s_{3\enn+1}\in P$, then as we have seen, $\{s_{3\enn + 1}, \dots, s_{4\enn + 1}\} \subset P$. Then the size of $P$ must be at least $(4\enn+2)+(3\enn-1)+\enn =8\enn+1>7\enn+1$, so $s_0$ would be worse off. Therefore $s_{3\enn}$ must obtain only $a_{3\enn}$ as a friend and hence $a_{3\enn}\in P$. For the size of $P$ to be at most $7\enn+1$ ($s_0$ can only weakly improve if this holds by \cref{obs:fas}(i\ref{obs:fas,1iii_S}')), only $\enn$ set agents can be included, so these set agents must form an exact cover.
\end{proof}

The other direction is similar to the \verif case.
The weakly blocking coalition~$P$ will be exactly the same with the only difference being that the special agents~$\{s_0, \dots, s_{3\enn}\}$ obtain the same number of enemies as in $\Pi$, so they only weakly improve.
\end{proof}
}

\noindent Additionally bounding the maximum degree~$\maxdeg$ does not help to break down the complexity.
\newcommand{\fasfasdelta}{%
  \FE-\verif\ and \FE-\sverif\ are \conp-com\-plete, even if $\fas=2$ and $\maxdeg=5$. %
}
\begin{theorem}[\appendixsymb]
\label{thm:sverif_cont_fas+delta}
\fasfasdelta %
\end{theorem}
\appendixproofwithstatement{thm:sverif_cont_fas+delta}{\fasfasdelta}{
  \begin{proof}%
  We first consider \FE-\sverif\ and provide a polynomial reduction from the \np-complete \xctg problem. In contrast to \pxct, each element appears in exactly three sets.
\decprob{\xctg}
{A $3\enn$-element set~$\mathcal{X}=[3\enn]$ and a collection~$\mathcal{C}=(C_1,\ldots,C_{3\enn})$ of $3$-element subsets of~$X$ such that each element~$i\in X$ appears in exactly three members of~$\mathcal{C}$.}{Does~$\mathcal{C}$ contain an \myemph{exact cover} for~$X$, i.e., a subcollection~$\mathcal{K} \subseteq \mathcal{C}$ such that each element of~$X$ occurs in exactly one member of~$\mathcal{K}$?}
  
  Let $I=([3\enn],\mathcal{C})$ denote an instance of \xctg, where
  the sets are $\mathcal{C}=\{C_0,\dots,C_{3\enn -1}\}$; for the sake of easier reasoning we assume that the set index starts with $0$. We construct an instance of our verification problem as follows.

     \begin{figure}
       \centering
       \begin{tikzpicture}[scale=1,every node/.style={scale=0.9}, >=stealth', shorten <= 2pt, shorten >= 2pt]

         \begin{scope}[xshift=-1.8cm, yshift=-1cm]
           \foreach  \x / \y / \n  / \st in
           {1.2/0/a1/pn, 1.2/-.5/a2/pn, 1.2/-1/a3np/pnn, 1.2/-1.5/a3n/pn,
             -1.5/0/x1/pn, -1.5/-.5/x2/pn, -1.5/-1/x3np/pnn, -1.5/-1.5/x3n/pn,
             0/0.25/s3n/pn, 0/-0.25/s3n1/pn, 0/-0.75/s3np/pnn, 0/-1.25/s6n/pn}
           {
             \node[\st] at (\x*0.8, \y) (\n) {};
           }
         \end{scope}

         \begin{scope}[xshift=.5cm, yshift=-.7cm]
           
           \foreach  \x / \y / \n / \st in
           {0/0/c1/pn, 0/-.5/c2/pn, 0/-1/cmp/pnn, 0/-1.5/cm/pn}
           {
             \node[\st] at (\x, \y) (\n) {};
           }
         \end{scope}

         \begin{scope}[xshift=2cm, yshift=-.7cm]
           
           \foreach  \x / \y / \n / \st in
           {0/0/t0/pn, 0/-.5/t1/pn, 0/-1/tmp/pnn, 0/-1.5/tm/pn,
             1/0/s0/pn, 1/-.5/s1/pn, 1/-1/snp/pnn, 1/-1.5/sm/pn}
           {
             \node[\st] at (\x, \y) (\n) {};
           }
         \end{scope}

         \begin{scope}[xshift=0cm, yshift=-3cm]
    
           \foreach  \x / \y / \n / \st in
           { -2.5/0/s66n/pn, -2/0/s6n1/pn, -1/0/snhp/pnn, 0/0/s9nm1/pn, .5/0/s9n/pn, 1/0/smpp/pnn, 1.5/0/s10nm1/pn, 2/0/s10n/pn}
           {
             \node[\st] at (\x, \y) (\n) {};
           }
         \end{scope}

         \foreach \n / \nn / \p / \l / \r  / \c in
         { x1/{x_1}/above left/1/-1/black, x2/{x_2}/above left/1/-1/black, x3n/{x_{3\enn-1}}/above left/0/-1/black,
           a1/{a_1}/below/1/0/black, a2/{a_2}/below/0/0/black, a3n/{a_{3\enn}}/below/1/0/black, 
            c1/{c_0}/above right/0/0/black, c2/{c_2}/above right/0/0/black, cm/{c_{3\enn-1}}/above right/0/0/black,
             t0/{t_0}/above right/0/0/black, t1/{t_1}/above right/0/0/black, tm/{t_{3\enn-1}}/above right/0/0/black,
             s3n/{s_{3\enn}}/above left/0/0/black, s3n1/{s_{3\enn+1}}/above right/0/0/black, s6n/{s_{6\enn-1}}/above left/0/0/black, 
            s0/{s_{0}}/above/-1/0/black, s1/{s_{1}}/above right/0/0/black, sm/{s_{3\enn-1}}/above right/0/0/black,
            s66n/{s_{6\enn}}/left/0/0/black, s6n1/{s_{6\enn+1}}/below right/0/0/black,
           s10n/{s_{10\enn}}/below/0/0/black, %
            s9nm1/{s_{9\enn-1}}/below/0/0/black}{
           \node[\p = \l pt and \r pt of \n, text=\c, inner sep=.5pt, fill=white] {\small ${\nn}$};
         }

         \node[xshift=-5pt] at (smpp) {$\ldots$};
         \node[xshift=5pt] at (snhp) {$\ldots$};
         \node[rotate=90] at (cmp) {$\ldots$};
         \node[rotate=90,xshift=-1pt] at (a3np) {$\ldots$};
         \node[rotate=90] at (x3np) {$\ldots$};
         \node[rotate=90,xshift=5pt] at (s3np) {$\ldots$};
         \node[rotate=90,xshift=5pt] at (tmp) {$\ldots$};
         \node[rotate=90,xshift=5pt] at (snp) {$\ldots$};

         \begin{pgfonlayer}{bg}
           \foreach \s / \t / \aa / \type in {
             s3n/a1/-5/fc, s3n1/a2/-5/fc,  
             s6n/a3n/-5/fc,
             s3n/s3n1/0/fc, s3np/s6n/0/fc,
             s3n/x1/0/fc, x1/s3n1/0/fc, s3n1/x2/0/fc, x2/s3np/0/fc, x1/s66n/0/fc, %
              x3n/s6n/0/fc,
             s6n/s66n/0/fc, s66n/s6n1/0/fc, s6n1/snhp/0/fc,
             x3n/s9nm1/0/fc, s9nm1/s9n/0/fc, smpp/s10nm1/0/fc, s10nm1/s10n/0/fc,
             t0/t1/0/fc, 
             tmp/tm/0/fc,  
              s0/s1/0/fc, %
             snp/sm/0/fc,
             c1/t0/0/fc,c2/t1/0/fc,cm/tm/0/fc} {
             \draw[->, \type] (\s) edge[bend right = \aa] (\t);
           }

            \foreach \i in {0,1,m} { 
             \draw[->, fc] (t\i) edge[] (s\i);
           }

			\draw[yellow, very thick, rounded corners] (sm.south) .. controls ($(sm)+(0, -0.4)$) and ($(sm)+(1.2, -0.4)$) .. ($(sm)+(1.2, 0.2)$) .. controls ($(s0)+(1.2, 0.5)$) and ($(s0)+(1.2, .8)$)   .. (s3n.east);

           \draw[->, fc, rounded corners] (sm.south) .. controls ($(sm)+(0, -0.4)$) and ($(sm)+(1.2, -0.4)$) .. ($(sm)+(1.2, 0.2)$) .. controls ($(s0)+(1.2, 0.5)$) and ($(s0)+(1.2, .8)$)   .. (s3n.east);
           
            \draw[->, fc, rounded corners] (s10n.west) .. controls ($(s10n)+(2, 0)$) and ($(s0)+(2, 0)$)   .. (s0.east);
            
            \draw[yellow, very thick, rounded corners] (tm.south) .. controls ($(tm)+(0, -0.4)$) and ($(tm)+(-0.4, -0.4)$) .. ($(tm)+(-0.4, 0.2)$) .. controls ($(t0)+(-0.4, 0.2)$) and ($(t0)+(-0.4, .4)$)   .. (t0.west);
             
             \draw[->, fc, rounded corners] (tm.south) .. controls ($(tm)+(0, -0.4)$) and ($(tm)+(-0.4, -0.4)$) .. ($(tm)+(-0.4, 0.2)$) .. controls ($(t0)+(-0.4, 0.2)$) and ($(t0)+(-0.4, .4)$)   .. (t0.west);
              
             \draw[->, fc, rounded corners] (x2.west) .. controls ($(x2)+(-0.4, 0)$) and ($(x2)+(-0.4, -0.4)$) .. ($(x3n)+( -0.4, -0.4)$) .. controls ($(s66n)+(0, -0.4)$) and ($(s6n1)+(-0.4, -0.4)$)   .. (s6n1.south);

           \draw[->] (a1) edge (c2);
           \draw[->] (a1) edge ($(c2)!0.5!(cmp)$);
           \draw[->] (a1) edge ($(c2)!0.9!(cmp)$);

           \draw[->] (a2) edge (c1);
           \draw[->] (a2) edge (cmp);
           
           \draw[->] (a3n) edge (cmp);
           \draw[->] (a3n) edge ($(cmp)!0.5!(cm)$);
           \draw[->] (a3n) edge (cm);
         \end{pgfonlayer}
     
       \end{tikzpicture}\caption{The construction in \cref{thm:sverif_cont_fas+delta} where \fas = 2, the highlighted arcs are the feedback arcs.}   \label{fig:sverif_const_fas+delta}
     \end{figure}

\begin{itemize}[--]
    \item For each element $i\in [3\enn]$, create an element agent $a_i$.
    \item For each set $C_j\in \mathcal{C}$, create a set agent $c_j$ and a helper agent $t_j$; note that $j$ starts with $0$.
    \item Create $10\enn$ special agents $\{s_0,s_1, \dots, s_{10\enn}\}$. %
    \item Create $3\enn-1$ dummy agents $\{x_1, \dots, x_{3\enn -1}\}$.
\end{itemize}
\noindent The friendship graph has the following arcs. %

\noindent $E(\goodG)=\{(s_{z},s_{z+1}) \mid z\in \{0\, \dots, 10\enn\}\}$
    $\bigcup \{(a_i,c_j) \mid$ for each~$i\in [3\enn]$ and $C_j\in \mathcal{C}$ with $i\in C_j\} \bigcup$
    $\{(s_{3\enn-1+z},x_z),(x_z,s_{3\enn +z}),(x_z,$ $s_{6\enn-1 +z}) \mid$ $z \in [3\enn -1]\}$
    $\bigcup \{ (t_l,t_{l+1}), \{(c_l,t_l), (t_l,s_l) \mid l\in \{0, \dots, 3\enn-1 \}\}
    \bigcup \{ (s_{3\enn-1 +i},a_i)\mid i\in [3\enn ]\}$,
    where $z+1$ is taken modulo $10\enn$ and $l+1$ is taken modulo $3\enn$.

Before specifying the initial partition, we analyze the maximum degree and the feedback arc set number.
First, we observe that the friendship graph has two feedback arcs:
Delete the arcs $(s_{3\enn-1},s_{3\enn}),(t_{3\enn-1},t_0)$.
Then, the order $(s_{3\enn},x_1,s_{3\enn +1},x_2, \dots, s_{6\enn -2},$ $x_{3\enn -1},s_{6\enn-1},
\dots, s_{10\enn},
a_1, \dots, a_{3\enn },c_0, \dots, c_{3\enn-1}, t_0,\dots, t_{3\enn-1}, s_0,$ $\dots,$ $s_{3\enn - 1})$ is a topological order of the vertices, so the remaining digraph is acyclic.

It is also straightforward to verify the following.
\begin{compactitem}[--]
  \item For every $i \in [3\enn]$, the element agent~$a_i$ has three friends and is a friend to the agent~$s_{3\enn-1+i}$.
  \item For every set $C_j \in \mathcal{C}$, the set agent~$c_j$ has one friend and is a friend to exactly three element agents.
  \item For every $z \in \{0, \dots, 3\enn-1, 6\enn, \dots,10\enn \}$ the agent $s_z$
  has one friend and is a friend to at most two agents, while the agent~$s_{6\enn-1}$ has two friends and is a friend of two agents. 
  \item For every $l \in [3\enn - 1]$, $x_l$ has two friends and is a friend to one agent. %
  \item For every $z \in \{3\enn, \dots, 6\enn - 2\}$, $s_z$ has three friends and is a friend to two agents. 
\end{compactitem}
Summarizing, every agent has at most five in- and out-neighbors in the friendship graph.

To complete the construction, we define the initial partition~$\Pi \coloneqq \{\{s_0, \dots, s_{10\enn}, x_1, \dots, x_{3\enn-1}\}\} \cup \{\{t_0, \dots, t_{3\enn-1}\}\} \cup \{\{c_j\} 
\mid j \in \{0,$ $ \dots, 3\enn - 1\}\} \cup \{\{a_i\} \mid i \in [3\enn]\}$.

We observe the following. %
\begin{observation} \label{obs:fas+deg}
   It holds that in the initial partition~$\Pi$, 
   \begin{compactenum}[(i)] 
  \item \label{obs:fas+deg,1i} For every $i \in [3\enn], C_j \in \mathcal{C}$, agents~$a_i$ and $c_j$ have each no friends and no enemies,
  \item \label{obs:fas+deg,1ii}For every $z \in \{3\enn, \dots, 6\enn - 2\}, l \in [3\enn - 1]$ agents $s_z$ and $x_l$ both have two friends and $13\enn -3$ enemies,
  \item  \label{obs:fas+deg,1iii}For every $z \in \{0, \dots, 3\enn - 1, 6\enn - 1, \dots, 10\enn\}$ agent $s_z$ has one friend and $13\enn -2$ enemies, and 
  \item \label{obs:fas+deg,1iv} For every $l\in \{0, \dots, 3\enn-1\}$, $t_l$
  has one friend and $3\enn -2$ enemies. 
  \end{compactenum}
\end{observation}

It remains to show the correctness, i.e., $I$ admits an exact cover if and only if $\Pi$ is not strictly core stable, i.e., $\Pi$ is weakly blocked by some coalition.
The ``only if'' part is shown by the following.
\begin{clm}
\label{claim:fas+deg2}
If $I$ admits an exact cover, then $\Pi$ is weakly blocked by some coalition.
\end{clm}
\begin{proof}[Proof of claim \ref{claim:fas+deg2}]
\renewcommand{\qedsymbol}{$\diamond$}
Let $\mathcal{K}$ be an exact cover.
Then, let $P$ be the coalition we get by adding all $3\enn $ element agents, the set agents corresponding to~$\mathcal{K}$, and the agents from~$\{t_0,t_1,\dots ,t_{3\enn-1},s_0 ,\dots,$ $ s_{6\enn-1}\}$. Formally, $P \coloneqq \{a_i \mid i \in [3\enn]\} \cup \{c_j \mid C_j \in \mathcal{K}\} \cup \{t_l \mid l \in \{0, \dots, 3\enn - 1\}\} \cup  \{s_z \mid z \in \{0, \dots, 6\enn - 1\}\}$.
Then, every set agent~$c_j$, $C_j \in \mathcal{K}$ and every element agent~$a_i, i \in [3\enn]$ have one friend in $P$, so they strictly improve by \cref{obs:fas+deg}\eqref{obs:fas+deg,1i}.
Every helper agent~$t_l, l \in \{0, \dots, 3\enn - 1\}$ obtains two friends, so she also strictly improves by \cref{obs:fas+deg}\eqref{obs:fas+deg,1iv}.
Finally, each agent~$s_z$, $z\in \{0, \dots, 6\enn - 1\}$ has the same number of friends in $\Pi$ and $P$, and since the size of $P$ is $3\enn +\enn+3\enn +6\enn=13\enn $, each has the same number of enemies, so they weakly improve.
\end{proof}

Now, we turn to the ``if'' part, and let $P$ be a weakly blocking coalition of~$\Pi$.
To improve readability, in the following, for each index~$z\in \{0, \dots, 10\enn\}$, we use \myemph{$\seqq{z,6\enn-1}$} to denote the following set~$\seqq{z,6\enn-1}\coloneqq \{z,\dots,6\enn-1\}$ if $z\le 6\enn-1$, and~$\seqq{z,6\enn-1}\coloneqq \{z,\dots,10\enn, 0,\dots,6\enn-1\}$ otherwise.
We first observe the following.
\begin{clm}
\label{claim:fas+deg-3} %
\begin{compactenum}[(i)]
  \item \label{fas+deg-3,ii}
  If $s_z\in P$ for some $z\in \{0, \dots, 10\enn - 1\}$,
  then
  $s_{z'}\in P$ for all $z'\in \seqq{z,6\enn-1}$.
  \item \label{fas+deg-3,i}
  If $t_l\in P$ for some $l\in \{0, \dots, 3\enn-1\}$, then $t_{l'}\in P$ for all $l'\in \{0, \dots, 3\enn-1\}$. %
 
  \item \label{fas+deg-3,iii} $\{x_1,\dots ,x_{3\enn -1}\} \cap P = \emptyset$.
  \end{compactenum}
\end{clm}
\begin{proof}[Proof of \cref{claim:fas+deg-3}]\renewcommand{\qedsymbol}{$\diamond$}
  For (\ref{fas+deg-3,ii}), assume that $s_{z}\in P$ for some $z\in \{0, \dots, 10\enn\}$.
  We only need to consider the case when $z\neq 6\enn-1$, as $\seqq{z, 6\enn-1}=\{6\enn-1\}$ if $z=6\enn-1$.
  If $z\in \{3\enn,\dots,6\enn -2\}$,
  then $s_z$ has three friends $a_{z-3\enn+1}$, $s_{z+1}$ and $x_{z-3\enn+1}$ in the friendship graph,
  from which she has two in~$\Pi$.
  Therefore, she must have at least two friends in~$P$.
  That is, $s_{z+1}\in P$ or $x_{z-3\enn+1}\in P$.
  Note that from the latter, since $x_{z-3\enn+1}$ has her only two friends in $\Pi$,
  so she does in $P$, implying that $s_{z+1}\in P$.
  By applying the above recursively, we obtain that $\{s_{z+1}, \dots,s_{6\enn-1}\} = \seqq{z, 6\enn - 1}\subseteq P$. 

  If $z\notin\{3\enn,\dots ,6\enn-1\}$, i.e., $z\in \{0,\dots,3\enn-1,6\enn,\dots,10\enn\}$,
  then $s_z$ has only one friend, namely $s_{z+1}$ that she has in $\Pi$ ($z+1$ is taken modulo $10\enn$), so $s_{z+1}\in P$.
  By applying the above recursively, we obtain that $\{s_{z+1}, \dots, s_{10\enn}, s_{0}, \dots, s_{3\enn}\}\subseteq P$.
  Together with the first case, we further obtain that  $\{s_{z+1}, \dots, s_{10\enn}, s_{0}, \dots, s_{6\enn-1}\} = \seqq{z, 6\enn - 1} \subseteq P$, as desired.

  For \eqref{fas+deg-3,i}, suppose, towards a contradiction, that there exists an index~$l\in\{0, \dots, 3\enn-1\}$ such that $t_l\in P$ but $t_{l+1}\notin P$. Throughout this paragraph we take $l + 1$ modulo $3\enn$.
  By \cref{obs:fas+deg}\eqref{obs:fas+deg,1iv}, agent~$t_l$ must get a friend in $P$, so $s_l\in P$.
  By \eqref{fas+deg-3,ii}, we have that $ \{s_{l+1},\dots,s_{6\enn-1} \} \subseteq P$.
  As $l\le 3\enn-1$, we obtain that in $P$, agent~$t_l$ has at least $3\enn $ enemies, so by \cref{obs:fas+deg}\eqref{obs:fas+deg,1iv} she can only improve if she gets another friend, so $t_{l+1}\in P$, contradiction. 
  Hence, we get that $t_l\in P$ implies $t_{l+1}\in P$, so $t_{l'}\in P$ for all $l'\in \{0, \dots, 3\enn-1\}$.

  To prove \eqref{fas+deg-3,iii}, suppose, towards a contradiction, that $x_l\in P$ for some~$l\in [3\enn-1]$.
By \cref{obs:fas+deg}\eqref{obs:fas+deg,1ii}, it implies that both of her friends $s_{3\enn+i}$ and $s_{6\enn +l-1}$ are in $P$.
By Statement~\eqref{fas+deg-3,ii} and by previous reasoning, we have that $\{s_{6\enn +l-1},\dots, s_{10\enn}, s_0,\dots,s_{6\enn-1}\}\subseteq P$.
Every $s_z$, $z \in \{3\enn,\dots ,6\enn-1\}$ must receive two friends in $P$ by \cref{obs:fas+deg}\eqref{obs:fas+deg,1ii}-\eqref{obs:fas+deg,1iii}. As no two of them have a mutual friend, they are all different.
We distinguish between two cases.

Case I: Suppose that none of the element agents $a_{i}, i \in \{0, \dots,$ $ 3\enn - 1\}$ are in $P$.
Therefore, for every~$z\in \{3\enn,\dots ,6\enn -2\}$, $s_z$ has the dummy~$x_z$ agent as a friend, and $s_{6\enn-1}$ receives $s_{6\enn}$ as friend.
By \cref{obs:fas+deg}\eqref{fas+deg-3,i}, $\{s_0, \dots, s_{10\enn}, x_1, \dots, x_{3\enn - 1}\} \subseteq P$.
No other agent can be inside $P$, as otherwise the size of $P$ would be more than $13\enn $ and $x_l$ would have more enemies, contradiction. But then, none of the agents in $P$ strictly improves, contradiction.

Case II: Suppose $a_{i}\in P$ for some ${i}\in [3\enn]$.
Then, by \cref{obs:fas+deg}\eqref{obs:fas+deg,1i}, $a_{i}$ must get a friend, so there is some $c_j \in P$ such that $i \in C_j$ and again by \cref{obs:fas+deg}\eqref{obs:fas+deg,1i}, agent~$c_j$ also must get a friend, so $t_j\in P$.
By \eqref{fas+deg-3,ii}, $\{t_0, \dots, t_{3\enn - 1}\} \subset P$. Therefore, the size of $P$ is at least $3\enn +3\enn +(7\enn +2)+1=13\enn + 3>13\enn $ (there are $3\enn$ agents from $\{t_0, \dots, t_{3\enn - 1}\}$,
at least $3\enn$ distinct friends for $\{s_{3\enn},\dots s_{6\enn-1}\}$,
at least $7\enn +2$ special agents,
and the set agent~$c_{j}$) so $x_l$ has more enemies in $P$, a contradiction. 
\end{proof}

The ``if'' part follows from the following claim.
\begin{clm}
\label{claim:fas+deg4}
If $\Pi$ is weakly blocked by some coalition, then  $I$ admits an exact cover.
\end{clm}
\begin{proof}[Proof of claim \ref{claim:fas+deg4}] \renewcommand{\qedsymbol}{$\diamond$}
Assume that $P$ is weakly blocking~$\Pi$. 
We can assume that $P$ is an inclusionwise minimal blocking coalition, therefore the graph induced by $P$ is strongly connected. 
This means that $s_{3\enn-1},s_{3\enn}$ or  $t_0,t_{3\enn-1}$ must be in~$P$.

First, if $t_0\in P$, then by \cref{claim:fas+deg-3}\eqref{fas+deg-3,i}, we have that $t_l\in P$ for all~$l\in \{0, \dots, 3\enn -1\}$. As there has to be an agent who strictly improves, there must be more agents in $P$ than only $\{t_l \mid l \in \{0 , \dots, 3\enn - 1\}\}$. For every $t_l, l \in \{0 , \dots, 3\enn - 1\}$, the initial coalition is smaller than $\Pi(t_l)$. Thus $t_l$ must gain a friend to weakly improve. Particularly, $s_0\in P$ and by claim \ref{claim:fas+deg-3}\eqref{fas+deg-3,ii}, $\{s_0,s_1,\dots, s_{6\enn-1}\} \subset P$. 

By \cref{obs:fas+deg}\eqref{obs:fas+deg,1ii}, each~$s_z$, $z\in \{ 3\enn,\dots, 6\enn-2\}$, must have a friend other than $s_{z+1}$ and by \cref{claim:fas+deg-3}\eqref{fas+deg-3,iii}, it cannot be $x_{z - 3\enn + 1}$.
Therefore, we get that $a_i\in P$ for every $i\in [3\enn-1]$.
We claim that $a_{3\enn}$ is also in $P$ because of the following: If $a_{3\enn} \notin P$, then by \cref{obs:fas+deg}\eqref{obs:fas+deg,1ii}, agent~$s_{6\enn-1}$ must get her other friend~$s_{6\enn}$,
 and by \cref{claim:fas+deg-3}\eqref{fas+deg-3,ii}, we would obtain that $\{s_z \mid z\in \{0, \dots, 10\enn\}\} \subseteq P$.
This implies that $|P|\ge 10\enn+1+3\enn+3\enn-1 > |\Pi(s_0)|$ so $s_0$ would have more enemies, a contradiction.
Hence, $\{a_i\mid i\in [3\enn]\}\subseteq P$. 
By \cref{obs:fas+deg}\eqref{obs:fas+deg,1i}, each $a_i, i \in [3\enn]$ must get a friend, so there are at least $\enn$ set agents~$c_j, C_j \in \mathcal{C}$ in $P$, as each of them is a friend to three element agents.
Suppose the number of set agents in $P$ is more than $\enn$.
Then $|P|\ge 3\enn +(\enn +1)+3\enn +6\enn =13\enn +1$, so $s_0$ would have more enemies, contradiction. Therefore, the number of set agents in $P$ has to be $\enn$. As each $a_i, i \in [3\enn]$ has a set agent friend, we obtain that the set agents in $P$ correspond to an exact cover.

Second, if $\{s_{3\enn-1},s_{3\enn}\}\subseteq P$, then by \cref{claim:fas+deg-3}\eqref{fas+deg-3,ii}, $\{s_{3\enn +1},\dots,$ $ s_{6\enn -1}\}\subset P$.
By \cref{obs:fas+deg}\eqref{obs:fas+deg,1ii}, every~$s_z, z\in \{ 3\enn,\dots, 6\enn -2\}$, must have a friend other than $s_{z+1}$ and by \cref{claim:fas+deg-3}\eqref{fas+deg-3,iii}, it cannot be $x_{z - 3\enn + 1}$.
Therefore, we get that $a_i\in P$ for all $i\in [3\enn-1]$.
Each $a_i, i \in [3\enn - 1]$ must get a friend, so there is some $c_j\in P, j \in \{0, \dots, 3\enn - 1\}$.
Agent~$c_j$ also must get a friend, so $t_j\in P$ and by \cref{claim:fas+deg-3}\eqref{fas+deg-3,i}, we get that $t_0,t_{3\enn-1}\in P$. We have already shown that if $t_0,t_{3\enn-1}\in P$, then $I$ admits an exact cover, so this concludes the proof.
\end{proof}
The proof for \FE-\verif\ is analogous with the only difference that we add an agent $s_{10\enn +1}$ in the cycle of special agents. We delete the friendship arc $(s_{10\enn},s_{0})$ and add the arcs $(s_{10\enn},s_{10\enn +1})$ and $(s_{10\enn +1},s_0)$. %

We can check that the same coalition as in \cref{claim:fas+deg2} is a strictly blocking coalition now, as even the agents $\{s_0,..,s_{6\enn-1}\}$ strictly improve, because they get less enemies.

The proof of \cref{claim:fas+deg-3} is similar, but with the additional agent~$s_{10\enn+1}$.
Moreover, if we change \cref{claim:fas+deg4} by replacing weakly blocking with strictly blocking, then the same proof shows that the existence of a strictly (and thus weakly) blocking coalition implies the existence of an exact cover. The only modification needed is that now to strictly improve $s_0$ must get strictly less enemies in the coalition, hence its size must be strictly less than $13\enn +1$ (because of the addition of $s_{10\enn + 1}$, the initial coalition containing $s_0$ has size $13\enn +1$). We have that at most $\enn$ set agents can be in $P$, as required.
\end{proof}
}

Next, we observe that checking whether a specific strictly blocking coalition exists can be done in linear time. 
This result will be useful for designing further algorithms.
\newcommand{\lemstrictlymore}{%
  Given a coalition structure~$\Pi$ with maximum coalition size~$\maxcoal$, in linear time, we can either find a blocking coalition where every agent obtains strictly more friends than in $\Pi$, or conclude that each weakly (resp.\ strictly) blocking coalition has size at most $\maxcoal$.
}
\begin{lemma}[\appendixsymb]
\label{lemma:preprocess}
\lemstrictlymore
\end{lemma}

\begin{proof}[Proof sketch.]
  Call a coalition a \myemph{wonderfully blocking coalition} if every agent in it has strictly more friends than in $\Pi$.
  We observe that if no coalitions are wonderfully blocking,
  then in any blocking (resp.\ weakly blocking) coalition~$U'$,
  there is an agent who has the same number of friends, so she cannot get more enemies than in $\Pi$, implying that $|U'|\le \maxcoal$. %
  Hence, checking whether wonderfully blocking coalitions exist in the desired time completes the proof. 

  For each agent~$v\in V$, let $f_{\Pi}(v)$ denote the number of friends she has in $\Pi$, i.e.,
  $f_{\Pi}(v) = |N_{\goodG}^+(v) \cap \Pi(v)|$, and let $r(v)=f_{\Pi}(v)+1$.
  Let $U$ be a hypothetical wonderfully blocking coalition.
  Then, each agent~$v\in U$ needs at least $r(v)$ friends in~$U$.
  Now, if there are agents in the input with out-degree less than $r(v)$, then we delete them since they cannot be included in $U$. 
  Then, we recursively delete the agents~$v$ that have less than $r(v)$ out-neighbors in the resulting friendship graph.
  We repeat this process as long as there is an agent~$v$ with out-degree less than $r(v)$.
  If this procedure terminates with some agents remaining,
  then they form a wonderfully blocking coalition; otherwise, there can be none.
  \appendixalg{lemma:preprocess}{%
    \ifshort
    The correctness proof and the running time are deferred to the full version.
    \else
    The correctness proof and the running time are deferred to the appendix.
  \fi}{\lemstrictlymore}{
    The correctness proof is as follows:
    That the remaining agents form a wonderfully blocking coalition is straightforward  since they have at least $r(v)= f_{\Pi}(v)+1$ friends from the remaining agents.
    If no agent remains, then no wonderfully coalition exists since a wonderfully blocking coalition induces a subgraph, where an agent $v$ has out degree at least $r(v)$.
    So we would not have deleted any of these agents in the above process.
    
    Initially computing $r(v)$ and $|N^+_{\goodG}(v)|$ for every $v \in V$ takes $O(n + m)$ time. When recursively deleting agents, we touch any edge only when we delete it. For each edge deletion, updating the number of neighbors and checking whether an agent and should be removed can be done in a constant time. Thus the number of operations of the deletion process is in $O(m)$. Therefore the time complexity of the algorithm is $O(n + m)$.

    }
\end{proof}

Based on \cref{lemma:preprocess}, core verification is polynomial-time solvable if the largest initial coalition has bounded size.
However, this result cannot be improved to obtain fixed-parameter tractability. 
\newcommand{\fecorekwhard}{
  \FE-\verif\ and \FE-\sverif\ are in \xp\ and \cowone-hard wrt.\ $\maxcoal$; hardness remains even for symmetric preferences.}

\begin{theorem}[\appendixsymb]\label{thm:W1h+XP_largest_coalition}
\fecorekwhard
\end{theorem}

\appendixproofwithstatement{thm:W1h+XP_largest_coalition}{\fecorekwhard}{
  \begin{proof}
  We first show the \xp\ result.
  Given an instance $I=(G=(V,E), \Pi)$ of \FE-(\textsc{strict}) \verif\~, where the largest coalition in $\Pi$ has size $\maxcoal$, we provide an algorithm
  to determine whether $\Pi$ is in the (strict) core and show that it runs in $O(|V|^{\maxcoal}\cdot |E|)$ time. %
  The algorithm consists of two steps. 
 In the first step, we check whether there is a strictly blocking coalition, where each agent get strictly more friends than in $\Pi$, with the procedure from Lemma \ref{lemma:preprocess}. 

 If we found such a coalition, then we stop and conclude $\Pi$ is not a core stable (and hence not strictly core stable) solution.
 Otherwise, a blocking coalition can have size at most $\maxcoal$, because $\Pi$ can only be blocked by some coalition~$B$, in which at least one agent~$x$ has the same number friends in $\Pi$ as in $B$.
 By the blocking coalition definition, agent~$x$ must have no more enemies in~$B$ than in $\Pi$.
 That is, $|B|\le |\Pi(x)|\le \maxcoal$.
 Hence, in the next step we check all $O(|V|^{\maxcoal})$ possible coalitions with at most $\maxcoal$ agents, and check whether they are (strictly) blocking. If there is a weakly (resp.\ strictly) blocking coalition, then $\Pi$ is not  strictly core (resp.\ core) stable, otherwise it must be.

 The running time of the second step is bounded by $O(|V|^{\maxcoal}\cdot |E|)$, since deciding whether a given coalition is (strictly) blocking can be done in $O(|E|)$ time.

 Now, we turn to the hardness result. We reduce from the following problem:

 \decprob{\clique\ }
{A graph $\hG = (\hV, \hE)$, an integer $h$.}{Does~$\hG$ contain a clique of size $h$, i.e., a $h$-vertex subgraph of $\hG$, such that every pair of vertices is adjacent?}

Given an instance $(\hG = (\hV, \hE), h)$ of \clique, we construct an instance of  \FE-\verif\ with $(h-1)|\hV|+ ( h + \binom{h}{2} + 1)|\hE|$ agents.

     \begin{figure}
     \centering
       \begin{tikzpicture}[scale=1,every node/.style={scale=0.9}, >=stealth', shorten <= 2pt, shorten >= 2pt]

      \foreach \x / \y / \n / \nn / \typ / \p / \dx / \dy in {
      0/3/ui/u_i/pn/{below left}/-1/-4,
      -1/4/a1i/a^1_i/pn/{below left}/-1/-4,
      -0.7/5/a2i/a^2_i/pn/{above left}/-1/-4,
      0.7/5/invi/\dots/pnn/{above right}/-1/-4,
      1/4/ah2i/a^{h-2}_i/pn/{below right}/-1/-4,
      4/3/uj/u_j/pn/{right}/1/1,
      3/4/a1j/a^1_j/pn/{below left}/-1/-4,
      3.3/5/a2j/a^2_j/pn/{above left}/-1/-4,
      4.7/5/invj/\dots/pnn/{above right}/-1/-4,
      5/4/ah2j/a^{h-2}_j/pn/{below right}/-1/-4,
      4/1.5/h1/\;/pnn/{above right}/-1/-4,
      5/1.5/h2/\;/pnn/{above right}/-1/-4,
      0.2/1.5/h3/\;/pnn/{above right}/-1/-4,
      2/1/f/f_{\{i,j\}}/pn/{right}/1/1,
      1.3/0/b1/b^1_{\{i,j\}}/pn/{below left}/1/1,
      2.7/0/b2/b^2_{\{i,j\}}/pn/{below right}/-1/1,
      1/-1/b3/b^3_{\{i,j\}}/pn/{below left}/-1/-4,
      2/-1.7/invb/\dots/pnn/{below right}/-1/-4,
      3/-1/bh2/b^{h + \binom{h}{2}}_{\{i,j\}}/pn/{below right}/-1/-4} {
        \node[\typ] at (\x,\y) (\n) {};
        \node[\p = \dx pt and \dy pt of \n] {$\nn$};
      }

      \begin{pgfonlayer}{bg}
           \foreach \s / \t / \aa / \typ in {
           ui/f/0/fc,uj/f/0/fc,ui/h3/0/fc,uj/h1/0/fc,uj/h2/0/fc,
           ui/a1i/0/fc,ui/a2i/0/fc,ui/invi/0/fc,ui/ah2i/0/fc,a1i/a2i/0/fc,a1i/invi/0/fc,a1i/ah2i/0/fc,
           a2i/invi/0/fc,a2i/ah2i/0/fc,ah2i/invi/0/fc,
           uj/a1j/0/fc,uj/a2j/0/fc,uj/invj/0/fc,uj/ah2j/0/fc,a1j/a2j/0/fc,a1j/invj/0/fc,a1j/ah2j/0/fc,
           a2j/invj/0/fc,a2j/ah2j/0/fc,ah2j/invj/0/fc,
           f/b1/0/fc,f/b2/0/fc,b1/b2/0/fc,b1/invb/0/fc,b1/bh2/0/fc,
           b2/invb/0/fc,b2/bh2/0/fc,bh2/invb/0/fc,
           b3/b1/0/fc,b3/b2/0/fc,b3/invb/0/fc,b3/bh2/0/fc} {
             \draw[-, \typ] (\s) edge[bend right = \aa] (\t);
           }

         \end{pgfonlayer}
       \end{tikzpicture}\caption{Illustration for the proof of \cref{thm:W1h+XP_largest_coalition}, depicting the reduction corresponding to vertices $i$, $j$ and the edge $\{i, j\}$ between them.}\label{fig:W1h+XP_largest_coalition}
     \end{figure}
     
For each vertex $i \in \hV$ we add a vertex agent $u_i$, and for each $j \in [h - 2]$ a private agent $a^j_i$.
For each edge $e \in \hE$ we add an edge agent $f_e$, and for each $j \in [h + \binom{h}{2}]$ a private agent $b^j_e$.

We add the following friendship-edges for every $i \in \hV$:
\begin{compactitem}[--]
\item for every $j \in [h - 2]$, add an edge $\{a^{j}_i, u_i\},$
\item for every $j, j' \in [h - 2], j < j'$, add an edge $\{a^j_i, a^{j'}_i\} $
\end{compactitem} 
and for every $e = \{i, j\} \in \hE$:
\begin{compactitem}[--]
\item for every $\ell \in [2]$, add an edge $\{f_e, b^{\ell}_e\}$,
\item for every $j, j' \in [h + \binom{h}{2}], j < j'$, add an edge $\{b^j_e, b^{j'}_e\}$
\item add the edges $\{u_i, f_e\}$, $\{u_j, f_e\}$.
\end{compactitem}
The construction is depicted in \cref{fig:W1h+XP_largest_coalition}.

We define the initial coalition $\Pi = \{\{ u_i, a^1_i, \dots a^{h - 2}_i  \} \mid i \in V \} \cup \{ \{ f_e, b^1_e, \dots, b^{h + \binom{h}{2}}_e \} \mid e \in E  \}$. The size of a largest coalition in $\Pi$ is $h + \binom{h}{2} + 1$, which is clearly a function of $h$.

Observe the following facts about our construction:
\begin{observation}In $\Pi$, the following holds:\label{obvs:init_coal_friends}
\begin{compactenum}[(i)]
\item For every $i \in \hV$, $u_i$ has $h - 2$ friends and $0$ enemies in $\Pi$. \label{obvs:init_coal_friends1}
\item For every $i \in \hV, j \in [h - 2]$, $a^j_i$ has $h - 2$ friends and $0$ enemies in $\Pi$.\label{obvs:init_coal_friends2}
\item For every $e \in \hE$, $f_e$ has $2$ friends and $h + \binom{h}{2} -2$ enemies in $\Pi$.\label{obvs:init_coal_friends3}
\item For every $e \in \hE, j \in [2]$, $b^j_i$ has $h + \binom{h}{2}$ friends and $0$ enemies in $\Pi$. For every $e \in \hE, j \in \{3, \dots, h + \binom{h}{2}\}$, $b^j_i$ has $h + \binom{h}{2}  - 1$ friends and $1$ enemy in $\Pi$.\label{obvs:init_coal_friends4}
\end{compactenum}
\end{observation}

We now show that the reduced instance is equivalent to $I$.

\begin{clm}\label{claim:init_coal_forward}
If $\hG$ has a clique of size $h$, then $\Pi$ is not core stable.
\end{clm}

\begin{proof}[Proof of \cref{claim:init_coal_forward}]
\renewcommand{\qedsymbol}{$\diamond$}
Assume $\hG$ has a clique induced by $K \subseteq \hV$ with $|K| = h$. Then we show that the agents $K' = \{u_i \mid i \in K\} \cup \{f_{\{i, j\}} \mid \{i, j\} \subseteq K\}$ form a blocking coalition. Note that because $K$ is a clique, $f_{\{i, j\}}$ exists for every pair $\{i, j\} \subseteq K$. Also, $|K'| = h + \binom{h}{2}$.

Observe that:
\begin{itemize}
\item For every $i \in K$, $u_i$ has $h - 1$ friends in $K'$, namely $f_{\{i,j\}}$ for each $j \in K\setminus \{i\}$, and thus improves from $\Pi$ by \Cref{obvs:init_coal_friends}\eqref{obvs:init_coal_friends1}.
\item For every $\{i, j\} \subseteq K$, $f_{\{i, j\}}$ has $2$ friends, namely $u_i$ and $u_j$, and $(h - 2) + (\binom{h}{2} - 1) = h + \binom{h}{2} - 3$ enemies in $K'$ and thus strictly improves from $\Pi$ by \Cref{obvs:init_coal_friends}\eqref{obvs:init_coal_friends3}.
\end{itemize}

As every agent in $K'$ strictly improves, $K'$ is a strongly blocking coalition.
\end{proof}

\begin{clm}\label{claim:init_coal_backward}
If $\Pi$ is not core stable, then $\hG$ has a clique of size $h$.
\end{clm}

\begin{proof}[Proof of~\cref{claim:init_coal_backward}]
\renewcommand{\qedsymbol}{$\diamond$}
Assume $\Pi$ admits a strictly blocking coalition $K'$.
First we show that for every $i \in \hV, j \in [h-2]$, $e \in E$, and $j' \in [h + \binom{h}{2}]$ no $a^j_i$ or $b^{j'}_e$ is in $K'$.

For every $i \in V, j \in [h - 2]$, $a^j_i$ has by \cref{obvs:init_coal_friends}\eqref{obvs:init_coal_friends2} all of her friends and no enemies in $\Pi$. Thus there is no coalition such that $a^j_i$ strictly improves. Similarly for each $e \in E, j \in [2]$, $b^j_e$ also has all her friends and no enemies in $\Pi$ by \cref{obvs:init_coal_friends}\eqref{obvs:init_coal_friends4}.

Assume $b^j_e \in K'$ for some $e \in E, j \in \{3, \dots,  h + \binom{h}{2}\}$. As $b^j_e$ is in the same coalition with all of her friends by \cref{obvs:init_coal_friends}\eqref{obvs:init_coal_friends4}, she must have all of her friends and fewer enemies in $K'$. However, $b^j_e$ is friends with $b^1_e$, who does not join any blocking coalition, a contradiction.

Assume $f_{\{i, j\}} \in K'$, for some  $\{i, j\} \in \hE$. By \cref{obvs:init_coal_friends}\eqref{obvs:init_coal_friends3} we know that $f_{\{i, j\}}$ must have at least two friends in $K'$. Because the private agents of $f_{\{i,j\}}$ do not join any blocking coalition, those friends must be $u_i$ and $u_j$. This implies that since $f_{\{i, j\}}$ can obtain at most the same number of friends as in $\Pi$, she must have fewer enemies in $K'$. Thus $|K'| \leq h + \binom{h}{2}$.

Let $V(K') = \{u_i \mid i \in \hat{V}\} \cap K'$ be the set of vertex agents in $K'$ and $E(K') =  \{f_e \mid e \in \hat{E}\} \cap K'$ the set of edge agents.  Assume $u_i \in V(K')$ for some $i \in \hV$. As $u_i$ has $h - 2$ friends and no enemies in $\Pi$, she must have at least $h - 1$ friends in $K'$, none of which is her private agent. Thus the blocking coalition must contain $h -1$ edge-agents $f_{\{i, j\}}, \{i, j\} \in \hE$. As edge-agents must obtain both their vertex-agent friends, the blocking coalition must also include $h - 1$ agents  $u_j$ such that $\{i, j\} \in \hE$. Therefore  $|V(K')| \geq h$.

We know that each vertex-agent has at least $h - 1$ friends, and those friends must be edge-agents. This implies that $|V(K')| \leq h$ as otherwise by the handshaking lemma we would have $|E(K')| + |V(K')| \geq \frac{(h + 1)(h - 1)}{2} + h + 1 > h + \binom{h}{2}$. Thus $|V(K')| \leq h$. As we already have $|V(K')| \geq h$, she must be that $|V(K')| = h$.

Let $K \coloneqq \{i \in V \mid u_i \in V(K')\}$. We show that $K$ induces a clique on $\hG$. We know that $|K| = |V(K')| = h$. Let $u_i \in V(K')$ be an arbitrary vertex agent in $K'$. We know she is friends with $h - 1$ edge-agents in $K'$ and every edge agent gets both of her friends in $K'$. Thus there must be at least $h - 1$ other vertices in $K$ such that $i$ is adjacent to them. Because $|K| = h$, every vertex in $K$ must be adjacent to $i$. Thus $K$ must induce a clique of size $h$. 
\end{proof}

Thus we have shown that $I$ has a clique of size $h$ if and only if $\Pi$ admits a strictly blocking coalition.\\

To prove the statement for \FE-\sverif\, we remove the agent $b^{h + \binom{h}{2}}_e$ and the associated friendship edges for every $e \in \hE$.
For every $e \in \hE$, agent~$f_e$ now has $2$ friends and $h + \binom{h}{2} - 3$ enemies.

If there is a clique induced by $K \subseteq \hV$,  $|K| = h$, $K' = \{u_i \mid i \in K\} \cup \{f_{\{i, j\}} \mid \{i, j\} \subseteq K\}$ forms a weakly blocking coalition. For every $i \in \hV$, $u_i$ obtains strictly more friends as before and thus improves. For every $e \in \hE$, $f_e$ obtains $2$ friends and $h + \binom{h}{2} - 3$ enemies, so she weakly improves.

If there is a weakly blocking coalition $K'$, no private agent of a vertex agent joins a blocking coalition through identical reasoning. For every $e \in \hE, j \in [2]$, $b^j_e$ does not join a blocking coalition through the identical reasoning as well. For every $e \in \hE, j \in \{3, \dots, h + \binom{h}{2} - 1\}$, $b^j_e$ must obtain all of her friends and at most one enemy. This enemy cannot be $f_e$, because otherwise $K' = \Pi(b^j_e)$. But we must have that $b^1_e$ obtains all of her friends including $f_e$ in $K'$, a contradiction.

To weakly improve, every $f_e$, where $e \in \hE$, must obtain at least $2$ friends and at most $h + \binom{h}{2} - 3$ enemies, just as in the case for \verif. On the other hand, every $u_i$, where $i \in \hV$, must obtain $h - 2$ friends and no enemies to weakly improve. However, $u_i$'s private agents do not join any blocking coalition, so $u_i$ must obtain $h -2$ edge agents. An edge agent does not join a blocking coalition without both of her friends, so there must be $u_j \in K'$, where $\{i, j\} \in \hE$. However, $u_j$ is an enemy to $u_i$, so $u_i$ does not weakly improve. Thus $u_i$ must strictly improve in any blocking coalition, just as in the proof of \cref{claim:init_coal_backward}. As an arbitrary blocking coalition must contain a vertex agent, there is always an agent that strictly improves. The rest of the reasoning is identical to the \verif\-case.
\end{proof}
}

\noindent Combining~$\maxcoal$ with $\maxdeg$, we obtain a fixed-parameter algorithm, based on random separation.

\newcommand{\coremaxcoaldelta}{%
    \FE-\verif and \FE-\sverif are \fpt wrt. $(\maxcoal, \maxdeg)$. %
}
\begin{theorem}[\appendixsymb]\label{thm:fpt_coalition+deg}
 \coremaxcoaldelta 
\end{theorem}

\appendixproofwithstatement{thm:fpt_coalition+deg}{\coremaxcoaldelta}{
\begin{proof}
To prove the theorem we first design an algorithm to find a blocking coalition (if exists) and then prove its correctness. Intuitively, the algorithm is divided into two phases. The first phase is a separation phase. To describe this phase formally, we  define $n$-$p$-$q$-{\em lopsided universal} family~\cite{FominLPS16}.

Given a universe $U$ and an integer $i$, we denote all the $i$-sized subsets of $U$ by ${U \choose i}$. We say that a  family 
$\mathcal{F}$ of sets over a universe $U$ with $\vert U\vert=n$, is an $n$-$p$-$q$-{\em lopsided universal} family if for every $A \in {U \choose p}$ and $B \in {U \setminus A \choose q}$, there is an $F \in \mathcal{F}$ such that $A \subseteq F$ and $B \cap F = \emptyset$. 

\begin{lemma}[\cite{FominLPS16}]
\label{lem:lopsidedUniversal}
There is an algorithm that given $n,p,q\in {\mathbb N}$ constructs an  $n$-$p$-$q$-lopsided universal family $\mathcal{F}$ of cardinality ${p+q \choose p} \cdot 2^{o(p+q)}  \log n$ in time 
$\vert  \mathcal{F} \vert  n$. 
\end{lemma}

Intuitively, first,  the algorithm separates the agents in a blocking coalition from the remaining agents by deleting the non-blocking neighbors of the blocking coalition in $\goodG$. This is achieved  using the $n$-$p$-$(\maxdeg\cdot p)$-{\em lopsided universal} family $\mathcal{F}$. 
 In the second phase, we identify either a  strictly or a weakly blocking coalition (if  exists) for the given core partition or strict core partition, respectively.

\noindent
\textbf{Algorithm.} 
Let $n = |V|$ and $\maxdeg$ %
denote the maximum degree of $\goodG$.
 We write the agents in $V$ as $[n]$ and for $X \subseteq [n]$, we write $\goodG[X]$ to denote the friendship graph induced by $X$. The algorithm is simple. First we check whether there is a (strictly) blocking coalition $S$ of size $|S|>\maxcoal$. If yes, we stop and return $S$, otherwise continue.

 Then, for each integer $1\leq p \leq \maxcoal$, we construct a $n$-$p$-$(\maxdeg \cdot p)$-lopsided universal family $\mathcal{F}$ using the algorithm in \Cref{lem:lopsidedUniversal}. 
 Then, for each $F \in \mathcal{F}$ we check if a (strongly) connected component $C$ in $\goodG[F]$ is (strictly) blocking $\Pi$. If so, we return $C$.

\noindent
\textbf{Correctness.}
Let $S$ be an inclusion-wise minimal hypothetical coalition blocking (or strictly blocking) $\Pi$. 
To show that our algorithm correctly computes such a blocking coalition $S$, we first observe that $|S| \leq \maxcoal$, otherwise we can find such an $S$ in polynomial time by lemma \ref{lemma:preprocess} in the first phase. %
Next we show that there exists $F \in \mathcal{F}$ such that $S$ is a (resp., strongly) connected component in $\goodG[F]$ which completes the proof.

Let the union of the neighbors of the agents of $S$ in $\goodG$ excluding $S$ be $N(S)$, i.e., $N(S) = (\cup_{v \in S} (N^-(v) \cup N^+(v))) \setminus S$. Observe that since $\maxdeg$ is the maximum degree of any agent in $\goodG$ and $\vert S \vert = p$, we have that $\vert N(S)  \vert \leq \maxdeg p$. 
Therefore, from the definition of $\mathcal{F}$,  there exists a set $F \in \mathcal{F}$ such that $S \subseteq F$ and $N(S) \cap F =\emptyset$. We call such a set $F$ a \emph{good} set.

For each set $F \in \mathcal{F}$ 
we define $\goodG[F]$, the friendship graph induced by the agents in $F$.
We know that an inclusion-wise minimal (resp., strictly) blocking coalition must be (resp., strongly) connected. Hence, $S$ is a (resp., strongly) connected component in $\goodG[F]$ for a good $F$ since $N(S) \cap F = \emptyset$. As there exist a good $F \in \mathcal{F}$, the algorithm finds it correctly by enumerating the (resp., strongly) connected components of $\goodG[F]$ and checking if it is  a (resp., strictly) blocking coalition. 

\noindent
\textbf{Running Time.} Note that the second step of the algorithm can be done in time $O(n^2)$ time. The algorithm in \Cref{lem:lopsidedUniversal} takes time $\vert  \mathcal{F} \vert  n$ to compute $\mathcal{F}$ and to enumerate the set $\mathcal{F}$. Hence, the total time required is ${(\maxcoal+\maxdeg \maxcoal) \choose \maxcoal}\cdot 2^{o(\maxcoal+\maxdeg \maxcoal)}  \log n \cdot n^{O(1)}$ which is $2^{O(\maxcoal \log \maxdeg)o(\maxdeg \maxcoal)} \cdot n^{O(1)}$. This concludes the proof of the theorem.
\end{proof}
}

Based on the observation below, we obtain a color-coding based fixed-parameter algorithm for the combined parameter~$(\maxcoal,\fas)$.
To this end, we call a coalition in a given coalition structure~$\Pi$ a \myemph{singleton} (resp.\ \myemph{non-singleton}) coalition if it has size one (resp.\ larger than one). 
Accordingly, an agent is a \myemph{singleton} (resp.\ \myemph{non-singleton}) agent (wrt.\ $\Pi$) if she is in a {singleton} (resp.\ {non-singleton}) coalition.
\newcommand{\obsfecorefptkfas}{
  If $\Pi$ is core stable, then there are at most $\maxcoal \cdot \fas$ non-singleton agents.
}

\begin{observation}[\appendixsymb]\label{obs:non_singleton_number}
\obsfecorefptkfas
\end{observation}

\appendixproofwithstatement{obs:non_singleton_number}{\obsfecorefptkfas}{
  \begin{proof}[Proof of \cref{obs:non_singleton_number}]
    \renewcommand{\qedsymbol}{$\diamond$}
    As noted by \cite{woeginger2013core}, it is straightforward that in a core stable partition every partition must be a strongly connected component.

    Each connected non-singleton component contains at least one cycle and because the initial coalitions are pairwise disjoint, these cycles are disjoint. To make $\goodG$ acyclic, we must remove an edge from each of these cycles. Therefore the feedback arc set number of $\goodG$ is at least the number of non-singleton coalitions in $\Pi$.    
    Since each coalition has at most $\maxcoal$ agents, the observation follows.
  \end{proof}
}

\newcommand{\thmfecorefptkdelta}{
\FE-\verif and \FE-\sverif are \fpt wrt.\ $(\maxcoal, \fas)$.
}
\begin{theorem}\label{thm:fe-core-fpt-k-delta}
  \thmfecorefptkdelta
\end{theorem}

\begin{proof}[Proof sketch]
  Let $(V, \goodG)$ be an instance of \FE and $\Pi$ an initial coalition structure.
  The algorithm has two phases.
  First, %
  we preprocess the instance so that each non-trivial blocking coalition has at most $\maxcoal$ agents and the reduced instance excludes some undesired cycles.
  Second, we further reduce the instance to one which is acyclic and observe that any non-trivial blocking coalition must ``contain'' an in-tree of size $O(\maxcoal^2)$.
  Hence, for each possible in-tree we can use color-coding to check whether it exists in FPT-time. In the following we use $\singles$ and $\nonsingles$ to denote the singleton and non-singleton agents, respectively.
  
 \noindent The first phase consists of the following polynomial-time steps:
  \begin{compactenum}[\textbf{(P}1\textbf{)}]
    \item Check whether $\Pi$ contains a coalition~$U$ such that $\goodG[U]$ is not strongly connected.
    If yes, then return NO since the strongly connected subgraph corresponding to the sink component in~$\goodG[U]$ is strictly blocking $\Pi$. 
    \item If $\goodG[\singles]$ contains a cycle, then the singletons agents on the cycle is strictly blocking $\Pi$, so return NO.  \label{fe-core-fpt-k-delta_prep_2}%
    \item\label{P3} By \cref{lemma:preprocess}, check in linear time whether there is a blocking (resp.\ weakly blocking) coalition of size greater than $\maxcoal$.
\end{compactenum}

\appendixalg{thm:fe-core-fpt-k-delta}{%
\noindent  The second phase is as follows:  
  For each subset~$\Bns\subseteq \nonsingles$ of size $k'\le \maxcoal$ and each size $b$ with $k' \le b\le \maxcoal$,
  we check whether there exists a blocking coalition of size~$b$ which contains all non-singletons from $\Bns$ and exactly $b-|\Bns|$ singletons; note that after phase one, we only need to focus on coalitions of size at most $\maxcoal$ and can assume that $|\nonsingles|\le \maxcoal\cdot \fas$.
  We return YES if and only if no pair~$(\Bns, b)$ can be extended to a blocking coalition (i.e., \cref{alg:dp_treesearch} returns NO for all $(\Bns, b)$).

  Given $(\Bns, b)$, the task reduces to searching for the $b-|\Bns|$ missing singleton agents, assuming that such an extension is possible.
  To achieve this, we reduce to searching for an in-tree of size $O(\maxcoal^2)$ in a directed acyclic graph~(DAG)~$H$, which using color coding, can be done in $f(\maxcoal)\cdot |H|^{O(1)}$ time where $f$ is some computable function.
  First of all, if $|\Bns| = b$, then we check whether $\Bns$ is blocking (resp.\ weakly blocking) in polynomial time and return $\Bns$ if this is the case; otherwise we continue with a next pair~$(\Bns, b)$.
  In the following, let $B$ be a hypothetical blocking coalition of size~$b> |\Bns|$ which consists of~$\Bns$ and $b-|\Bns|$ singleton agents. 
  The searching has two steps.

  \noindent \textbf{(C1) Construct a search graph~$\hgoodG$ from $\goodG$.} Based on $\goodG$ we construct a DAG~$\hgoodG$, where
  we later search for the crucial part of the blocking coalition.
  We compute the minimum number~$r(a_i)$ of singleton-friends each non-singleton agent~$a_i$ in $\Bns$ should obtain from $B$ by checking how many friends she initially has.
  Let \myemph{$\frPi$} and \myemph{$\frBns$} denote the number of friends agent~$a_i$ has in~$\Pi(a_i)$ and $\Bns$, respectively.
  If $\frPi <  \frBns$, then let $r(a_i) \coloneqq 0$.
  Otherwise for \verif, let

  {\centering
    $r(a_i) \coloneqq
    \begin{cases}
      \frPi + 1 - \frBns, & \text{ if } b \geq |\Pi(a_i)| \\
      \frPi - \frBns, &  \text{ otherwise.}
    \end{cases}$
    \par
  }
  
  \noindent For \sverif, the first if-condition is $b > |\Pi(a_i)|$ instead of $b \geq |\Pi(a_i)|$.
  After the computation, we check whether some non-singleton agent~$a_i\in \Bns$ has $r(a_i) > b - |\Bns|$.
  If $a_i$ is such an agent, then she will not weakly prefer~$B$ to~$\Pi(a_i)$ since there are not enough friends for her, so we continue with a next pair~$(\Bns,b)$.
  Now, we construct~$\hgoodG$.
  First, duplicate the vertices in $\Bns$ as \myemph{$\hBns$} $\coloneqq \{a^z_i \mid (a_i,z) \in \Bns \times [r(a_i)]\}$.
  The vertex set and the arc set of $\hgoodG$ are defined as
  \myemph{$\hV$} $\coloneqq \hBns \cup \singles \cup \{t\}$, where $t$ is an artificial sink, and
  \myemph{$\hE$} $\coloneqq \{(a^z_i, s) \mid (a_i, s) \in E(\goodG)
  \cap (\Bns\times \singles), z\in [r(a_i)]\}
  \cup  E(\goodG[\singles]) \cup \{(s, t) \mid (s, a_i) \in E(\goodG) \cap (\singles \times \Bns)\}$, respectively.
  Now let $\hgoodG \coloneqq (\hV, \hE)$.
  Briefly put, we remove all arcs that are not incident to the singletons,
  and redirect every arc from a singleton to a non-singleton vertex in $\Bns$ to the artificial sink~$t$. 
  Note that $\hgoodG$ is acyclic since by (P\ref{fe-core-fpt-k-delta_prep_2}) no singleton agents induce a cycle.

  \noindent \textbf{(C2) Search for a tree structure in $\hgoodG$.}
  Observe that in $\goodG[B]$, each non-singleton agent~$a_i\in \Bns$ has at least $r(a_i)$ singleton friends and each singleton agent in $B\setminus \Bns$ has at least one friend.
  Equivalently, in the modified induced subgraph~$\hgoodG[B]$, each non-singleton agent in~$\hBns$ (resp.\ singleton agent in $B\setminus \hBns$) has at least one out-arc. Then, $\hgoodG[B]$ contains an in-tree~$\TB$ on vertex set~$\hBns\cup (B\setminus \Bns) \cup \{t\}$ such that 
\begin{compactenum}[(t1)]
  \item \label{tree1} every vertex~$a_i^z\in \hBns$ has exactly one out-neighbor and this out-neighbor is a singleton vertex such that no two non-singletons~$a_i^z$ and $a_i^{j}$ ($z\neq j$) share the same out-neighbor, 
  \item \label{tree2} every singleton vertex in $V(\TB)\setminus (\hBns\cup \{t\})$ has exactly one out-neighbor and this out-neighbor is either the root~$t$ or some singleton vertex, and
  \item \label{tree3} $t$ does not have any out-neighbors.
\end{compactenum}
Observe that $\TB$ has exactly $|\hBns|+|B\setminus \Bns|$ arcs. 
Since $\hgoodG$ is acyclic and the artificial sink~$t$ does not have any out-arcs, by the above conditions, $\TB$ must be a directed in-tree with root at~$t$.
In particular, $\TB$ is connected.
For ease of reasoning, let us call an in-tree~$T$ \myemph{good} if
there exists a subset~$B'$ of $\singles$ with $b-|\Bns|$ vertices
such that $T$ is a directed graph on $\hBns\cup B' \cup \{t\}$
and satisfies condition~(t\ref{tree1})--(t\ref{tree3}) above, replacing the name $\TB$ with $T$.

If $\hgoodG$ contains a good in-tree~$T'$, then the vertices in~$\Bns$ and the singleton vertices in $T'$ forms a desired blocking coalition. 
By applying the color-coding algorithm of \citet{Alon1995Colorcoding}, we can already search for a good in-tree in FPT-time.
For the sake of completeness and to better analyze the running time, we show how to combine color-coding with a polynomial-time algorithm to search for it.

We describe the approach via~\cref{alg:dp_treesearch}.
By \citet{naor1995splitters}, in line~\ref{alg:colorcoding-hash}, we compute in $f(\maxcoal)\cdot |\singles|^{O(1)}$ time a family~$\mathcal{F}$ of coloring functions  (aka.\ \myemph{perfect Hash family}) from $\singles$ to $[b-|\Bns|]$ which guarantees to contain a \emph{good} {coloring} function.
Here, a function $\chi\colon \singles\to [b-|\Bns|]$ is called \myemph{good} (wrt.~$B$) if it assigns to each singleton vertex in $B\cap \singles$ a distinct color from $[b-|\Bns|]$\ifshort; see the full version for more details on this.
\else
See appendix for more details on this.
\fi
Hence, in line~\ref{alg:colorcoding-coloring} we iterate through each coloring~$\chi$ in~$\mathcal{F}$.
Note that if $\chi$ is good for~$B$,
then after coloring the vertices in $\singles$ according to $\chi$,
there must exist a good in-tree on vertices $\hBns\cup [b-|\Bns|]\cup \{t\}$ as well. %
Hence, in line~\ref{alg:colorcoding-tree}, we iterate through all good in-trees~$T$. %
\begin{algorithm}[t!]
\caption{\textbf{(C2)} Searching for $\TB$ in $\hgoodG$ given $\Bns$ and $b$}\label{alg:dp_treesearch}

  $\mathcal{F} \gets $ $(|\singles|,b-|\Bns|)$-perfect Hash family on the universe~$\singles$\label{alg:colorcoding-hash}

  \ForEach{coloring~$\coloring \in \mathcal{F}$ with $\chi\colon \singles \to [b-\Bns]$\label{alg:colorcoding-coloring}}{
    \ForEach{Good in-tree~$T$~on vertex set~$\hBns\cup [b-|\Bns|] \!\cup\! \{t\}$\label{alg:colorcoding-tree}}{%
      $\Bs \gets \emptyset$; $\tau(\singles) \gets$ a topological order of $\singles$ in $\hgoodG$
      
      \ForEach{$v \in \tau(V_S)$\label{alg:colorcoding-singletons-A}}{
        $\Colors_v \gets N^-_T(\coloring(v)) \cap [b-|\Bns|]$ \label{alg:colors}
        
        \If{$\Colors_v\subseteq \coloring(N^{-}_{\hgoodG}(v)\cap \singles)$
          and $N^-_T(\coloring(v)) \cap \hBns = N^{-}_{\hgoodG}(v)\cap \hBns$\label{eq:in-neighbors}}
        {let $\sing_v \!\subseteq\! N^{-}_{\hgoodG}\!(v)\!\cap\! \singles$ s.t.\  $|\sing_v|\!=\!|\Colors_v|$ and $\coloring(\sing_v)\!=\!\Colors_v$\label{alg:colorcoding-S_v}}
        \lElse{$\hgoodG\gets \hgoodG-v$}\label{alg:colorcoding-singletons-Z}
      }
      \If{$N^-_T(t) \subseteq \coloring(N^{-}_{\hgoodG}(t))$\label{alg:colorcoding-t}}
      {
        let $\sing_t \!\subseteq\! N^{-}_{\hgoodG}(t)$ s.t.\ $|\sing_t|\!=\!|N^-_T(t)|$ and $\coloring(\sing_t)\!=\!N^-_T(t)$\label{alg:colorcoding-in-t}%

        $\Bs\gets \sing_t$; $Q \gets \sing_t$

        \While{$Q \neq \emptyset$}{
          $P \gets \bigcup_{v\in Q}\sing_v$

          $\Bs\gets \Bs \cup P$;     $Q\gets P$            \label{alg:colorcoding-Bs}%

        }\label{alg:t-loop}
        \lIf{$|\Bs|=b-|\Bns|$}
        {\Return{$\Bs\cup \Bns$}}
      }
    }
  }

\Return NO, i.e., $(\Bns,b)$ cannot be extended to a blocking coalition
\end{algorithm}

For ease of reasoning, we also use color to refer to a vertex in $[b-|\Bns|]$, and given a subset~$S'\subseteq \singles$ let $\chi(S')=\{\chi(s)\mid s\in S'\}$.
In lines~\ref{alg:colorcoding-singletons-A}--\ref{alg:colorcoding-singletons-Z}, we iterate through each singleton vertex~$v$ in the topological order~$\tau(\singles)$ in $\hgoodG$ (recall that $\hgoodG$ is a DAG),
and check whether it has enough in-neighbors whose colors match the in-neighbors of its color~$\chi(v)$ in the tree~$T$ (line~\ref{eq:in-neighbors}).
More specifically, we check whether $v$ has singleton in-neighbors of colors~$\Colors_v$ indicated by the in-neighbors of $\chi(v)$ in $T$ and whether it has the same non-singleton in-neighbors as its color~$\chi(v)$ in $T$.
If yes, then we use $\sing_v$ to store a subset of such singleton in-neighbors for~$v$ (line \ref{alg:colorcoding-S_v}). 
Otherwise, assuming that $\chi$ is good, $v$ cannot be used for the blocking coalition, so we delete it from the search graph.
Note that the order~$\tau(\singles)$ ensures that we do not mistakenly store a singleton vertex which cannot be used later on. %
Finally, in line~\ref{alg:colorcoding-t}, we check whether the in-neighbors of $t$ contain enough singletons with appropriate colors.
If yes, we use the stored set of singleton vertices~$\sing_v$ to iteratively collect all vertices in $\sing_v$ from root~$t$ to leaves; note that in an in-tree, the root is the sink and the leaves are the sources.
We return the set~$\Bs$ if it contains exactly $b-|\Bns|$ singletons, and return NO if no iteration gives a desired set~$\Bs$.
\ifshort
The correctness proof and running time analysis are deferred to the full version.
\else
The correctness proof and running time analysis are deferred to the appendix.
\fi
}
{\thmfecorefptkdelta}{%
  \smallskip
\noindent \textbf{Missing material for the perfect Hash family.}

\begin{definition}[\cite{Alon1995Colorcoding,naor1995splitters}]
  A \emph{$(|U|,k)$-perfect Hash family~$\mathcal{F}$} on universe $U$ is a family of functions from $U$ to $[k]$
  such that for every subset~$S\subset U$ of size~$k$ there exists a function~$f$ in the family~$\mathcal{F}$ that is bijective on~$S$, i.e., $f^{-1}([k])=S$.
\end{definition}

\begin{lemma}[\cite{naor1995splitters}]\label{lem:naor_splitters}
  There is an algorithm that given $U$ and $ k\in \mathds{N}$ constructs a $(|U|, k)$-perfect Hash family~$\mathcal{F}$ on $U$ of cardinality $e^k \cdot k^{O(\log k)}\cdot \log|U|$ in $e^k\cdot k^{O(\log k)}\cdot |U| \log|U|$ time.
\end{lemma}

\noindent \textbf{Correctness of the algorithm.}
The correctness proofs for \FE-\verif\ and \FE-\sverif\ work almost the same. The only difference lies in the value~$r(a_i)$ for each non-singleton agent~$a_i$. 

\begin{sloppypar}We start by showing that if after the two phases the algorithm returns a coalition~$B$,
then it is blocking (resp.\ strictly) $\Pi$.
By construction, let $B$ consist of the non-singletons~$\Bns$ and the singletons~$\Bs$, and let $\chi$, $T$, and $(\sing_s)_{s\in \Bs}$ be the good coloring function, the in-tree, and the stored vertex sets, respectively, that were used to find $\Bs$.
Let $b=|B|$.
\end{sloppypar}

Using the inverse~$\chi^{-1}$, we claim that $\hgoodG$ contains a tree which is derived from $T$ by renaming the colors back to the singleton agents: %

\begin{claim}\label{cl:T-is-contained-G}
  The following function~$\phi\colon V(T) \to \hBns \cup \Bs \cup \{t\}$ is bijective.
  \begin{align*}
    \phi(\chi(v)) = v,   \text{ for every } v\in \Bs, \text{ and }\\
    \phi(a) = a,  \text{ for every } a \in \hBns \cup \{t\}.
  \end{align*}
 For each~$(i,j)\in E(T)$ we have $(\phi(i), \phi(j))$ $\in$ $E(\hgoodG[\hBns\cup \Bs \cup \{t\}])$.
\end{claim}
\begin{proof}[Proof of \cref{cl:T-is-contained-G}] \renewcommand{\qedsymbol}{$\diamond$}
  We first show that $\phi$ is well-defined and bijective.
  For the well-definedness of $\phi$, it suffices to show that no two vertices in $\Bs$ share the same color under~$\chi$ since $|\Bs|=b-|\Bns|$.
  Certainly, by construction, all agents in $\sing_t$ have different colors (see line~\ref{alg:colorcoding-t}).
  Thus, we need to show that if $\Bs$ does not contain any two vertices with the same color before the execution in line~\ref{alg:colorcoding-Bs},
  then it will neither after this line.
  Now, assume that $\Bs$ does not have two vertices with the same color before executing line~\ref{alg:colorcoding-Bs}, and let $Q$ denote a set of vertices that were added in a previous iteration or in line~\ref{alg:colorcoding-t} if this is the first iteration.
  Suppose, for the sake of contradiction, that some two vertices in $\Bs\cup (\cup_{v\in Q}\sing_v)$ has the same color, say $c$.
  Since we only add a vertex~$u$ if it has out-arc to some vertex~$v$ in $Q$ and its color~$\chi(u)$ is an in-neighbor of $\chi(v)$ in $T$,
  there must exists two colors in $T$ which has the same in-neighbor~$c$,
  implying that $c$ has two out-arcs in~$T$, a contradiction to $T$ being an in-tree. 
  
  Now, it remains to show that $(\phi(i), \phi(j)) \in  E(\hgoodG[\hBns\cup \Bs \cup \{t\}])$ holds for each $(i,j)\in E(T)$.
  We distinguish between two cases, either $j=t$ or $j\in [b-|\Bns|]$; recall that by definition, no vertex has out-arc to a non-singleton in~$T$.
  If $j=t$, then by construction, there must exist a singleton vertex~$s\in \sing_t$ with $\chi(s)=i$.
  By definition, we have that $s\in \Bs$, meaning that $(s,t)\in E(\hgoodG[\hBns\cup \Bs \cup \{t\}])$, as desired.

  If $j\in [b-|\Bns|]$, then by the bijection of $\phi$, there must exist a singleton~$v\in \Bs$ with $\chi(v)=j$.
  Hence, $i\in \hBns\cup [b-|\Bns|]$.
  By line~\ref{eq:in-neighbors} in \cref{alg:dp_treesearch}, there must exist a vertex~$u\in \hBns\cup \Bs$ such that $\chi(u)=i$ and $(u,v)\in  E(\hgoodG)$, as desired.
\end{proof}

Next, we show that $B$ has enough friends for each agent in it.
\begin{claim}\label{cl:B-enoughfriends}
  \begin{compactenum}[(i)]
    \item Each non-singleton agent~$a_i \in \Bns$ has at least $r(a_i)$ friends in $B$.
    \item Each singleton agent~$s\in \Bs$ has at least one friend in $B$.
  \end{compactenum}
\end{claim}

\begin{proof}[Proof of \cref{cl:B-enoughfriends}] \renewcommand{\qedsymbol}{$\diamond$}
  For the first statement, we observe that in $T$, each non-singleton agent~$a_i\in \Bns$ is duplicated $r(a_i)$ times such that each copy~$a_i^j$, $j\in [r(a_i)]$, has exactly one out-neighbor towards a singleton vertex, but no two copies share the same out-neighbor (see condition (t\ref{tree1})).
  By the bijection defined in \cref{cl:T-is-contained-G}, we infer that each non-singleton agent has at least $r(a_i)$ friends in~$B$.

  The second statement works analogously by
  observe that each vertex~$i\in [b-|\Bns|]$ has exactly one out-neighbor in $T$ which is either $t$ (if in the blocking coalition she will have a friend to a non-singleton agent) or another vertex in $[|B|-|\Bns|]$ (if in the blocking coalition she will have a friend to a singleton agent.
\end{proof}

The above two claims immediately imply the following.
\begin{claim}\label{cl:B-blocking}
  If \cref{alg:dp_treesearch} returns a set~$B$, then
  for \FE-\verif, $B$ is strictly blocking~$\Pi$,
  while for \FE-\sverif, $B$ is weakly blocking~$\Pi$.
\end{claim}

\begin{proof}[Proof of \cref{cl:B-blocking}] \renewcommand{\qedsymbol}{$\diamond$}
  By \cref{cl:B-enoughfriends} and by the definition of $r(a_i)$, we know that in the case of \FE-\verif, each agent in $B$ prefers $B$ to her initial coalition, so $B$ is strictly blocking~$\Pi$. 
For \FE-\sverif, each non-singleton agent weakly prefers $B$ to her initial coalition and each singleton agent in $\Bs$ strictly prefers $B$ to her initial coalition.
Since $B\neq \Bns$, i.e., there is at least one singleton agent who strictly prefers $B$ over $\Pi$, $B$ is weakly blocking~$\Pi$, as desired.
\end{proof}

Now, we show that if there exists a blocking coalition, then our algorithm will return NO in the first phase or return a blocking coalition in the second phase.
Clearly, if $\goodG[U]$ is not strongly connected for some coalition~$U$ in $\Pi$ or $\goodG[\singles]$ contains a cycle, then phase one will detect it and return NO.
Similarly, in (P\ref{P3}), we use the linear time algorithm behind \cref{lemma:preprocess} to check whether a blocking coalition of size larger than $\maxcoal$ and report it accordingly.
Hence, let us assume that we are in phase two and there is a blocking coalition which is not found in the preprocessing in the first phase.
Then,  $|\nonsingles|\le \maxcoal\cdot \fas$ and any blocking coalition has at most $\maxcoal$ agents.
Let $B^*$ be an inclusionwise minimal blocking (resp.\ weakly blocking) coalition.
Let us consider the iteration in phase two where $\Bns = B^*\cap \nonsingles$ and $b=|B^*|$.
We can also assume that $|\Bns| < b$ as otherwise $\Bns=B^*$ and we can check whether $\Bns$ is blocking  (resp.\ strictly blocking) in polynomial-time. 

Let $(r(a_i))_{a_i\in \Bns}$ and $\hBns$ be the numbers of friends needed and the copies of the non-singletons, respectively.
Let $\Bs^*=B^*\setminus \Bns$ be the singleton agents in $B^*$ and let $\hgoodG$ be the modified search graph.
Since $B^*$ is blocking~$\Pi$, the modified graph~$\hgoodG[\hBns\cup \Bs^* \cup \{t\}]$ must contain a good in-tree~$\TBs$. %
If $B^*$ is blocking, then each non-singleton agent $a_i$ in $\Bns$ must have 
at least $r(a_i)$ singleton friends in $B^*\setminus \Bns$ i.e., each agent in $\Bns$ (weakly) prefers $B^*$ to her coalition in $\Pi$ and each singleton in $B^*$ must have a friend in $B^*$, otherwise she would strictly prefer $\Pi$ over $B^*$ since she would have no friends but at least one enemy in $B^*$. Hence, $\hgoodG[\hBns\cup \Bs^* \cup \{t\}]$ contain a in-tree which satisfies (t\ref{tree1})--(t\ref{tree3}). That is, for both the modified graph~$\hgoodG[\hBns\cup \Bs^* \cup \{t\}]$ must contain a good in-tree~$\TBs$.
By the perfect Hash functions, let $\chi$ be a good coloring for~$\Bs^*$ (i.e., bijective on $\Bs^*$).
Now observe that if we color the singleton vertices in~$\TBs$ according to $\chi$,
then we obtain a new good in-tree on $\Bns\cup [b-\Bns] \cup \{t\}$. %
Let $T$ be such a good in-tree. Formally, $V(T)=\hBns\cup [b-|\Bns|]\cup \{t\}$ such that
$E(T)=\{(a^z_i,\chi(s))\mid (a^z_i,s)\in E(\TBs) \cap (\hBns\times \Bs^*) \}\cup
\{(\chi(u), \chi(v))\mid (u,v)\in E(\TBs[\Bs^*])\}\cup
\{(\chi(s), t)\mid (s,t)\in E(\TBs), s\in \Bs^*\}$.

Let us go to the iteration when $\chi$ and $T$ are considered in line~\ref{alg:colorcoding-tree}.
Since $T$ is good and isomorphic to $\TBs$, 
no vertex in $\Bs^*$ is deleted in line \ref{alg:colorcoding-singletons-Z}.
Now, we claim that at the latest in this iteration, our algorithm will return a coalition which is blocking.
After going through all singleton agents in lines \ref{alg:colorcoding-singletons-A}--\ref{alg:colorcoding-singletons-Z}, the algorithm selects a non-empty subset~$\sing_t$ as described in line~\ref{alg:colorcoding-in-t} which is possible since $\TBs$ is good and $\Bs^*$ is not deleted.
Let $\Bs$ be the subset computed at the end of the while-loop in line \ref{alg:t-loop}.
Clearly, if $|\Bs|=b-|\Bns|$, then by \cref{cl:B-blocking}, we know that $\Bs$ is blocking (resp.\ strictly blocking)~$\Pi$.
Hence, it remains to show that $|\Bs|=b-|\Bns|$. Clearly, $|\Bs| \le b-|\Bns|$ since $\hgoodG$ is a DAG, so no vertex is considered more than once in the while-loop.
Now, suppose that one color~$c$ exists such that $\chi^{-1}(c)\cap \Bs =\emptyset$.
Then, by the goodness of~$T$, it follows that for the only out-neighbor~$d$ of~$c$ in $T$, it also holds that $\chi^{-1}(d)\cap \Bs =\emptyset$ as otherwise, when
some vertex~$v\in \chi^{-1}(d)$ was considered in $Q$, all vertices from $\sing_v$ will be added to $\Bs$, including the one with color~$c$.
Since $T$ is connected, all colors on the path from $c$ to $t$ are not present in $\chi^{-1}(\Bs)$, a contradiction to $t$ satisfying the condition given in line~\ref{alg:colorcoding-t}.

\noindent \textbf{Running time.}
Clearly, the processing steps in the first phase can be conducted in polynomial time.
By \cref{obs:non_singleton_number} we know that if $\goodG$ passes phase one, then there are at most $\maxcoal\cdot \fas$ non-singleton agents, every blocking coalition has at most $\maxcoal$ agents, and no blocking coalitions consist of only singleton agents.
Thus, there are $\sum_{i=1}^{\maxcoal}\binom{\maxcoal\cdot \fas}{i}$ ways to select the non-singleton agents~$\Bns$ for a potential blocking coalition and for each $\Bns$,
the size of the potential blocking coalition ranges from $|\Bns|$ to $\maxcoal$.
In total, the number of all combinations for $(\Bns,b)$ is $\maxcoal\cdot (\maxcoal\cdot \fas)^{\maxcoal}$.

Consider a combination~$(\Bns,b)$ and let $n$ and $m$ denote the number of agents and arcs in $V$ and $\goodG$, respectively.
Computing for every $a_i \in \nonsingles$, the value~$r(a_i)$ can be done in $O(n)$ time.
Similarly, $\hgoodG$ can be computed in $O(m + n\cdot \maxcoal)$-time; note that each non-singleton agent can have up to $\maxcoal$ copies.
Computing a $(|\singles|, b-|\Bns|)$-perfect Hash family~$\mathcal{F}$ takes $e^{\maxcoal}\cdot \maxcoal^{O(\log \maxcoal)} \cdot |\singles| \cdot \log{|\singles|}$ time. There are $e^{\maxcoal} \cdot \maxcoal^{O(\log \maxcoal)}\cdot \log|\singles|$ Hash functions to iterate through.

There are $(b-|\Bns|)^{(b-|\Bns|)\cdot |\Bns| + b-|\Bns|}=O(\maxcoal^{\maxcoal^2})$ many in-trees satisfying conditions~(t\ref{tree1})--(t\ref{tree3}) since each non-singleton
agent has $b-|\Bns|$ ways to select a singleton out-neighbor,
and each singleton agent has $b-|\Bns|$ ways to select out-neighbor which is either another singleton or the root~$t$.

A topological order can be computed in linear time. 
The for-loop in lines~\ref{alg:colorcoding-singletons-A}--\ref{alg:colorcoding-singletons-Z} can be done in $O(n+m+n\cdot \maxcoal)$ time since non-singleton agent has $O(\maxcoal)$ copies and finding an appropriate singleton agent subset~$\sing_v$ can be done in linear time. 
Finally, the collection of the stored vertices in lines~\ref{alg:colorcoding-t}--\ref{alg:t-loop} can also be done in linear time.

In total, the algorithm takes
$O(\maxcoal^{\maxcoal +1}\cdot \fas^{\maxcoal} \cdot e^{\maxcoal}\cdot \maxcoal^{O(\log{\maxcoal})} \cdot \maxcoal^{\maxcoal^2}\cdot (\log n + n+m+n\cdot \maxcoal))$-time, which is \fpt\ with respect to $(\maxcoal, \fas)$.}
\end{proof}

\looseness=-1
Next, we determine the existence of Nash stable partitions and show a dichotomy result, strengthening a result by \citet{Brandt_Bullinger_Tappe_2022}.

\newcommand{\thmfenashfasdelta}{%
  \FE-\NS\ is polynomial-time solvable if $\fas \leq 2$, 
  whereas it is
  \np-hard even if $\fas=3$ and $\maxdeg=5$. 
}
\begin{theorem}[\appendixsymb]\label{thm:fe-nash-fas-delta-nph}
  \thmfenashfasdelta %
\end{theorem}

  \begin{proof}[Proof sketch]
    We only show the first part and give a linear-time algorithm for $\fas\le 2$.
    Let $(V, \goodG)$ be an instance of \FE-\NS.
    For ease of reasoning, call an agent a \myemph{sink agent} if she does not have any friend in $\goodG$; otherwise call her a \myemph{non-sink} agent. Let $S$ denote the set of all sink agents. First, we put each sink agent into a singleton coalition. 
After that, if there is a non-sink agent who has \emph{only} sink friends, then we return NO.
If every non-sink agent that has a sink friend also has at least two non-sink friends, then we return YES. Otherwise there is an agent $v$ who has at least one sink friend and only one non-sink friend $w$. We place $v$ and $w$ in a size-two coalition~$\{v, w\}$ if the friendship relation is symmetric (i.e., $w$ also considers $v$ a friend), otherwise we stop and return NO.
Finally, we put all remaining agents in the same coalition. If the obtained partition is Nash stable,
then we output YES, otherwise~NO.

For correctness, observe that all sink agents must be in singleton coalitions in any Nash stable partition. If there is an agent~$v$ who has only sink friends,
there is no way to place her into a coalition which contains at least one of her friends.
However, because $v$ has at least one sink friend, she will envy the singleton coalition of the friend, and we cannot have a Nash stable partition.
If every non-sink agent, who has a sink friend, has at least two non-sink friends,
the following partition~$\{\{s\} \mid s \in S\} \cup \{\{V \setminus S\}\}$ is Nash stable:
Clearly, each agent~$v \in V \setminus S$ that has a sink friend has at least 2 friends in $V \setminus S$ but at most one friend in any $\{s\}$, where $ s \in S$. If $v \in V \setminus S$ has no sink friends, then all of $v$'s friends are by construction in $V \setminus S$.

Assume there is a vertex $v$, who has a sink friend $s$ but only one non-sink friend $w$.
Then by the above, in any Nash stable partition~$\Pi$,
$\Pi(s)=\{s\}$ but $v$ has one friend and no enemies in $\{s\}$. %
This means, in~$\Pi$, we have $\Pi(v) = \{v, w\}$.
Hence, if the friendship is not symmetric, then there is no Nash stable solution.
Finally, if the algorithm assigned $v$ and $w$ together, then $(v,w,v)$ is a cycle and must contain a feedback arc. Every remaining non-sink agent must have a friend in any Nash stable solution. This implies that they must be in a coalition that contains a cycle. Since there is only one feedback arc left in $\goodG[V \setminus (S \cup \{v, w\}]$,
the remaining agents must be in the same coalition in a Nash stable partition. Therefore, if there is a Nash stable solution, then the algorithm finds one and otherwise outputs NO. 
\ifshort
The analysis of the running time is straightforward and deferred to the full version.
\else
The analysis of the running time is straightforward and deferred to the appendix.
\fi
\appendixalg{thm:fe-nash-fas-delta-nph}{\thmfenashfasdelta}{\thmfenashfasdelta}{
 \paragraph*{Runtime of the algorithm for $\fas \leq 2$}
It remains to show that the algorithm runs in time $O(n+m)$, where $n$ is the number of agents and $m$ the number of friendship relations.
Dividing agents into sinks and non-sinks takes $O(n)$-time. In the same step we can already place the sinks into singleton coalitions. We can check in $O(n + m)$-time how many non-sink friends each non-sink agent has. Checking whether each non-sink agent that has a sink friend has at least two non-sink friends is done in time $O(n)$. Similarly, finding an agent who has only one non-sink friend and placing it in a coalition with her can be done in time $O(n)$. In the final step we check whether a given coalition structure is Nash stable. This can be done in time $O(n+m)$. Thus the whole time complexity is $O(n+m)$.

\paragraph*{\np-hardness for $\fas = 3$}
Now, we show that the problem becomes \np-hard for $\fas=3$.
We reduce from \pxct.
  Let $I=([3\enn], \mathcal{C})$ denote an instance of \pxct\ with $\mathcal{C}=\{C_1$, $\ldots$, $C_{\emm}\}$. %
  Without loss of generality, assume that each element appears in at least two sets as otherwise the unique set~$C$ which contains the element is necessarily in the solution and we can delete the elements contained in~$C$ and the set~$C$ from the instance.  
  We create an instance of \FE-\NS\ as follows; all unmentioned relations are enemy relations.
   \begin{compactitem}[--]
     \item For each element~$i\in [3\enn]$, create two \myemph{element}-agents~$x_i$ and $y_i$ such that agent~$y_i$ considers $x_i$ a friend.
     For each~$i\in [3\enn-1]$, create a dummy agent~$w_i$ who does not consider anyone a friend, but is considered by $y_i$ a friend. 
     Additionally, for each~$i\in [3\enn-1]$, agents~$x_i$ and $y_i$ consider $y_{i+1}$ as a friend. 
     \item For each set~$C_j\in \mathcal{C}$, create a \myemph{set}-agent~$s_j$ and two auxiliary agents~$t_j$ and $r_j$ such that $s_j$ considers both $r_j$ and $t_j$ friends.
     For each element~$i\in [3\enn]$ and each set~$C_j\in \mathcal{C}$ with $i\in C_j$, agent~$x_i$ considers~$s_j$ a friend. 

     \item Create five enforcer-agents, $\{b_1, b_2, b_3, b_4, b_5\}$, such that $b_5$ does not have any friends, agent $b_3$ and $b_4$ are mutual friends, agent $b_4$ additionally considers $b_5$ a friend, and agent $b_3$ additionally considers $b_1$ and $b_2$ friends.
     Both agents $b_1$ and $b_2$ consider both $y_1$ and $t_1$ friends.     
     The sequence~$(b_1$, $y_1,$ $x_1$, $\ldots$, $y_{3\enn}$, $x_{3\enn}$, $r_{\emm}$, $\ldots$, $r_{1}$, $b_1)$ forms a friendship cycle.

     \item Create $8\enn-\emm$ dummy agents~$t_{\emm+1},\ldots,t_{8\enn}$; recall that $\emm \le 3\enn$.
     The sequence~$(b_2$, $t_1$, $\ldots$, $t_{8\enn}$, $b_1)$ forms a friendship cycle.
   \end{compactitem}
   
   \noindent This completes the construction of the instance, which can clearly be done in polynomial time.
   Note that we have created $17\enn+2\emm+4$ agents. 
   See \cref{fig:fe-nash-fas-delta-nph} for an illustration.
     
     We analyze the feedback arc set number and the maximum degree in the friendship graph.
     It is straightforward to check that deleting the three arcs $(b_4, b_3)$, $(r_1, b_1)$, and $(t_{8\enn}, b_2)$ results in an acyclic graph with a topological order such as $(b_3$, $b_4$, $b_5$, $b_1$, $b_2$, $y_1$, $x_1$, $\ldots$, $y_{3\enn}$, $x_{3\enn}$, $w_{1}$, $\ldots$,  $w_{3\enn-1}$, $s_1$, $\ldots$, $s_{\emm}$, $r_{\emm}$, $\ldots$, $r_1$, $t_1$, $\ldots$,  $t_{8\enn})$.
     As for the maximum degree,
     we note that agents~$y_i$, $i\in [3\enn-1]$, $s_j$, $j\in [\emm]$, and some element-agent~$x_i$, $i\in [3\enn]$ have maximum degree, which is five. 

     It remains to show the correctness, i.e., $I$ admits an exact cover if and only if the constructed instance has a Nash stable partition.
     \begin{figure}
       \centering
       \begin{tikzpicture}[scale=1,every node/.style={scale=0.9}, >=stealth', shorten <= 2pt, shorten >= 2pt]
         
         \foreach  \x / \y / \n in
         {-0.5/0/b1, 0.5/0/b2, 0/.5/b3, 1/.5/b4, 2/.5/b5}
         {
           \node[pn] at (\x, \y) (\n) {};
         }
         \begin{scope}[xshift=-1.8cm, yshift=-1cm]
           \foreach  \x / \y / \n  / \st in
           {1.2/0/x1/pn, 1.2/-.5/x2/pn, 1.2/-1/x3np/pnn, 1.2/-1.5/x3n/pn,
             -1/0/w1/pn, -1/-.5/w2/pn, -1/-1/w3np/pnn, %
             0/0.25/y1/pn, 0/-0.25/y2/pn, 0/-0.75/y3np/pnn, 0/-1.25/y3n/pn}
           {
             \node[\st] at (\x*0.8, \y) (\n) {};
           }
         \end{scope}

         \begin{scope}[xshift=.5cm, yshift=-.7cm]
           
           \foreach  \x / \y / \n / \st in
           {0/0/s1/pn, 0/-.5/s2/pn, 0/-1/smp/pnn, 0/-1.5/sm/pn}
           {
             \node[\st] at (\x, \y) (\n) {};
           }
         \end{scope}

         \begin{scope}[xshift=2cm, yshift=-.7cm]
           
           \foreach  \x / \y / \n / \st in
           {0/0/t1/pn, 0/-.5/t2/pn, 0/-1/tmp/pnn, 0/-1.5/tm/pn,
             1/0/a1/pn, 1/-.5/a2/pn, 1/-1/amp/pnn, 1/-1.5/am/pn}
           {
             \node[\st] at (\x, \y) (\n) {};
           }
         \end{scope}

         \begin{scope}[xshift=.6cm, yshift=-3cm]
    
           \foreach  \x / \y / \n / \st in
           { 0/0/rm/pn, .5/0/rmp/pnn, 1/0/rmpp/pnn, 1.5/0/r2/pn, 2/0/r1/pn}
           {
             \node[\st] at (\x, \y) (\n) {};
           }
         \end{scope}

         \foreach \n / \nn / \p / \l / \r  / \c in
         {b1/{b_1}/above left/0/0/black, b2/{b_2}/above right/0/0/black, b3/{b_3}/above left/0/0/black, b4/{b_4}/above right/0/0/black, b5/{b_5}/right/0/0/black, y1/{y_1}/above left/0/0/black, y2/{y_2}/above left/0/0/black, y3n/{y_{3\enn}}/above left/0/0/black, w1/{w_1}/above left/1/-1/black, w2/{w_2}/above left/1/-1/black, %
           x1/{x_1}/below/1/0/black, x2/{x_2}/below/0/0/black, x3n/{x_{3\enn}}/below/1/0/black, r1/{r_1}/below/0/0/black, r2/{r_2}/below/0/0/black, rm/{r_{\emm}}/below/0/0/black, s1/{s_1}/above right/0/0/black, s2/{s_2}/above right/0/0/black, sm/{s_{\emm}}/above right/0/0/black, t1/{t_1}/above right/0/0/black, t2/{t_2}/above right/0/0/black, tm/{t_{\emm}}/above right/0/0/black, a1/{t_{\emm+1}}/above right/-1/0/black, a2/{t_{\emm+2}}/above right/0/0/black, am/{t_{8\enn}}/above right/0/0/black}{ 
           \node[\p = \l pt and \r pt of \n, text=\c, inner sep=.5pt, fill=white] {\small ${\nn}$};
         }

         \node at ($(rmp)!0.5!(rmpp)$) {$\ldots$};
         \node[rotate=90] at (smp) {$\ldots$};
         \node[rotate=90, xshift=4pt] at (x3np) {$\ldots$};
         \node[rotate=90,xshift=5pt] at (y3np) {$\ldots$};
         \node[rotate=90,xshift=5pt] at (tmp) {$\ldots$};
         \node[rotate=90,xshift=5pt] at (amp) {$\ldots$};

         \begin{pgfonlayer}{bg}
           \foreach \s / \t / \aa / \type in {b3/b1/5/fc, b3/b2/-5/fc, b3/b4/20/fc, b4/b3/20/fc, b4/b5/0/fc,
             y1/x1/-5/fc, x1/y2/-5/fc, y2/x2/-5/fc, x3np/y3n/-5/fc,
             y3n/x3n/-5/fc,
             y1/y2/0/fc, y3np/y3n/0/fc,
             y1/w1/0/fc, y2/w2/0/fc, %
             x3n/rm/0/fc, rm/rmp/0/fc, rmpp/r2/0/fc, r2/r1/0/fc,
             t1/t2/0/fc, %
             tmp/tm/0/fc, tm/a1/0/fc, a1/a2/0/fc, %
             amp/am/0/fc,
             b1/y1/5/fc, b2/y1/0/fc, b1/t1/0/fc, b2/t1/-5/fc} {
             \draw[->, \type] (\s) edge[bend right = \aa] (\t);
           }

            \foreach \i in {1,2,m} { 
             \draw[->, fc] (s\i) edge[bend right=10] (r\i);
             \draw[->, fc] (s\i) edge[] (t\i);
           }

           \draw[->, fc, rounded corners] (r1.east) .. controls ($(r1)+(.3,0)$) and ($(r1)+(.3,-0.2)$) .. ($(r1)+(0.2, -0.4)$)  .. controls ($(rm)+(-3.8,-0.4)$) and ($(rm)+(-3.8,-0.4)$) .. ($(rm)+(-4, 0.4)$) %
           .. controls ($(rm)+(-4.1, 2.8)$) and ($(rm)+(-4.1, 2.9)$)
           .. (b1.west);

           \draw[->, fc, rounded corners] (am.south) .. controls ($(am)+(0, -0.4)$) and ($(am)+(1.2, -0.4)$) .. ($(am)+(1.2, 0.2)$) .. controls ($(a1)+(1.2, 0.5)$) and ($(a1)+(1.2, .8)$)   .. (b2.east);

           \draw[->] (x1) edge (s2);
           \draw[->] (x1) edge ($(s2)!0.5!(smp)$);
           \draw[->] (x1) edge ($(s2)!0.9!(smp)$);

           \draw[->] (x2) edge (s1);
           \draw[->] (x2) edge (smp);
           
           \draw[->] (x3n) edge (smp);
           \draw[->] (x3n) edge ($(smp)!0.5!(sm)$);
           \draw[->] (x3n) edge (sm);
         \end{pgfonlayer}
     
       \end{tikzpicture}\caption{Illustration for the proof of \cref{thm:fe-nash-fas-delta-nph}.}\label{fig:fe-nash-fas-delta-nph}
     \end{figure}

     For the ``only if'' part, assume that $\mathcal{K}$ is an exact cover for~$I$.
     Construct a partition~$\Pi \coloneqq \{\{b_5\} , \{b_3,b_4\}\} \cup \{\{w_i\} \mid i\in [3\enn-1]\} \cup \{\{b_1, y_1,x_1,\ldots, y_{3\enn},x_{3\enn}, r_{\emm}, \ldots, r_1\} \cup \{s_j \mid C_j \in \mathcal{K}\}\} \cup \{\{b_2, t_1,\ldots,$ $t_{8\enn}\} \cup \{s_j\mid C_j\in \mathcal{C}\setminus \mathcal{K}\}\}$.
      We show that $\Pi$ is Nash stable.
      First of all, it is individually rational since each agent is either alone or in a coalition with at least one friend. 
      Next, no agent~$\{b_5, w_1, \ldots, w_{3\enn-1}\}$ envies any coalition since each of them is alone and has no friends.
      Consequently, agent~$b_4$ does not envy any coalition since she has in total two friends and is with her one friend in a size-two coalition.
      Agent~$b_3$ is in a size-two coalition which contains one of her friends.
      The other coalitions each contain at most one friend.
      Thus, she does not envy any coalition.      
      Neither does any agent from $\{r_j\mid j\in [\emm]\}\cup \{t_j\mid j\in [8\enn]\}\cup \{y_{3\enn}\}$ envy some other coalition since each of them has only one friend and is in the same coalition as the friend.
      No agent from $\{y_i\mid i\in [3\enn-1]\}$ envies any coalition since each of them has in total three friends and is in a coalition with two friends for her.
    
      No agent from $\{b_1,b_2, s_1, \ldots, s_{\emm}\}$ envies any coalition since each of them has in total two friends and is in a coalition with exactly one friend, and the other coalition which contains her other friend has one more enemy for her; note that $|\Pi(b_1)|=|\Pi(b_2)|=7\enn+\emm+1$.
      Similarly, no agent $x_i$, $ i\in [3\enn]$ envies any coalition since each of them has in total at most four friends and is in a coalition with two friends (namely, $y_{i+1}$ or $r_{\emm}$) and some set-agent~$s_j$ since $\mathcal{K}$ is a set cover),
      and each other coalition contains at most two friends but has at least one more enemy.
      Since no agent envies any other coalitions, the partition is indeed Nash stable.

      For the ``if'' part, let $\Pi$ be a Nash stable partition.
      We claim that the sets which correspond to the set-agents contained in $\Pi(b_1)$ form an exact cover. %
      Before showing this, we first observe the following.
      \begin{clm}\label{clm:fe-nash-delta-fas}
        \begin{compactenum}[(i)]
          \item\label{fe-nash-delta-fas:b345} $\Pi(b_5)=\{b_5\}$, $\Pi(b_3)=\{b_3, b_4\}$, and $\Pi(w_i)=\{w_i\}$ for all $i\in [3\enn-1]$.
          \item\label{fe-nash-delta-fas:b1b2} $\Pi(b_1)\neq \Pi(b_2)$.
          \item\label{fe-nash-delta-fas:t} $\{t_1,\ldots, t_{8\enn}\}\subseteq \Pi(b_2)$.
          \item\label{fe-nash-delta-fas:x} $\{y_1,x_1,\ldots, y_{3\enn},x_{3\enn}, r_{\emm},\ldots,r_{1}\}\subseteq \Pi(b_1)$. 
          \item\label{fe-nash-delta-fas:s} For each $j\in [\emm]$, either~$s_j\in \Pi(b_1)$ or
          $s_j\in \Pi(b_2)$ holds.
        \end{compactenum}
      \end{clm}

      \begin{proof}[Proof of \cref{clm:fe-nash-delta-fas}]
        \renewcommand{\qedsymbol}{$\diamond$}
        Statement~\eqref{fe-nash-delta-fas:b345}: Since, by definition, an agent without any friends must be alone, we have that  $\Pi(b_5)=\{b_5\}$, and  $\Pi(w_i)=\{w_i\}$ for all $i\in [3\enn-1]$.
        Then, since $b_4$ considers $b_5$ a friend and has in total two friends, in order for agent~$b_4$ to not envy $\Pi(b_5)$, she must be in a coalition which consists of herself and her other friend~$b_3$.
        That is, $\Pi(b_4)=\{b_3, b_4\}$.

        Statement~\eqref{fe-nash-delta-fas:b1b2}: If $\Pi(b_1)\cap \Pi(b_2)\neq \emptyset$, meaning that $b_1$ and $b_2$ are in the same coalition, then $b_3$ will envy this coalition since she has only one friend in her coalition (see the first statement), a contradiction.

        Statement~\eqref{fe-nash-delta-fas:t}: Since an agent who has only one friend must be in the same coalition as her friend, we infer that every agent~$t_i$ , $i\in [8\enn-1]$, must be in the same coalition as $t_{i+1}$,
        and agent $t_{8\enn}$ must be in the same coalition as $b_2$, as desired.

        Statement~\eqref{fe-nash-delta-fas:x}: First of all, by a similar reasoning, we have that $\{r_{\emm}, \ldots,r_1\}\subseteq \Pi(b_1)$. By Statements~\eqref{fe-nash-delta-fas:b1b2}--\eqref{fe-nash-delta-fas:t}, we have that $y_1\in \Pi(b_1)$ as otherwise $b_1$ envies $\Pi(b_2)$.
        This means that $y_1$ is in a coalition with at least one enemy for her. 
        In order to let $y_1$ not envy $\Pi(w_1)$ (see Statement~\eqref{fe-nash-delta-fas:b345}), she must have two friends in her coalition.
        By construction, we have $x_1, y_2\in \Pi(y_1)=\Pi(b_1)$.
        Analogously, we have that $x_i, y_{i+1}\in \Pi(b_1)$ for all~$i\in [3\enn-1]$.
        Finally, $x_{3\enn}\in \Pi(b_1)$ since $x_{3\enn}$ is the only friend of $y_{3\enn}$ and $y_{3\enn}$ is in $\Pi(b_1)$ which contains at least one enemy for her.

        Statement~\eqref{fe-nash-delta-fas:s}: This follows directly from Statements~\eqref{fe-nash-delta-fas:b1b2}--\eqref{fe-nash-delta-fas:x} and from the fact that each agent~$s_j$, $j\in [\emm]$ considers $r_j$ and $t_j$ friends.
      \end{proof}

      Now, let $\mathcal{K}=\{C_j\mid s_j\in \Pi(b_1)\}$, and we show that $\mathcal{K}$ is a set cover (i.e., it covers each element at least once) and has size $\enn$.
      Suppose, for the contradiction of $\mathcal{K}$ being a set cover, that element~$i\in [3\enn]$ is not covered by $\mathcal{K}$.
      By construction this means that for each set~$C_j\in \mathcal{C}$ with $i\in C_j$ it holds that $s_j\notin \Pi(b_1)$.
      By \cref{clm:fe-nash-delta-fas}\eqref{fe-nash-delta-fas:s}, it follows that $\Pi(b_2)$ contains all set-agents~$s_j$ such that $i\in C_j$, implying that $\Pi(b_2)$ contains at least two friends for $x_i$ (see the assumption in the beginning).
      Since $x_i$ is in $\Pi(b_1)$ (see \cref{clm:fe-nash-delta-fas}\eqref{fe-nash-delta-fas:x}), we have that $x_i$ has only one friend in her coalition and she envies $\Pi(b_2)$, a contradiction.

      It remains to show that $|\mathcal{K}|=\enn$.
      Clearly, $|\mathcal{K}|\ge \enn$ since $\mathcal{K}$ is a set cover. 
      Suppose, towards a contradiction, that $|\mathcal{K}|>\enn$.
      By \cref{clm:fe-nash-delta-fas}\eqref{fe-nash-delta-fas:s}, at least $\enn+1$ set-agents are contained in $\Pi(b_1)$, i.e., 
      $|\Pi(b_1)| \ge 2 \cdot 3\enn + 1 + \emm + \enn + 1 =7\enn+\emm + 2$.
      This means that there exists at least one set-agent~$s_j\in \Pi(b_1)$ who has one friend and at least $7\enn+\emm$ enemies in her coalition.
      This agent will envy $\Pi(b_2)$ since $\Pi(b_2)$ contains a friend (namely $r_j$) and at most $1 + 8\enn - 1 + \emm - \enn - 1=7\enn+\emm-1$ enemies (see \cref{clm:fe-nash-delta-fas}\eqref{fe-nash-delta-fas:b345}--\eqref{fe-nash-delta-fas:x}), a contradiction.
    }
    \end{proof}

\section{The Model with Neutrals}\label{sec:fen}
\appendixsection{sec:fen} 
In this section, we consider the model with neutrals~\cite{ohta2017core}.
First, we observe that for acyclic friendship graphs, checking core stability is easy since it is equivalent to checking individual rationality.

\newcommand{\dagfriendcoreeasy}{%
  For acyclic friendship graphs, \FEN-\verif\ is linear-time solvable.
}
\begin{proposition}[\appendixsymb]\label{prop:dag-friends-core-P}
  \dagfriendcoreeasy
\end{proposition}
\appendixproofwithstatement{prop:dag-friends-core-P}{\dagfriendcoreeasy}{
  \begin{proof}
    Let $(V, \goodG, \badG, \Pi)$ be an \FEN-\verif-instance where $\goodG$ is acyclic. %

 We show that $\Pi$ is core stable if and only if it is individually rational.

For the ``only if'' part, assume that $\Pi$ is core stable. Then, no agent forms a blocking coalition by itself, so $\Pi$ is individually rational.

For the ``if'' part, suppose that $\Pi $ is individually rational but there is a strictly blocking coalition $U$. Then, every agent in $U$ must also have a friend in $U$, otherwise they cannot strictly improve, as $\Pi$ was individually rational. Therefore, the friendship graph $\goodG[U]$ induced by $U$ must have at least one outgoing arc for each vertex, contradicting the fact that $\goodG$ is acyclic.

Linear-time solvability of \verif when the friendship graph is acyclic follows from the fact that checking whether a coalition structure is individually rational can be done efficiently by checking whether an agent can improve by being alone.
\end{proof}
}

For acyclic graphs, \IndS\ and \NS\ can be solved by 
a clever greedy algorithm operated on the reverse topological order.

\newcommand{\fendagnashindp}{
  If the friendship graph (resp.\ the union graph) is acyclic, then
  every \FEN-instance admits an individually stable (resp.\ Nash stable) partition,
  which can be found in polynomial time. 
}
\begin{theorem}[\appsymb]
\label{thm:ns_is_acyclic}
\fendagnashindp
\end{theorem}
  \begin{proof}[Proof Sketch]%

\appendixalg{thm:ns_is_acyclic}{    Let $(V, \goodG, \badG)$ be an \FEN-instance.
    We first consider  \FEN-\IndS and assume that $\goodG $ is acyclic and thus has a topological order $v_1, \dots, v_n$ of $V$.
    The algorithm proceeds as follows:
    Iterate over $V$ in the reverse topological order $v_n, \dots, v_1$.
    In each step, check whether there exists a coalition~$U$ where $v_i$ has at least one friend and no one in~$U$ considers her an enemy.
    If no such coalition exists, then $v_i$ starts a new coalition.
    Otherwise, let $v_i$ join the most preferred coalition~$U$ among all such coalitions.

Now, we turn to \FEN-\NS and assume that $\goodG \cup \badG$ is acyclic and thus has a topological order $v_1, \dots, v_n$ of $V$.
The algorithm proceeds as follows: We iterate over $V$ in the reverse topological order $v_n, \dots, v_1$. In each step we let the current agent $v_i$ join her most preferred existing coalition, or in the case when $v_i$ has no friends in any, to start a new one.

The correctness of both algorithms relies on iterating through the agents in the topological order. Each of the agents selects in her turn her most preferred feasible coalition, and due to the ordering no agent will change her choice about her most preferred coalition later in the execution.
\iflong
The details of the correctness and the running time are deferred to the full version.
\else
The details of the correctness and the running time are deferred to the appendix.
\fi}{\fendagnashindp}{
\smallskip

\noindent Throughout this section, let $(V, \goodG, \badG)$ be an \FEN-instance.
\paragraph*{Correctness of the algorithm for \FEN-\IndS}
Let $v_i \in V$ be an arbitrary agent and $\Pi(v_i)$ her coalition at the end of the algorithm. For the correctness, we show that $v_i$ is an enemy to someone in every coalition that she strictly prefers.
To this end, let $P'$ be a coalition constructed by the algorithm such that $v_i$ strictly prefers~$P'\cup \{v_i\}$ to~$\Pi(v_i)$.
Let $P=P'\cup \{v_i\}$ and $V_i = \{v_1, \dots, v_{i - 1}\}$. 
    Because the algorithm works on the reverse topological order of $\goodG$,
    implying that $v_i$ does not have any friend in~$V_i$,
    for every coalition~$Q$ constructed by the algorithm,
    $|N^+_{\goodG}(v_i) \cap Q| = |N^+_{\goodG}(v_i) \cap (Q \setminus V_i)|$.
    Since the algorithm dictates that no enemy of $v_i$ joins $v_i$'s coalition after her, we have that $|N^+_{\badG}(v_i) \cap \Pi(v_i)| = |N^+_{\badG}(v_i) \cap( \Pi(v_i) \setminus V_i)|$.
    Therefore,~$v_i$ weakly prefers $P \setminus V_i$ to $P$, and is indifferent between $\Pi(v_i)$ and $\Pi(v_i) \setminus V_i$. 
    Since $v_i$ strictly prefers $P$ to $\Pi(v_i)$,
    it follows that $v_i$ also strictly prefers $P \setminus V_i$ to $\Pi(v_i) \setminus V_i$.
    Then, the only reason why $v_i$ has not picked $P\setminus V_i$ over $\Pi(v_i)\setminus V_i$ (on the iteration for $v_i$) is that there is an agent $w$ in $P\setminus V_i$, who considers $v_i$ an enemy, as desired.
    
\paragraph*{Correctness of the algorithm for \FEN-\NS}
For correctness, we show that no agent $v_i \in V$ strictly prefers any coalition to her  coalition. Let $\Pi$ be the partition returned by the algorithm. Assume towards a contradiction that there exists an agent~$v_i$ and a coalition~$P'\in \Pi$ such that
$v_i$ prefers~$P'\cup \{v_i\}$ to her own one~$\Pi(v_i)$ and let $P=P'\cup \{v_i\}$.
Because $\goodG \cup \badG$ is acyclic, $v_i$ is neutral towards $V_i = \{v_1, \dots, v_{i - 1}\}$. Thus, if $v_i$ strictly prefers $P$ to $\Pi(v_i)$, then she must strictly prefer $P \setminus V_i$ to $\Pi(v_i) \setminus V_i$.
However, to join $\Pi(v_i)$ in the algorithm execution, either $P\setminus V_i$ contains no friends for $v_i$ or $v_i$ strictly prefers $\Pi(v_i)\setminus V_i$ to $P\setminus V_i$, a contradiction in the latter case. In the former case, $\Pi(v_i)\setminus V_i$ is a singleton so $v_i$ does not strictly prefer $P$ to $\Pi(v_i)$, a contradiction.

\paragraph*{Running time}

Finally, both algorithms run in polynomial time: %
we first find a topological order of the vertices, which can be done in linear time. Then, we iterate over the agents once. In each iteration the current agent goes over all her neighbors to find her most preferred coalition, or the most preferred one in which she is not an enemy to anyone. Thus the algorithms takes $O(n^2\cdot m)$-time, where $m$ is the number of friendship and enemy relations.}
\end{proof}

It is know that for symmetric and additive separable preferences, Nash stable partitions always exist~\cite{bogomolnaia2002stability}.
For the \FEN-model, we can even find one in linear time, which consists of singletons who do not have any friends and the remaining agents in a grand coalition.

\newcommand{\symmnashind}{%
For symmetric friendship relations, Nash (and hence individually) stable partitions can be found in linear time. %
}
\begin{observation}[\appendixsymb]
\label{thm:sym-ns-is}
\symmnashind
\end{observation}
\appendixproofwithstatement{thm:sym-ns-is}{\symmnashind}{
\begin{proof}
We provide a simply way to find a Nash-stable partition in any such instance, which must also be individually stable. 

Let $(V, \goodG, \badG)$ be an \FEN-instance.

Let $S \subseteq V$ be the set of agents who have no friends, i.e, $s \in V$ for which $|N^+_{\goodG}(s)| = 0$.
We claim that the coalition structure $\Pi =\bigcup_{s\in S}\{ \{s\}\} \cup \{ V\setminus S\}$ is Nash stable. 

Let $s \in S$ be arbitrary. Because $|N^+_{\goodG}(s)| = 0$, no coalition has more friends than $\Pi(s)$. Because $\Pi(s) = \{s\}$, $s$ has no enemies in $\Pi(s)$ and thus no coalition can have fewer enemies. Therefore $s$ does not strictly prefer any coalition to $\Pi(s)$.

Let $v \in V \setminus S$ be arbitrary. By construction $v$ must have at least one friend in $\Pi(v) = V \setminus S$. Because the relationships are symmetric, $v$ does not have friends in any of the coalitions $\{s\}, s \in S$. Therefore there is no coalition that $v$ prefers to $\Pi(v)$.
This concludes that $\Pi$ is Nash stable.

As the set $S$ can easily be constructed in linear time, the theorem follows.
\end{proof}
}

\noindent The following  complements~\cref{prop:dag-friends-core-P} regarding~$\fas$ and show that both core verification problems remain hard even if $\fas, \maxdeg$ and $\maxcoal$ are bounded. %
\newcommand{\fencorethreepara}{%
\FEN-\verif\ (resp.\ \FEN-\sverif) is \conp-com\-plete even if $\maxdeg=12$, $\maxcoal=3$, and $\fas=1$ (resp.\ $\fas=0$).
}
\begin{theorem}[\appendixsymb]
\label{thm:verif-fas+deg+k}
\fencorethreepara%
\end{theorem}
\appendixproofwithstatement{thm:verif-fas+deg+k}{\fencorethreepara}{
  \begin{proof}
We first consider \FEN-\verif and reduce from \pxct. Let the elements be $\mathcal{X}=\{ 1,2\dots,3\enn \}$ and the sets be $\mathcal{C} = \{C_1,\dots,C_{\emm}\}$.

We construct an instance $I'$ of \FEN-\verif\ as follows:
\begin{compactitem}[--]
    \item For each element $i\in [3\enn ]$ we have an element agent $a_i$.
    \item For each set $C_j$, $j\in [\emm]$ we have a set agent $c_j$ and two private agents $d_j^1,d_j^2$ and a forwarder agent $x_j$.
    \item We add $s_1,\dots s_{3\enn }$ special agents and for each $s_l$, $l\in [3\enn -1]$ we add a private agent $t_l$. 
\end{compactitem}
Next, we describe the friendship arcs.
$E(\goodG)=\{ (s_l,s_{l+1}), (s_l,t_l),$ $(s_l,a_l)\mid l\in [3\enn -1]\} \cup \{ (s_{3\enn },a_{3\enn })\} \cup \{ (a_i,c_j)\mid i,j:i\in C_j\} \cup \{ (c_j,d_j^1),(d_j^1,d_j^2),(c_j,x_j)\mid j\in [\emm]\} \cup \{ (x_j,x_{j+1}\mid j\in [\emm -1]\} \cup \{ (x_{\emm },s_1)\} $.

The enemy arcs are:
$E(\badG) = \{ (c_j,d_j^2)\mid j\in [\emm ]\} \cup \{ (c_j,c_z)\mid j,z: j<z, C_j\cap C_z\ne \emptyset \}$.
     \begin{figure}
     \centering
       \begin{tikzpicture}[scale=0.85,every node/.style={scale=0.9}, >=stealth', shorten <= 2pt, shorten >= 2pt]
         
      \foreach \x / \y / \n / \nn / \typ / \p / \dx / \dy in {
      1/4/s1/s_1/pn/{above left}/-1/-4,
      4/4/si/s_i/pn/{below left}/-1/-4,
      5/4/si1/s_{i+1}/pn/{below left}/-1/-4,
      8/4/s3n/s_{3\enn}/pn/{below left}/-1/-4,  
      2/4/ps2/\;/pnn/{below left}/0/0,
      3/4/ps3/\;/pnn/{below left}/0/0,
      6/4/psi2/\;/pnn/{below left}/0/0,
      7/4/psi3/\;/pnn/{below left}/0/0,   
      1/3/a1/a_1/pn/{left}/-1/-4,
      4/3/ai/a_i/pn/{above left}/-1/-2,
      5/3/ai1/a_{i+1}/pn/{above right}/-1/-2,
      8/3/a3n/a_{3\enn}/pn/{left}/-1/-4,
      9/1/c1/c_{1}/pn/{below left}/-1/-4,
      7/1/cj/c_{j}/pn/{below left}/-1/-4,
      5/1/ck/c_{z}/pn/{below left}/-1/-2,
      4/1/ck1/c_{z + 1}/pn/{left}/-1/-2,
      2/1/cl/c_{l}/pn/{below left}/-1/-4,
      0/1/cm/c_{\emm}/pn/{above right}/-1/-4,
      0.5/2/pa1c1/\;/pnn/{left}/0/0,
      1/2/pa1c2/\;/pnn/{left}/0/0,
      1.5/2/pa1c3/\;/pnn/{left}/0/0,
      5/2.5/pai1c1/\;/pnn/{left}/0/0,
      6/2/pai1c2/\;/pnn/{left}/0/0,
      6.5/2/pai1c3/\;/pnn/{left}/0/0,
      7.5/2/pa3nc1/\;/pnn/{left}/0/0,
      8/2/pa3nc2/\;/pnn/{left}/0/0,
      8.5/2/pa3nc3/\;/pnn/{left}/0/0,
      4/2/pcka1/\;/pnn/{left}/0/0,
      5/2/pcka2/\;/pnn/{left}/0/0,
      9/-0.5/x1/x_{1}/pn/{below left}/-1/-4,
      7/-0.5/xj/x_{j}/pn/{below left}/-1/-4,
      5/-0.5/xk/x_{z}/pn/{below left}/-1/-4,
      4/-0.5/xk1/x_{z + 1}/pn/{below left}/-1/-4,
      2/-0.5/xl/x_{l}/pn/{below left}/-1/-4,
      0/-0.5/xm/x_{\emm}/pn/{below left}/-1/-4,
      8.2/-0.5/px2/\;/pnn/{below left}/-1/-4,
      7.8/-0.5/px3/\;/pnn/{below left}/-1/-4,
      6.2/-0.5/px4/\;/pnn/{below left}/-1/-4,
      5.8/-0.5/px5/\;/pnn/{below left}/-1/-4,
      3.2/-0.5/px6/\;/pnn/{below left}/-1/-4,
      2.8/-0.5/px7/\;/pnn/{below left}/-1/-4,
      1.2/-0.5/px8/\;/pnn/{below left}/-1/-4,
      0.8/-0.5/px9/\;/pnn/{below left}/-1/-4} {
        \node[\typ] at (\x,\y) (\n) {};
        \node[\p = \dx pt and \dy pt of \n] {$\nn$};
      }      
      
            \foreach \x / \y / \pn / \nn in {
      1/4.7/s1/1f,
      4/4.7/si/1f,
      5/4.7/si1/1f,
      8.2/1.5/c1/{1f,1e},
      7.6/0.2/cj/{1f,1e},
      5.6/0.2/ck/{1f,1e},
      3.4/0.2/ck1/{1f,1e},
      1.4/1.5/cl/{1f,1e},
      0.6/0.2/cm/{1f,1e}} {
        \node[privaten] at (\x,\y) (p\pn) {\nn};
        \draw[->, privatee] (\pn) edge[] (p\pn);
      }  

		\node at ($(ps2)!0.5!(ps3)$) {$\ldots$};
		\node at ($(psi2)!0.5!(psi3)$) {$\ldots$};
		\node at ($(px2)!0.5!(px3)$) {$\ldots$};
		\node at ($(px4)!0.5!(px5)$) {$\ldots$};
		\node at ($(px6)!0.5!(px7)$) {$\ldots$};
		\node at ($(px8)!0.5!(px9)$) {$\ldots$};

      \begin{pgfonlayer}{bg}
           \foreach \s / \t / \aa / \type in {
           s1/ps2/0/hiddenfc,
           ps3/si/0/hiddenfc,
           si/si1/0/fc,
           si1/psi2/0/hiddenfc,
           psi3/s3n/0/hiddenfc,
           s1/a1/0/fc,
           si/ai/0/fc,
           si1/ai1/0/fc,
           s3n/a3n/0/fc,
           ai/cj/0/fc,
           ai/ck/0/fc,
           ai/cl/0/fc,
           a1/pa1c1/0/hiddenfc,
           a1/pa1c2/0/hiddenfc,
           a1/pa1c3/0/hiddenfc,
           ai1/pai1c1/0/hiddenfc,
           ai1/pai1c2/0/hiddenfc,
           ai1/pai1c3/0/hiddenfc,
           a3n/pa3nc1/0/hiddenfc,
           a3n/pa3nc2/0/hiddenfc,
           a3n/pa3nc3/0/hiddenfc,
           pcka1/ck/0/hiddenfc,
           pcka2/ck/0/hiddenfc,
           cj/ck/0/ec,
           ck/cl/20/ec,
           x1/px2/0/hiddenfc,
           px3/xj/0/hiddenfc,
           xj/px4/0/hiddenfc,
           px5/xk/0/hiddenfc,
           xk/xk1/0/fc,
           xk1/px6/0/hiddenfc,
           px7/xl/0/hiddenfc,
           xl/px8/0/hiddenfc,
           px9/xm/0/hiddenfc,
           c1/x1/0/fc,
           cj/xj/0/fc,
           ck/xk/0/fc,
           ck1/xk1/0/fc,
           cl/xl/0/fc,
           cm/xm/0/fc,
           xm/s1/-40/fc} {
             \draw[->, \type] (\s) edge[bend right = \aa] (\t);
           }

         \end{pgfonlayer}
       \end{tikzpicture}\caption{Illustration for the proof of \cref{thm:verif-fas+deg+k} for \FE-\verif.
         Here $a_i$ is contained in the sets $C_j, C_z, C_l$, where $j < z < l$. The green rectangles indicate the private agents and the agent's relationship to them.}\label{fig:core-fas+deg+k}
     \end{figure}

The construction is illustrated in \cref{fig:core-fas+deg+k}.
The initial partition $\Pi$ is $\bigcup_{l\in [3\enn -1]}\{\{ s_l,t_l\}\} \cup \{\{ s_{3\enn }\}\} \bigcup_{i\in [3\enn ]}\{\{ a_i\}\} \bigcup_{j\in [\emm]} \{\{ c_j,d_j^1,d_j^2\}\}$ $\bigcup_{j\in [\emm]}\{\{ x_j\}\}$.

One can verify that the degree of each agent is at most 12 and each initial coalition has size at most 3. To show that the graph has feedback arc set number 1, remove the following arc: $ (x_{\emm},s_1)$. Then the following order is a topological order for the remaining graph: $s_1,t_1,s_2,t_2,\dots,s_{3\enn},a_1\dots,a_{3\enn},$ $c_1,\dots,c_{\emm},d_1^1,\dots,d_{\emm}^1,d_1^2,\dots,d_{\emm}^2$,\\
$x_1\dots,x_{\emm}$.

\begin{observation}
\label{obs:core-fas+dag+k}
It holds in $\Pi$ that:
\begin{compactenum}[(i)]
  \item \label{core-fas+deg+k,i} for each $i\in [3\enn ]$ $a_i$ has 0 friends and 0 enemies,
  \item \label{core-fas+deg+k,ii} for each $j\in [\emm]$ $c_j$ has 1 friend and 1 enemy, $d_j^1$ has 1 friend and 0 enemies and $d_j^2$ and $x_j$ has 0 friends and 0 enemies,
  \item \label{core-fas+deg+k,iii} for $l\in [3\enn -1]$ $s_l$ has 1 friend and 0 enemies and $t_l$ has 0 friends and 0 enemies and $s_{3\enn }$ has no friends or enemies.
\end{compactenum}

\end{observation}

\begin{clm}
\label{claim:core:fas+deg+k-1}
If there is an exact cover, then there is a strictly blocking coalition to $\Pi$.
\end{clm}
\begin{proof}[Proof of \cref{claim:core:fas+deg+k-1}]
\renewcommand{\qedsymbol}{$\diamond$}
Let $C_{j_1},\dots,C_{j_{\enn}}$ be the exact cover. Take the following coalition $P=\{ s_i,a_i,c_j,x_l\mid i\in [3\enn],j\in \{ j_1,\dots,j_{\enn}\},$ $ l\in [\emm] \}$.

We claim that $P$ is strictly blocking. Each $a_i$ and $x_l$ agent for $i\in [3\enn] ,l\in [\emm]$ have 1 friend in $P$, therefore by \cref{obs:core-fas+dag+k}(\ref{core-fas+deg+k,i})-(\ref{core-fas+deg+k,ii}), they strictly improve.

For each $l\in [3\enn-1]$, $s_l$ has 2 friends and 0 enemies, while $s_{3\enn }$ has 1 friend in $P$ and therefore by \cref{obs:core-fas+dag+k}(\ref{core-fas+deg+k,iii}) they strictly improve too. 

For each $i\in [\enn]$, $c_{j_i}$ has 1 friend and 0 enemies in $P$ because $C_{j_1}, \dots, C_{j_n}$ is an exact cover. By \cref{obs:core-fas+dag+k}(\ref{core-fas+deg+k,ii}) they also strictly improve.
\end{proof}

\begin{clm}
For any strictly blocking coalition $P$
\label{claim:fas+deg+k-2}
\begin{compactenum}[(i)]
\item \label{fas+deg+k,i}no private agent ($t_l$ or $d_j^z)$ can be included for each $l\in [3\enn-1]$, $j \in [\emm]$, and $z \in [2]$,
\item \label{fas+deg+k,ii}for $l\in [3\enn -1]$, $s_l$ can only strictly improve if $s_{l+1}, a_l \in P$ and $s_{3\enn } $ can only strictly improve if $a_{3\enn } \in P$,
\item \label{fas+deg+k,iii} for $i\in [3\enn ]$, $a_i$ can only strictly improve if there is exactly one $c_j \in P$ such that $i\in C_j$.

\end{compactenum}
\end{clm}
\begin{proof}[Proof of claim \ref{claim:fas+deg+k-2}]
\renewcommand{\qedsymbol}{$\diamond$}
To prove (\ref{fas+deg+k,i}), we only have to observe that all private agents have all their friends and none of their enemies in $\Pi$, so they cannot strictly improve. 

For (\ref{fas+deg+k,ii}), by observation \ref{obs:core-fas+dag+k}(\ref{core-fas+deg+k,iii}), $s_l$ has to get 2 friends in $P$, and none of them can be a private friend by (\ref{fas+deg+k,i}),  so  they must be $s_{l+1}$ and $a_l$. In the case of $s_{3\enn }$, she must get her  only friend $a_{3\enn }$.

By (\ref{fas+deg+k,i}) and observation \ref{obs:core-fas+dag+k}(\ref{core-fas+deg+k,i}), $a_i$ must have a friend that can only be a set agent. Suppose she has 2 set agents as friends in $P$: $c_j$ and $c_z$, for some $j<z$. Then, it must hold that $C_j\cap C_z\ne \emptyset$, so $c_j$ has an enemy in $P$. As $c_j$ has only 1 non-private friend $x_j$, by (\ref{fas+deg+k,i}) and \cref{obs:core-fas+dag+k}(\ref{core-fas+deg+k,ii}) she cannot strictly improve, proving (\ref{fas+deg+k,iii}). 
\end{proof}

\begin{clm}
\label{claim:k+fas+deg-3}
If there is a strictly blocking coalition, then there is an exact cover. 
\end{clm}
\begin{proof}[Proof of \cref{claim:k+fas+deg-3}]
\renewcommand{\qedsymbol}{$\diamond$}
As $\Pi$ is individually rational by \cref{obs:core-fas+dag+k}, in any strictly blocking coalition $P$ each agent must get a friend. As $E(G)\setminus \{ (x_{\emm},s_1)\}$ is acyclic, so is $G\setminus \{ s_1\}$. Therefore we obtain that $s_1$ has to be in $P$, because otherwise $G[P]$ would be acyclic, contradicting our above observation.

By iteratively using \cref{claim:fas+deg+k-2}(\ref{fas+deg+k,ii}), we get that $\{ s_1,\dots, s_{3\enn },$ $a_1$, $\dots,a_{3\enn }\} \subset P$.

By \cref{claim:fas+deg+k-2}(\ref{fas+deg+k,iii}), we get that each $a_i$ must obtain exactly one set agent $c_j$ as a friend. By the definition of the arcs, this means that the set agents in $P$ satisfy that their corresponding sets form an exact cover.
\end{proof}

\noindent The correctness follows from Claims~\ref{claim:core:fas+deg+k-1} \& \ref{claim:k+fas+deg-3}.

Now we prove that \FEN-\sverif\ remains hard, even if the union of the friendship arcs and enemy arcs forms a DAG.
We reduce from \xctg . Let the elements be $\mathcal{X}=\{ 1,2\dots,3\enn \}$ and the sets be $C_1,\dots,C_{\emm}$.

We construct an instance $I'$ of \FEN-\sverif\ as follows: %
     \begin{figure}[t!]
     \centering
       \begin{tikzpicture}[scale=1,every node/.style={scale=0.9}, >=stealth', shorten <= 2pt, shorten >= 2pt]
         
      \foreach \x / \y / \n / \nn / \typ / \p / \dx / \dy in {
      1/4/s1/s_1/pn/{below left}/-1/-4,
      4/4/si/s_i/pn/{below left}/-1/-4,
      5/4/si1/s_{i+1}/pn/{below left}/-1/-4,
      8/4/s3n/s_{3\enn}/pn/{below left}/-1/-3,  
      2/4/ps2/\;/pnn/{below left}/0/0,
      3/4/ps3/\;/pnn/{below left}/0/0,
      6/4/psi2/\;/pnn/{below left}/0/0,
      7/4/psi3/\;/pnn/{below left}/0/0,   
      1/3/a1/a_1/pn/{left}/-1/-4,
      4/3/ai/a_i/pn/{above left}/-1/-2,
      5/3/ai1/a_{i+1}/pn/{above right}/1/-2,
      8/3/a3n/a_{3\enn}/pn/{right}/-1/-2,
      2/1/cj/c_{j}/pn/{below left}/-1/-4,
      5/1/ck/c_{z}/pn/{below left}/-1/-2,
      7/1/cl/c_{l}/pn/{below left}/-1/-4,
      0.2/4.7/g/g/pn/{above left}/-1/-4,
      0.5/2/pa1c1/\;/pnn/{left}/0/0,
      1/2/pa1c2/\;/pnn/{left}/0/0,
      1.5/2/pa1c3/\;/pnn/{left}/0/0,
      5/2.5/pai1c1/\;/pnn/{left}/0/0,
      6/2/pai1c2/\;/pnn/{left}/0/0,
      6.5/2/pai1c3/\;/pnn/{left}/0/0,
      7.5/2/pa3nc1/\;/pnn/{left}/0/0,
      8/2/pa3nc2/\;/pnn/{left}/0/0,
      8.5/2/pa3nc3/\;/pnn/{left}/0/0,
      4/2/pcka1/\;/pnn/{left}/0/0,
      5/2/pcka2/\;/pnn/{left}/0/0} {
        \node[\typ] at (\x,\y) (\n) {};
        \node[\p = \dx pt and \dy pt of \n] {$\nn$};
      }      
      
            \foreach \x / \y / \pn / \nn in {
      1/4.7/s1/2f,
      4/4.7/si/2f,
      5/4.7/si1/2f,
      8/4.7/s3n/1f,
      1.8/3/a1/1f,
      3.2/3/ai/1f,
      5.8/3/ai1/1f,
      7.2/3/a3n/1f,
      1/5.5/g/{1f,1e}} {
        \node[privaten] at (\x,\y) (p\pn) {\nn};
        \draw[->, privatee] (\pn) edge[] (p\pn);
      }

		\node at ($(ps2)!0.5!(ps3)$) {$\ldots$};
		\node at ($(psi2)!0.5!(psi3)$) {$\ldots$};

      \begin{pgfonlayer}{bg}
           \foreach \s / \t / \aa / \type in {
           s1/ps2/0/hiddenfc,
           ps3/si/0/hiddenfc,
           si/si1/0/fc,
           si1/psi2/0/hiddenfc,
           psi3/s3n/0/hiddenfc,
           s1/a1/0/fc,
           si/ai/0/fc,
           si1/ai1/0/fc,
           s3n/a3n/0/fc,
           ai/cj/0/fc,
           ai/ck/0/fc,
           ai/cl/0/fc,
           a1/pa1c1/0/hiddenfc,
           a1/pa1c2/0/hiddenfc,
           a1/pa1c3/0/hiddenfc,
           ai1/pai1c1/0/hiddenfc,
           ai1/pai1c2/0/hiddenfc,
           ai1/pai1c3/0/hiddenfc,
           a3n/pa3nc1/0/hiddenfc,
           a3n/pa3nc2/0/hiddenfc,
           a3n/pa3nc3/0/hiddenfc,
           pcka1/ck/0/hiddenfc,
           pcka2/ck/0/hiddenfc,
           cj/ck/0/ec,
           ck/cl/0/ec,
           g/s1/0/fc} {
             \draw[->, \type] (\s) edge[bend right = \aa] (\t);
           }

         \end{pgfonlayer}
       \end{tikzpicture}\caption{Illustration for the proof of \cref{thm:verif-fas+deg+k} for
         the \sverif\ case. The green rectangles indicate the private agents and the agent's relationship to them. Every agent considers the private agents of her  friends as enemies, whereas the private agents consider the friends of their owners an enemy.}\label{fig:sverif-fdk}
     \end{figure}

\begin{compactitem}[--]
    \item For each element $i\in [3\enn]$ we have an element agent $a_i$ with a private friend $b_i$.
    \item For each set $C_j$, $j\in [\emm]$ we have a set agent $c_j$.
    \item We add $s_1,\dots s_{3\enn}$ special agents and for each $s_l$, $l\in [3\enn -1]$ we add 2 private agents $t_l^1,t_l^2$, while for $s_{3\enn }$, we only have one private friend $t_{3\enn }$.
    \item Finally, we add a greedy agent $g$ with private agents $g^1$ and $g^2$.
\end{compactitem}
Next, we describe the friendship arcs: 
$\{ (s_l,s_{l+1}) , (s_l,t_l^1) , (s_l,t_l^2),$ $ (s_l,a_l) \mid l\in [3\enn-1]\} \cup \{ (s_{3\enn} , t_{3\enn}) ,(s_{3\enn},a_{3\enn})\} \bigcup \{ (a_i,b_i)\mid i\in [3\enn] \} \bigcup \{ (a_i,c_j)\mid i,j:i\in C_j\} \cup \{ (g,s_1),(g,g^1),(g^1,g^2)\} $.

The enemy arcs are:
$\{ (s_l,t_{l+1}^1)$, $(s_l,t_{l+1}^2)$, $(t_l^1,s_{l+1})$, $(t_l^2,s_{l+1})$,\\$(t_l^1,a_l), $ $ (t_l^2,a_l) , (s_l,b_l) \mid l\in [3\enn -1] \} \cup \{ (s_{3\enn},b_{3\enn}),(t_{3\enn},a_{3\enn})\}$ $ \bigcup \{ (b_i,c_i)\mid i\in [3\enn]\}$  $\bigcup \{ (c_j,c_z)\mid j,z: j<z, C_j\cap C_z\ne \emptyset \} \cup \{ (g,g^2),$ $(g^1,s_1),(g,t_1^1),(g,t_1^2)\}$.

The construction is illustrated in~\cref{fig:sverif-fdk}.
The initial partition is 
\begin{align*}
\Pi \coloneqq \bigcup_{l\in [3\enn-1]}\{\{ s_l,t_l^1,t_l^2\}\} \cup \{\ s_{3\enn},t_{3\enn}\}\} \bigcup_{i\in [3\enn]}\{\{ a_i,b_i\}\}\\
 \bigcup_{j\in [\emm]} \{\{ c_j\}\}  \cup \{\{ g,g^1,g^2\}\}
\end{align*}
 Clearly, $\Pi$ is individually rational.

As before, one can verify that the degree of each agent is at most 12 and each initial coalition has size at most 3. To show that the graph is a DAG, we give a topological order of the vertices: $g,g^1,g^2,s_1,t_1^1,t_1^2,s_2,t_2^1,t_2^2,\dots,s_{3\enn},t_{3\enn},a_1\dots,a_{3\enn},b_1, $ $\dots, b_{3\enn},$ $c_1,\dots,c_{3\emm}$.

\begin{observation}
\label{obs:sverif-fdk}
It holds in $\Pi$ that:
\begin{compactenum}[(i)]
  \item \label{obs:sverif-fdk,i} for each $i\in [3\enn ]$ $a_i$ has 1 friend and 0 enemies, $b_i$ has 0 friends and 0 enemies,
  \item \label{obs:sverif-fdk,ii} for each $j\in [\emm]$ $c_j$ has 0 friends and 0 enemies,
  \item \label{obs:sverif-fdk,iii} for $l\in [3\enn -1]$ $s_l$ has 2 friends and 0 enemies and $t_l^1,t_l^2$ have 0 friends and 0 enemies, while $s_{3\enn }$ has 1 friend and 0 enemies, $t_{3\enn }$ has no friends or enemies and
  \item \label{obs:sverif-fdk,v} $g$ has 1 friend and 1 enemy, $g^1$ has 1 friend and 0 enemies, $g^2$ has no friends or enemies.
\end{compactenum}

\end{observation}

\begin{clm}
\label{claim:sverif-fdk-1}
If there is an exact cover, then there is a weakly blocking coalition to $\Pi$.
\end{clm}
\begin{proof}[Proof of \cref{claim:sverif-fdk-1}]
\renewcommand{\qedsymbol}{$\diamond$}
Let $C_{j_1},\dots,C_{j_{\enn}}$ be the exact cover. Take the following coalition $P=\{ s_i,a_i,c_j,g\mid i\in [3\enn ],j\in \{ j_1,\dots,j_{\enn}\} \}$.

We claim that $P$ is blocking. Each $a_i$ agent for $i\in [3\enn ]$ has 1 friend and 0 enemies in $P$. Therefore they weakly improve by \cref{obs:sverif-fdk}(\ref{obs:sverif-fdk,i}).

For each $l\in [3\enn -1]$, $s_l$ has 2 friends and 0 enemies, while $s_{3\enn }$ has 1 friend and 0 enemies in $P$.  Therefore they weakly improve by \cref{obs:sverif-fdk}(\ref{obs:sverif-fdk,iii}). 

For each $i\in [\enn]$, $c_{j_i}$ has 0 friends and 0 enemies in $P$ (as they formed an exact cover, so none of them intersected). By \cref{obs:sverif-fdk}(\ref{obs:sverif-fdk,ii}), they also weakly improve.

Finally, $g$ has 1 friend and 0 enemies in $P$, so by \cref{obs:sverif-fdk}(\ref{obs:sverif-fdk,v}) she strictly improves.

This concludes that $P$ is a weakly blocking coalition. 
\end{proof}

\begin{clm}
\label{claim:sverif-fdk-3}
In any weakly blocking coalition $P$, only the greedy agent $g$ can strictly improve. 
\end{clm}
\begin{proof}[Proof of \cref{claim:sverif-fdk-3}]
\renewcommand{\qedsymbol}{$\diamond$}
Let $P$ be a weakly blocking coalition. 
Clearly, private agents cannot strictly improve as they receive all their friends and no enemies in  $\Pi$.

Suppose $s_l$ strictly improves for some $l\in [3\enn -1]$. Then, by \cref{obs:sverif-fdk}(\ref{obs:sverif-fdk,iii}), she must obtain 3 friends. Therefore, either $a_l$ or $s_{l+1}$ is in $P$. But both of them are an enemy to $t_l^1$ and $t_l^2$. As they have no friends, this implies that $t_l^1,t_l^2\notin P$. Therefore, $s_l$ can have at most 2 friends in $P$, contradiction. For $s_{3\enn }$ to strictly improve, she must get both of her friends by \cref{obs:sverif-fdk}(\ref{obs:sverif-fdk,iii}), so $t_{3\enn },a_{3\enn }\in P$. But $t_{3\enn }$ has no friends and considers $a_{3\enn }$ an enemy, contradiction.

The agents $c_j$, $j\in [\emm]$ cannot strictly improve, as they have no friends.

Suppose $a_i$ strictly improves for some $i\in [3]$. Then, by \cref{obs:sverif-fdk}(\ref{obs:sverif-fdk,i}), she must receive at least 2 friends, so there is some $c_j\in P$, such that $i\in C_j$ holds. As $c_j$ is an enemy to private agent $b_i$, and she has no friends, it follows that $b_i\notin P$. Hence, there is at least 2 set agents $c_j,c_z\in P$ with $j<z$, such that $a_i\in C_j\cap C_z$. Therefore $c_j$ has an enemy in $P$. As she has no friends, this is a contradiction.
\end{proof}

\begin{clm}
If there is a weakly blocking coalition $P$, then there is an exact cover.
\label{claim:sverif-fdk-4}
\end{clm}
\begin{proof}[Proof of \cref{claim:sverif-fdk-4}]
\renewcommand{\qedsymbol}{$\diamond$}
By \cref{claim:sverif-fdk-3}, we have that the greedy agent $g$ must be in $P$, as no one else can strictly improve. Also, $g$ has to strictly improve. By \cref{obs:sverif-fdk}(\ref{obs:sverif-fdk,v}), either $g$ has 2 friends or 1 friend and no enemies in $P$. If $s_1\notin P$, then the only friend $g$ can obtain is $g^1$. But then, for $g^1$ to weakly improve, she must get her  only friend $g^2$, so $g$ also has an enemy in $P$ and cannot strictly improve, a contradiction. Hence, $s_1\in P$. As $s_1$ is an enemy to $g^1$ and she has only one friend, by \cref{obs:sverif-fdk}(\ref{obs:sverif-fdk,v}) she cannot weakly improve, so $g^1\notin P$. Therefore, by \cref{obs:sverif-fdk}(\ref{obs:sverif-fdk,v}) and \cref{claim:sverif-fdk-3}, $g$ cannot have enemies in $P$. As $t_1^,t_1^2$ are enemies to $g$, they are not in $P$.

By \cref{obs:sverif-fdk}(\ref{obs:sverif-fdk,iii}), $s_1$ must have both of her other friends in $P$, so $a_1\in P$ and $s_2\in P$. 
Also, for $s_1$ to weakly improve, she must not obtain any enemies, hence $t_2^1,t_2^2,b_1\notin P$.
Again, by \cref{obs:sverif-fdk}(\ref{obs:sverif-fdk,iii}), $s_2$ must get at least two friends without $t_2^1,t_2^2$, so she also cannot get any enemies. Therefore, $s_3,a_2\in P$ and $t_3^1,t_3^2,b_2\notin P$. Iterating this, we obtain that $\{ s_1,s_2,\dots, s_{3\enn },a_1,a_2,\dots,a_{3\enn }\} \subset P$ and $b_1,\dots,b_{3\enn }\notin P$. 

By \cref{claim:sverif-fdk-3} we know that each $a_i$ can only weakly improve in $P$. Hence, by \cref{obs:sverif-fdk}(\ref{obs:sverif-fdk,i}), $a_i$ must have exactly 1 friend and no enemies in $P$ for each $i\in [3\enn ]$. As $b_i\notin P$, we have that there is exactly one set agent $c_j\in P$, such that $i\in C_j$. Therefore, the sets corresponding to the set agents in $P$ must form an exact cover. 
\end{proof}

\noindent The correctness follows from Claims~\ref{claim:sverif-fdk-1} \& \ref{claim:sverif-fdk-4}.
\end{proof}
}

Unlike the Nash and individual stability, symmetric preferences do not help in reducing the complexity.
\newcommand{\fencoresymm}{
  \FEN-\verif (resp.\ \FEN-\sverif) is \conp-com\-plete even when the preferences are symmetric, $\maxdeg\!=\!7$ (resp.\ $\maxdeg\!=\!26$), and $\maxcoal\!=\!3$ (resp.\ $\maxcoal\!=\!4$). 
}
\begin{theorem}[\appendixsymb]\label{thm:core_verify_neutals}
  \fencoresymm %
\end{theorem}

\appendixproofwithstatement{thm:core_verify_neutals}{\fencoresymm}{
\begin{proof}
  Let $I=(\mathcal{X}, \mathcal{C})$ denote an instance of \xct\ with $\mathcal{X}=[3\enn]$ and $\mathcal{C}=\{C_1,\ldots,C_{\emm}\}$. We first provide a construction for \FEN-\verif and prove the correctness. Then, we modify our gadget by adding some extra agents to construct an instance of \FEN-\sverif and prove its correctness.
  
     \begin{figure}[t!]
     \centering
       \begin{tikzpicture}[scale=1,every node/.style={scale=0.9}, shorten <= 2pt, shorten >= 2pt]
         
      \foreach \x / \y / \n / \nn / \typ / \p / \dx / \dy in {
      1/4/d11/d_1^1/pn/{above}/-1/-4,
      1.5/4/d12/d_1^2/pn/{above}/-1/-4,
      3.5/4/di1/d_i^1/pn/{above left}/-1/-4,
      4/4/di2/d_i^2/pn/{above}/-1/-4,
      5/4/di11/d_{i+1}^1/pn/{above}/-1/-4,
      5.5/4/di12/d_{i+1}^2/pn/{above right}/-1/-4,
      8/4/d3n1/d_{3\enn}^1/pn/{above}/-1/-4,  
      8.5/4/d3n2/d_{3\enn}^2/pn/{above}/-1/-4,  
      1/3/a1/a_1/pn/{above left}/-1/-4,
      4/3/ai/a_i/pn/{left}/-1/-2,
      5/3/ai1/a_{i+1}/pn/{right}/-1/-2,
      8/3/a3n/a_{3\enn}/pn/{left}/-1/-4,
      1.2/1/cj/c_{j} /pn/{below left}/-1/0,
      2.5/1/cj0/c_{j}^0 /pn/{below left}/-1/-3,
      3.7/1/cj1/c_{j}^1 /pn/{below left}/-1/-2,      
      7/1/ck/c_{z}/pn/{below left}/-1/-4,
      7.5/1/ck0/c_{z}^0/pn/{below left}/-1/-4,
      8/1/ck1/c_{z}^1/pn/{below left}/0/-4,
      1/0/ej1/e_{j}^1 /pn/{below}/-1/-4,
      1.5/0/ej2/e_{j}^2 /pn/{below}/-1/-3.5,
      2.5/0/ej01/e_{j}^{0,1} /pn/{below}/-1/-3,
      3/0/ej02/e_{j}^{0,2} /pn/{below}/-1/-2.5,
      4/0/ej11/e_{j}^{1,1} /pn/{below}/-1/-2,
      4.5/0/ej12/e_{j}^{1,2} /pn/{below}/-1/-1.5,
      -1/1/pa1c2/\;/pnn/{left}/0/0,
      -.5/1/pa1c3/\;/pnn/{left}/0/0,
      4/1.5/paic1/\;/pnn/{left}/0/0,
      5/1.5/pai1c1/\;/pnn/{left}/0/0,
      6/1/pai1c2/\;/pnn/{left}/0/0,
      6.5/1/pai1c3/\;/pnn/{left}/0/0,
      8.5/1/pa3nc1/\;/pnn/{left}/0/0,
      9/1/pa3nc2/\;/pnn/{left}/0/0,
      9.5/1/pa3nc3/\;/pnn/{left}/0/0,
      6.5/3/pcka1/\;/pnn/{left}/0/0,
      7/3/pcka2/\;/pnn/{left}/0/0} {
        \node[\typ] at (\x*0.8, \y) (\n) {};
        \node[\p = \dx pt and \dy pt of \n] {$\nn$};
      }      

      \begin{pgfonlayer}{bg}
           \foreach \s / \t / \aa / \type in {
           d11/a1/0/fc,
           d12/a1/0/fc,
           di1/ai/0/fc,
           di2/ai/0/fc,
           di11/ai1/0/fc,
           di12/ai1/0/fc,
           d3n1/a3n/0/fc,
           d3n2/a3n/0/fc,
           a1/cj0/0/fc,
           ai/cj1/0/fc,
           ai/ck0/0/fc,
           ai1/cj1/0/fc,
           a1/pa1c2/0/hiddenfc,
           a1/pa1c3/0/hiddenfc,
           ai/paic1/0/hiddenfc,
           ai1/pai1c1/0/hiddenfc,
           ai1/pai1c2/0/hiddenfc,
           a3n/pa3nc1/0/hiddenfc,
           a3n/pa3nc2/0/hiddenfc,
           a3n/pa3nc3/0/hiddenfc,
           ck1/pcka1/0/hiddenfc,
           ck1/pcka2/0/hiddenfc,
           ej1/cj/0/fc,
           ej1/ej2/0/fc,
           ej01/cj0/0/fc,
           ej02/cj0/0/fc,
           ej11/cj1/0/fc,
           ej12/cj1/0/fc,
           cj/cj0/0/fc,
           cj0/cj1/0/fc,
           ck/ck0/0/fc,
           ck0/ck1/0/fc,
           ej2/cj/0/ec} {
             \draw[-, \type] (\s) edge[bend right = \aa] (\t);
           }
           \foreach \s / \t / \aa / \type in {
           cj/ck/0/ec}{
             \draw[-, \type] (\s) edge[bend left=20] (\t);
           }

         \end{pgfonlayer}
       \end{tikzpicture}\caption{Illustration for the proof of \cref{thm:core_verify_neutals} where $C_j = \{a_1, a_i, a_{i+1}\}$ and $C_j \cap C_z = a_i$.}\label{fig:core_verify_neutals}
     \end{figure}

\noindent \textbf{Construction for \FEN-\verif.}  For each element $i \in \mathcal{X}$ we create an element agent $a_i$ and two private agents $d^1_i, d^2_i$. For each set $C_j \in \mathcal{C}$ we create three set agents $c_j, c^0_j, c^1_j$ and private agents $e^1_j, e^2_j, e^{0,1}_j, e^{0,2}_j, e^{1,1}_j, e^{1,2}_j$. This completes the construction of the agents. In total, we have $V \coloneqq \{a_i, d^1_i, d^2_i \mid i \in \mathcal{X}\} \cup \{c_j, c^0_j, c^1_j, e^1_j, e^2_j, e^{0,1}_j, e^{0,2}_j,$ $e^{1,1}_j, e^{1,2}_j \mid C_j \in \mathcal{C}\}$.
  
  Next, we describe the friendship edges. For every element $i \in \mathcal{X}$ we construct edges $\{a_i, a_{i + 1}\}$ where we take $i + 1$ modulo $3\enn$. Additionally, we create the internal friendship edges $\{a_i, d^l_i\}$ for every $l \in [2]$. For every set $C_j \in \mathcal{C}$ we construct edges $\{c_j, c_j^0\}, \{c^0_j, c^1_j\}$. We also create internal friendship edges $\{c_j, e^1_j\}, \{e^1_j, e^2_j\}$ and $\{c^{z}_j, e^{z,l}_j\}$ for every  $l \in [2], z \in \{0,1\}$. Finally, to connect the set and element agents, we construct the edges $\{c^0_j, a_{j^1}\}, \{c^1_j, a_{j^2}\}, \{c^1_j, a_{j^3}\}$, where $C_j = \{j^1, j^2, j^3\}$ such that $j^1 < j^2 < j^3$.
  
  We construct the following enemy edges: For every two sets $C_j, C_z \in \mathcal{C}$ such that $C_j \cap C_z \neq \emptyset$, $\{c_j, c_z\}$ are enemies. Additionally, we construct a private enemy edge $\{c_j, e^2_j\}$ for every $C_j \in \mathcal{C}$.
  
  Our initial coalition structure is defines as follows
\begin{multline*}
\Pi \coloneqq \{ \{a_i, d^1_i, d^2_i \} \mid i \in \mathcal{X} \} \cup \{\{ c_j, e^1_j, e^2_j\}, \{c^z_j, e^{z,1}_j, e^{z,2}_j\} \mid \\
 C_j \in \mathcal{C}, z \in \{0, 1\}\}.
\end{multline*}
 
 It is clear that each initial coalition is of size 3. 
 
 Next we show that both the number of friends and the number of enemies any agent has in $V$ are bounded by 7. Each agent $a_i, i \in \mathcal{X}$ has 2 friends from $a_{i - 1}, a_{i +1}$, 2 friends from her  private agents and at most 3 friends from the sets it is included in, in total at most 7 friends. It has no enemies. Each agent $c_j, j \in C_j$ has 1 friend from $c^0_j$ and 1 from her  private agents, in total 2 friends. It has 1 enemy from her  private agents and at most 6 enemies from $\{c_{j'} \mid C_{j'} \in \mathcal{C}\}$, because each $i \in C_j$ is in at most 2 other sets. In total it has thus at most 7 enemies. Each agent $c^0_j, C_j = \{j^1, j^2, j^3\}\in \mathcal{C}$ has 3 public friends $c_j, c^1_j, a_{j^1}$ and 2 private friends, in total 5 friends. Similarly $c^1_j$ has 3 public friends $c^0_j, a_{j^2}, a_{j^3}$ and 2 private friends. Neither of these agents has any enemies. It is easy to see that all the private agents have at most 2 friends and at most 1 enemy. We can see that any agent has at most 7 non-neutral agents.

 \begin{observation}\label{obs:sym_fen_init_coal_degree1}
 It holds that in $\Pi$
    \begin{compactenum}[(i)]
  \item \label{obs:sym_fen_init_coal_degree1,1i} For every $i \in \mathcal{X}$, $a_i$ has 2 friends and 0 enemies,
  \item \label{obs:sym_fen_init_coal_degree1,1ii}for every $i \in \mathcal{X}$, $d^j_i, j \in [2]$ has 1 friend and 0 enemies,
  \item \label{obs:sym_fen_init_coal_degree1,1iii} for every $ C_j \in \mathcal{C}, l \in [2]$ such that $i \in C_j$, $c_j$ has 1 friend and 1 enemy, and for every $ l \in [2], c^l_j$ has 2 friends and 0 enemies,
  \item \label{obs:sym_fen_init_coal_degree1,1iv}  for every $C_j \in \mathcal{C}$, $e^1_j$ has 2 friends and 0 enemies, $e^2_j$ has 1 friend and 0 enemies, and for every $z \in \{0, 1\}, l \in [2], e^{z,l}_j$ has 1 friend and 0 enemies.
  \end{compactenum}
 \end{observation}
 
 It remains to show that $I$ admits an exact cover if an only if $\Pi$ is not core stable.
 
 \begin{clm}\label{cla:sym_fen_init_coal_degree1}
If $I$ admits an exact cover, then $\Pi$ is not core stable.
 \end{clm}
 
 \begin{proof}[Proof of \cref{cla:sym_fen_init_coal_degree1}]
 Let $\mathcal{K}$ be an exact cover of $I$. We claim that the following coalition \[P \coloneqq \{a_i \mid i \in \mathcal{X}\} \cup \{ c_j, c^0_j, c^1_j \mid C_j \in \mathcal{K}, i \in C_j\} \] is blocking $\Pi$.
 \begin{compactenum}
  \item For every $i \in \mathcal{X}$, $a_i$ has at least 3 friends, $a_{i - 1}, a_{i + 1}$ and $c^l_j$ for some $l \in \{0,1\}, C_j \in \mathcal{C}$ such that $C_j$ covers $i$ in~$\mathcal{K}$.
  \item For every $C_j \in \mathcal{K}$, $c_j$ has $1$ friend $c^0_j$. The enemies of $c_z$ are the $c_l$ corresponding to sets such that $C_z \cap C_l \neq \emptyset$. Because $\mathcal{K}$ is an exact cover, no such $c_l$ can be in $P$ if $c_z \in P$ and $c_z$ has $0$ enemies.
  \item For every $C_j = \{j^1, j^2, j^3\} \in \mathcal{K}$, $c^0_j$ has 3 friends $c_j, c^1_j, a_{j^1}$ and $c^1_j$ has 3 friends $c^0_j, a_{j^2}, a_{j^3}$.
  \end{compactenum}
 From \cref{obs:sym_fen_init_coal_degree1} we can see that every agent in $P$ strictly improves. Therefore $P$ is a blocking coalition.
 \end{proof}

  \begin{clm}\label{cla:sym_fen_init_coal_degree2}
Every blocking coalition $P$ of $\Pi$ satisfies
    \begin{compactenum}[(i)]
  \item \label{cla:sym_fen_init_coal_degree2,1i} For every $i \in \mathcal{X}$, if $a_i \in P$, then $a_{i-1}, a_{i+1} \in P$, where subscript~$i + 1$ (resp.\ $i -1$) denotes $1$ (resp.\ $3\enn$) if $i=3\enn$ (resp.\ $i=1$).
  \item \label{cla:sym_fen_init_coal_degree2,1ii} For every $i \in \mathcal{X}$, if $a_i \in P$, then there is $C_j \in \mathcal{C}$ such that $i \in C_j$ and $c_j \in P$.
  \item \label{cla:sym_fen_init_coal_degree2,1iii} For every $\{C_j, C_z\} \subset \mathcal{C}$ such that $C_j \cap C_z \neq \emptyset$, if $c_j \in P$, then $c_z \notin P$.
  \item \label{cla:sym_fen_init_coal_degree2,1iv} For every $C_j \in \mathcal{C}$, if $c_j \in P$, then $c^0_j \in P$.
  \item \label{cla:sym_fen_init_coal_degree2,1v} For every $C_j = \{j^1, j^2, j^3\} \in \mathcal{C}$, if $c^0_j \in P$, then $c_j, c^1_j, a_{j^1} \in P$.
  \item \label{cla:sym_fen_init_coal_degree2,1vi} For every $C_j = \{j^1, j^2, j^3\} \in \mathcal{C}$, if $c^1_j \in P$, then $c^0_j, a_{j^2}, a_{j^3} \in P$.
  \end{compactenum}
 \end{clm}
 
 \begin{proof}[Proof of \cref{cla:sym_fen_init_coal_degree2}]
 \renewcommand{\qedsymbol}{$\diamond$}
 Assume that $\Pi$ admits a strictly blocking coalition $P$.
 
 Note that none of the private agents are in any blocking coalition. For every $i \in X, C_j \in \mathcal{C}, z \in \{0, 1\}, l \in [2], d^l_i, e^1_j, e^{z,l}_j$ obtain all of their friends in $\Pi$ and they have no enemies, so they cannot improve. For every $C_j \in \mathcal{C}, e^2_j$ has all of her  friends and 1 enemy, so it could improve by joining a coalition that does not contain $c_j$ with her  only friend $e^1_j$. However, as shown earlier, $e^1_j$ cannot join a blocking coalition.
 
 To show (\ref{cla:sym_fen_init_coal_degree2,1iii}), note that since $c_j$ has only 1 non-private friend, to improve it must have no enemies in her  blocking coalition. The non-private enemies of $c_j$ are precisely the agents $c_z$ such that $C_j \cap C_z \neq \emptyset$.
 
 To show (\ref{cla:sym_fen_init_coal_degree2,1iv}), note that since $c_j$ must obtain at least 1 friend to join any blocking coalition and the private friend will not join, the only friend $c_j$ can obtain is $c^0_j$.
 
 To show (\ref{cla:sym_fen_init_coal_degree2,1v}-\ref{cla:sym_fen_init_coal_degree2,1vi}) note that $c^0_j, c^1_j$ must obtain 3 friends in any blocking coalition. Since the private agents do not join a blocking coalition, they must obtain all their other friends.

Statement (\ref{cla:sym_fen_init_coal_degree2,1ii}) follows from (\ref{cla:sym_fen_init_coal_degree2,1iv}-\ref{cla:sym_fen_init_coal_degree2,1vi}): Since $a_i$ must have at least 3 friends in any blocking coalition, it must have at least one friend $c^z_j$ in $P$, where $z \in \{0, 1\}, i \in C_j$. By Statements (\ref{cla:sym_fen_init_coal_degree2,1v}-\ref{cla:sym_fen_init_coal_degree2,1vi}), $c_j$ must also be in $P$.

We show (\ref{cla:sym_fen_init_coal_degree2,1i}) through a contradiction. Assume that $a_i \in P$ but $a_{i-1} \notin P$ or $a_{i+1} \notin P$. As $a_i$ must obtain 3 friends in $P$, there must be $C_j, C_l \in \mathcal{C}$ such that $c^z_j, c^{z'}_l\ \in P$ and $i \in C_j \cap C_l, z, z' \in \{0,1\}$. By (\ref{cla:sym_fen_init_coal_degree2,1v}-\ref{cla:sym_fen_init_coal_degree2,1vi}), we obtain that $c_j, c_l \in P$,
a contradiction to (\ref{cla:sym_fen_init_coal_degree2,1iii}) since $C_l \cap C_j \neq \emptyset$.
 \end{proof}
 
 \begin{clm}\label{cla:sym_fen_init_coal_degree3}
   If $\Pi$ is not core stable, then $I$ admits an exact cover.
 \end{clm}
 
 \begin{proof}[Proof of \cref{cla:sym_fen_init_coal_degree3}]
 \renewcommand{\qedsymbol}{$\diamond$}
 Assume that $\Pi$ admits a strictly blocking coalition $P$. We claim that $\mathcal{K} \coloneqq \{C_j  \mid C_j \in \mathcal{C}, c_j \in P\}$ is an exact cover of $I$.
 
 Suppose, towards a contradiction, that $\mathcal{K}$ does not cover $\mathcal{X}$. Then there is an element $i \in \mathcal{X}$ such that no set $C_j \in \mathcal{K}$ satisfies $i \in C_j$ and thus none of the corresponding agents $c_j$ satisfies $c_j \in P$. By contra-positive of \cref{cla:sym_fen_init_coal_degree2}(\ref{cla:sym_fen_init_coal_degree2,1i}) we obtain that $a_i \notin P$.  By contra-positive of \cref{cla:sym_fen_init_coal_degree2}(\ref{cla:sym_fen_init_coal_degree2,1ii}) we get that $a_{i + 1} \notin P$. By iteratively applying this argument, we obtain that $a_{i'} \notin P$ for every $i' \in \mathcal{X}$. By contra-positives of \cref{cla:sym_fen_init_coal_degree2}(\ref{cla:sym_fen_init_coal_degree2,1iv}-\ref{cla:sym_fen_init_coal_degree2,1vi}) we get that $c_j, c^0_j, c^1_j \notin P$ for every $C_j \in \mathcal{C}$. Therefore $P = \emptyset$, a contradiction.
 
 Suppose, towards a contradiction, that $\mathcal{K}$ covers $\mathcal{X}$ but is not an exact cover. Then there is $i \in \mathcal{X}$ such that $c_j, c_z \in P$ where $i \in C_j \cap C_z, C_z, C_j \in \mathcal{C}$. Therefore $C_j \cap C_z \neq \emptyset$ but $c_j, c_z \in P$, a contradiction to \cref{cla:sym_fen_init_coal_degree2}(\ref{cla:sym_fen_init_coal_degree2,1iii}).
 Summarizing, $\mathcal{K}$ is an exact cover, as desired.
 \end{proof}
\noindent The correctness follows from Claims~\ref{cla:sym_fen_init_coal_degree1} \& \ref{cla:sym_fen_init_coal_degree3}.
 
\smallskip
\noindent\textbf{Construction for \FEN-\sverif.}  Now we modify the constructed instance to show the hardness for \FEN-\sverif. In addition to the agents created for \FEN-\verif, for each $i \in \mathcal{X}$, we create an agent $d_i^3$ and for each set $C_j \in \mathcal{C}$, we create private agents $e_j^{0,3}, e_j^{1,3}$. In total, we have $V \coloneqq \{a_i, d^1_i, d^2_i,d^3_i \mid i \in \mathcal{X}\} \cup \{c_j, c^0_j, c^1_j, e^1_j, e^2_j, e^{0,l}_j, e^{1,l}_j \mid C_j \in \mathcal{C},l\in [3]\}$.
 
 We add the following additional friendship edges: $\{a_i,d_i^3\}, \{c_j^z,$ $e_j^{z,3}\}$ for each $i \in \mathcal{X}$, and the following additional enemy edges: $\{c_j^0, c_z\},$ $\{c_j, c_z^0\}$ for each two $C_j, C_z \in \mathcal{C}$ such that $C_j \cap C_z \neq \emptyset$. Then, we construct a private enemy edge $(c_j, e^2_j)$ for every $C_j \in \mathcal{C}$. For each $d_i^l$, $l\in [3]$ and $j$ such that $a_i\in C_j$ we add enemy edges $( d_i^l,c_j^0), (d_i^l,c_j^1)$ and $(e_j^{z,l},a_i)$ for $z\in \{ 0,1\}, l\in [3]$. Finally we construct the additional enemy edges $(e_j^{0,l},c_j),(e_j^{1,l}c_j^0),(e_j^l,c_j^0)$ for $j\in [\emm ] , l\in [3]$.
  
  Our initial coalition structure is defines as follows
  
$
\Pi \coloneqq \{ \{a_i, d^1_i, d^2_i,d_i^3 \} \mid i \in \mathcal{X} \} \cup \{\{ c_j, e^1_j, e^2_j\}, \{c^z_j, e^{z,1}_j, e^{z,2}_j, e^{z,3}_j\} \mid C_j \in \mathcal{C}, z \in \{0, 1\}\}.
$
 
 It is clear that each initial coalition is of size at most 4. 
 
 Next we show that both the number of friends and the number of enemies any agent has in $V$ are bounded by 26. Each agent $a_i, i \in \mathcal{X}$ has 2 friends from $a_{i - 1}, a_{i +1}$, 3 friends from her  private agents and at most 3 friends from the sets it is included in, in total at most 8 friends. It has at most $3\cdot 6=18$ enemies from the $\{ e_j^{z,l}\}$ agents corresponding to the at most 3 sets containing $a_i$. Each agent $c_j, j \in C_j$ has 1 friend from $c^0_j$ and 1 from her  private agents, in total 2 friends. It has 1 enemy from her  private agents and at most 12 enemies from $\{c_{j'},c_{j'}^0 \mid C_{j'} \in \mathcal{C}\}$, because each $i \in C_j$ is in at most 2 other sets. In total it has thus at most 13 enemies. Each agent $c^0_j, C_j = \{j^1, j^2, j^3\}\in \mathcal{C}$ has 3 public friends $c_j, c^1_j, a_{j^1}$ and 3 private friends, in total 6 friends. It has also 6+9=15 enemies $\{ c_j^l,c_j^{1,l} \mid l\in [3]\}$ and $\{ d_i^l \mid l\in [3]\}$ for each $i\in C_j$. Similarly $c^1_j$ has 3 public friends $c^0_j, a_{j^2}, a_{j^3}$ and 3 private friends. Also they have 3 enemies $\{ c_j^{0,l} \mid l\in [3]\}$. It is easy to see that all the private agents have at most 2 friends and at most 3 enemies. We can see that any agent has at most 26 non-neutral agents.

 \begin{observation}\label{obs:sym_fen_strict_obs}
  In~$\Pi$ the following holds.
    \begin{compactenum}[(i)]
  \item \label{obs:sym_fen_strict_obs,1i} For every $i \in \mathcal{X}$, $a_i$ has 3 friends and 0 enemies,
  \item \label{obs:sym_fen_strict_obs,1ii}for every $i \in \mathcal{X}$, $d^l_i, l \in [2]$ has 1 friend and 0 enemies,
  \item \label{obs:sym_fen_strict_obs,1iii} for every $ C_j \in \mathcal{C}$ such that $i \in C_j$, $c_j$ has 1 friend and 1 enemy, and for every $ z \in \{ 0,1\} $ , $c^z_j$ has 3 friends and 0 enemies,
  \item \label{obs:sym_fen_strict_obs,1iv}  for every $C_j \in \mathcal{C}$, $e^1_j$ has 2 friends and 0 enemies, $e^2_j$ has 1 friend and 1 enemy, and for every $z \in \{0, 1\}, l \in [3], e^{z,l}_j$ has 1 friend and 0 enemies.
  \end{compactenum}
 \end{observation}
 
 It remains to show that $I$ admits an exact cover if an only if $\Pi$ is not strict core stable.
 
 \begin{clm}\label{cla:sym_fen_strict-1}
If $I$ admits an exact cover, then $\Pi$ is not strict core stable.
 \end{clm}
 
 \begin{proof}[Proof of \cref{cla:sym_fen_strict-1}]
 \renewcommand{\qedsymbol}{$\diamond$}
 Let $\mathcal{K}$ be an exact cover of $I$. We claim that the following coalition \[P \coloneqq \{a_i \mid i \in \mathcal{X}\} \cup \{ c_j, c^0_j, c^1_j \mid C_j \in \mathcal{K}, i \in C_j\} \] is weakly blocking $\Pi$.
 \begin{compactenum}[(i)]
  \item For every $i \in \mathcal{X}$, $a_i$ has at least 3 friends, $a_{i - 1}, a_{i + 1}$ and $c^l_j$ for some $l \in \{0,1\}, C_j \in \mathcal{C}$ such that $C_j$ covers $i$ in $\mathcal{K}$ and no enemies, as there is no agent of the form $e_j^{z,l}$ inside $P$.
  \item For every $C_j \in \mathcal{K}$, $c_j$ has 1 friend $c^0_j$. The enemies of a $c_j$ agent are the $c_z,c_z^0$ agents corresponding to sets such that $C_j \cap C_z \neq \emptyset$ and the $e_j^2,e_j^{0,l}$ agents for $l\in [3]$. Because $\mathcal{K}$ is an exact cover, no such $c_z$ or $c_z^0$ can be in $P$ if $c_j \in P$ so $c_j$ has 0 enemies.
  \item For every $C_j = \{j^1, j^2, j^3\} \in \mathcal{K}$, $c^0_j$ has 3 friends $c_j, c^1_j, a_{j^1}$ and $c^1_j$ has 3 friends $c^0_j, a_{j^2}, a_{j^3}$. As there are no $e_j^{z,l}$, $e_j^l$ or $d_i^l$ agents inside $P$ for any $j,z,l$ and the sets inside $\mathcal{K}$ form an exact cover, they have no enemies.
  \end{compactenum}
 From \cref{obs:sym_fen_strict_obs} we can see that every agent in $P$ weakly improves and the $c_j$ set agents strictly improve, as they have one less enemy. Therefore $P$ is a blocking coalition.
 \end{proof}

  \begin{clm}\label{cla:sym_fen_strict-2}
Every inclusionwise minimal weakly blocking coalition $P$ of $\Pi$ satisfies
    \begin{compactenum}[(i)]
  \item \label{cla:sym_fen_strict,1iv} For every $C_j \in \mathcal{C}$, if $c_j \in P$, then $c^0_j \in P$.
  \item \label{cla:sym_fen_strict,1v} For every $C_j = \{j^1, j^2, j^3\} \in \mathcal{C}$, if $c^0_j \in P$, then $c_j, c^1_j, a_{j^1} \in P$.
  \item \label{cla:sym_fen_strict,1vi} For every $C_j = \{j^1, j^2, j^3\} \in \mathcal{C}$, if $c^1_j \in P$, then $c^0_j, a_{j^2}, a_{j^3} \in P$.
  \item \label{cla:sym_fen_strict,1iii} For every $\{C_j, C_z\} \subset \mathcal{C}$ such that $C_j \cap C_z \neq \emptyset$, if $c_j \in P$, then $c_z \notin P$.
  \item \label{cla:sym_fen_strict,1ii} For every $i \in \mathcal{X}$, if $a_i \in P$, then there is $C_j \in \mathcal{C}$ such that $i \in C_j$ and $c_j \in P$.
  \item \label{cla:sym_fen_strict,1i} For every $i \in \mathcal{X}$, if $a_i \in P$, then $a_{i-1}, a_{i+1} \in P$, where subscript~$i + 1$ (resp.\ $i -1$) denotes $1$ (resp.\ $3\enn$) if $i=3\enn$ (resp.\ $i=1$).
  \end{compactenum}
 \end{clm}
 
 \begin{proof}[Proof of \cref{cla:sym_fen_strict-2}]
 \renewcommand{\qedsymbol}{$\diamond$}
 Assume $\Pi$ admits a weakly blocking coalition and let $P$ be an inclusionwise minimal one. 
 
 First, we show that none of the private agents are in $P$.
 For every $(i,C_j,z,l) \in \mathcal{X}\times \mathcal{C}\times \{0, 1\}\times
 [3]$, agents~$d^l_i$, $e^1_j$, $e^{z,l}_j$ obtain all of their friends in $\Pi$ and no enemies, so they cannot strictly improve. 
 Since each of them are enemies with the non-private friends of their corresponding $a_i,c_j,c_j^0$ or $c_j^1$ agent, it follows that if that agent would get a friend other than her  private friends, then none of her private friends can be included, as they would receive enemies. So, in any weakly blocking coalition, such agents can only be included, if they are together with everyone from their original coalition, and none of them and their corresponding agent receives another friend. This also implies that they are weakly worse in $P$ than in $\Pi$. Therefore, as they are not a friend to anyone else in $P$ because of the symmetric friendships, dropping these agents from $P$ would still give us a weakly blocking coalition. This would contradict the minimality of~$P$.
 
 The case for $e_j^2$ is slightly different.
 For every $C_j \in \mathcal{C}, e^2_j$ has all of her friends and 1 enemy, so in order to participate in $P$, her only friend $e_j^1$ also must be in $P$, and therefore so does $c_j$. But again, $e_j^2$ is enemies with all non private friends of $c_j$, so no other friend of $c_j$ can be in $P$, meaning all of them have the same number of friends and at least the same number of enemies as in $\Pi$. Therefore we could drop $e_j^1,e_j^2$ and $c_j$ from $P$ and still obtain a weakly blocking coalition, contradiction.

 To show (\ref{cla:sym_fen_strict,1iv}), note that since $c_j$ must obtain at least 1 friend to join any blocking coalition and the private friend $e_j^1$ will not join, the only friend $c_j$ can obtain is $c^0_j$.
 
  To show (\ref{cla:sym_fen_strict,1v}-\ref{cla:sym_fen_strict,1vi}) note that $c^0_j, c^1_j$ must obtain 3 friends in any weakly blocking coalition. Since the private agents do not join a blocking coalition, they must obtain all their other friends.
 
 To show (\ref{cla:sym_fen_strict,1iii}), note that since $c_j$ has only 1 non-private friend, to weakly improve it must have at most one enemy in her blocking coalition. The non-private enemies of $c_j$ are precisely the agents $c_z$ and $c_z^0$ such that $C_j \cap C_z \neq \emptyset$. Also, by (\ref{cla:sym_fen_strict,1iv}) and (\ref{cla:sym_fen_strict,1v}), if either of $c_z,c_z^0$ are in $P$, then both of them are, so $c_j$ would get at least two enemies, contradiction.

Statement (\ref{cla:sym_fen_strict,1ii}) follows from (\ref{cla:sym_fen_strict,1iv}-\ref{cla:sym_fen_strict,1vi}): Since $a_i$ must have at least 3 friends in any blocking coalition, it must have at least one friend $c^z_j$ in $P$, where $z \in \{0, 1\}, i \in C_j$. By Statements (\ref{cla:sym_fen_strict,1v}-\ref{cla:sym_fen_strict,1vi}), $c_j$ must also be in $P$.

We show (\ref{cla:sym_fen_strict,1i}) through a contradiction. Assume that $a_i \in P$ but $a_{i-1} \notin P$ or $a_{i+1} \notin P$. As $a_i$ must obtain 3 friends in $P$, there must be $C_j, C_l \in \mathcal{C}$ such that $c^z_j, c^{z'}_l\ \in P$ and $i \in C_j \cap C_l, z, z' \in \{0,1\}$. By (\ref{cla:sym_fen_strict,1v}-\ref{cla:sym_fen_strict,1vi}), we obtain that $c_j, c_l \in P$. Because $C_l \cap C_j \neq \emptyset$, this contradicts~(\ref{cla:sym_fen_strict,1iii}).
 \end{proof}
 
 \begin{clm}\label{cla:sym_fen_strict-3}
   If $\Pi$ is not strictly core stable, then $I$ admits an exact cover.
 \end{clm}
 
 \begin{proof}[Proof of \cref{cla:sym_fen_strict-3}]
 \renewcommand{\qedsymbol}{$\diamond$}
 Assume $\Pi$ admits a blocking coalition $P$. We claim that $\mathcal{K} \coloneqq \{C_j  \mid C_j \in \mathcal{C}, c_j \in P\}$ is an exact cover of $I$.
 
 Suppose, towards a contradiction, that $\mathcal{K}$ does not cover $\mathcal{X}$. Then there is an element $i \in \mathcal{X}$ such that no set $C_j \in \mathcal{K}$ satisfies $i \in C_j$ and thus none of the corresponding agents $c_j$ satisfies $c_j \in P$. By contrapositive of \cref{cla:sym_fen_strict-2}(\ref{cla:sym_fen_strict,1i}) we obtain that $a_i \notin P$.  By contrapositive of \cref{cla:sym_fen_strict-2}(\ref{cla:sym_fen_strict,1ii}) we get that $a_{i + 1} \notin P$. By iteratively applying this argument, we obtain that $a_{i'} \notin P$ for every $i' \in \mathcal{X}$. By contrapositives of \cref{cla:sym_fen_strict-2}(\ref{cla:sym_fen_strict,1iv}-\ref{cla:sym_fen_strict,1vi}) we get that $c_j, c^0_j, c^1_j \notin P$ for every $C_j \in \mathcal{C}$. Therefore $P = \emptyset$, a contradiction.
 
 Suppose, towards a contradiction, that $\mathcal{K}$ covers $\mathcal{X}$ but is not an exact cover. Then there is $i \in \mathcal{X}$ such that $c_j, c_z \in P$ where $i \in C_j \cap C_z, C_z, C_j \in \mathcal{C}$. Therefore $C_j \cap C_z \neq \emptyset$ but $c_j, c_z \in P$, a contradiction to \cref{cla:sym_fen_strict-2}(\ref{cla:sym_fen_strict,1iii}).
 Summarizing,~$\mathcal{K}$ is an exact cover, as desired.
 \end{proof}
 \noindent The correctness follows from Claims~\ref{cla:sym_fen_strict-1} \& \ref{cla:sym_fen_strict-3}.
\end{proof}
}

The following two theorems complement \cref{thm:ns_is_acyclic} regarding~$\fas$ and show that determining Nash (resp.\ individually) stable partitions remain hard even if both $\fas$ and $\maxdeg$  are bounded. 
\newcommand{\fennashfasone}{%
  \FEN-\NS remains \np-complete even if $\maxdeg=9$ and $\fas=1$, and both the friendship and the enemy graphs are respectively acyclic.
}
\begin{theorem}[\appendixsymb]
\label{thm:ns-deg_fas}
\fennashfasone. %
\end{theorem}

\appendixproofwithstatement{thm:ns-deg_fas}{\fennashfasone}
{\begin{proof}
    We give a reduction from \pxct.
    Let $I=(\mathcal{X}, \mathcal{C})$ be an instance of \xct, where $\mathcal{C}=C_1,\dots,C_{\emm}$ and $\mathcal{X}=[3\enn]$.  We construct an instance $I'$ of \NS as follows.
\begin{compactitem}[--]
    \item For each element $i\in [3\enn]$ we have an element agent $a_i$.
    \item For each set $C_j$, $j\in [\emm ]$ we have a set agent $c_j$.
    \item We add $s_1,\dots s_{6\enn-1}$ special agents.
\end{compactitem}

First, we define the friendship arcs.
\begin{align*}
  E(\goodG)  = &  \{ (s_i,s_{i+1}) \mid i\in [6\enn-2] \} \cup \{ (a_1,s_1) \} \cup\\
               & \{ (a_i,s_{2i-2}),(a_i,s_{2i-1})\mid i\in [3\enn],i\ne 1 \} \cup \\
               & \{ (a_i,a_{i-1}) \mid i\in [3\enn],i\ne 1 \} \cup\\
  & \{ (a_i,c_j) \mid i \in [3\hat{n}], j \in [\hat{m}] \text{ with } i\in C_j\}.
\end{align*}
The enemy arcs are $(s_{6\enn -1},a_{3\enn})$ and $(c_j,c_k)$ for all $j,k\in [\emm]$ such that $j<k$ and $C_j\cap C_k\ne \emptyset$.
All other relationships are neutral.

\noindent One can verify that the maximum degree is $9$:
\begin{compactitem}[--]
  \item Each element agent has at most two friends from the special agents, at most one from the element agents and at most three from the set agents. It is also friends to at most one agent and enemy to at most one agent, so their degree is bounded by $8$.
  \item Each special agent has at most one friend and at most one enemy, and is a friend to at most one element agent and one special agent.
  Thus, the degree of each special agent is bounded by four.
  \item Each set agent is a friend to at most three element agent and is in an asymmetric enemy relation with at most six other set agents (as $C_j$ intersects at most 6 other sets). Thus the degree of each set agent is bounded by $9$.
\end{compactitem}
One can also verify that the enemy graph is acyclic with e.g., the following topological order~$(s_{6\enn -1},a_{3\enn}, c_1,\dots ,c_{\emm},a_1$, $\dots, a_{3\enn -1},$ $s_1,\dots,s_{6\enn -2})$. 
Now if we remove the enemy arc~$(s_{6\enn -1},a_{3\enn})$, then the union of the friendship and the remaining enemy graph (containing all friendship arcs) is acyclic: $a_{3\enn},a_{3\enn -1},\dots, a_1,s_1,\dots,s_{6\enn -1},$ $c_1,\dots,c_{\emm}$ is a topological order of the vertices.

It remains to show the correctness.
\begin{clm}
If $I$ admits an exact cover, then $I'$ admits a Nash stable partition.
\label{claim:ns1}
\end{clm}
\begin{proof}[Proof of claim~\ref{claim:ns1}]
\renewcommand{\qedsymbol}{$\diamond$}
Assume that $I$ admits an exact cover~$\mathcal{K}$.
Let the partition $\Pi$ be the following: $\{\{ s_i \mid i\in [6\enn-1]\}\} \cup \{ \{a_i,c_{j},\mid i\in [3\enn],C_j\in \mathcal{K}\}\} \cup \bigcup_{C_j\notin \mathcal{K}
  }\{\{ c_j\}\}$.
  Observe that every agent in a non-singleton coalition has zero enemies in her coalition.
  Hence, no agent prefers to be alone, implying that $\Pi$ is individual rational.
  It remains to show that no agent envy any other coalition.
  
  The set agents~$c_j$ have no friends in the graph and they have no enemies in $\Pi$, as the sets in the cover did not intersect.
  Therefore they do not envies any other the coalitions.

  The element agents~$a_i$ all have two friends, except $a_1$ has one friend and no enemies in $\Pi$.
  Therefore, they could only envy a coalition containing at least three friends for them (or two for $a_1$). They are only friends with the set agents and two special agents. The remaining set agents are alone, so they cannot envy such a coalition. No $a_i$ envies the coalition $\{ s_1,\dots, s_{6\enn -1}\}$ either, as $a_1$ has only one friend, and the others have only two friends inside.

  Finally, the special agents $s_1,\dots , s_{6\enn -1}$ have all their friends and no enemies in $\Pi$, therefore they also do not envy any coalition.
  This concludes that $\Pi$ is Nash stable. 
\end{proof}
The following observation follows from the fact that the set agents have no friend, and among two intersecting set agents, one considers the other an enemy.
\begin{observation}
\label{obs:ns}
Each $a_i$ agent can have only one friend from the set agents in any Nash stable solution.
\end{observation}
\begin{clm}
\label{claim:ns2}
If $I'$ admits a Nash stable coalition structure, then $I$ admits an exact cover.
\end{clm}
\begin{proof}[Proof of Claim~\ref{claim:ns2}]
\renewcommand{\qedsymbol}{$\diamond$}
Let $\Pi$ be any Nash stable coalition structure.
We will show that all element agents must be together and with exactly $\enn$ set agents, whose sets corresponds to an exact cover. 

Since $\Pi$ is Nash stable, $s_1$ does not envy any coalition. This is only possible, if $s_1$ is together with her only friend $s_2$. With the same argument, we can conclude that $s_1,\dots s_{6\enn -1}$ are all together. Let $P_1$ denote the coalition containing them.

We claim that no element agent is contained in $P_1$.
Towards a contradiction, suppose, there is some~$a_i$ in $P_1$.
If $i<3\enn$, then $a_{i+1}$ will have at least three friends in $P_1$.
It has three remaining set agents as friends. By \cref{obs:ns}, it can get at most one in $\Pi$. Therefore, if $a_{i+1}$ is not in $P_1$, then she envies it.
Hence we get that $a_{i+1}\in P_1$, and with the same argument that $a_{3\enn}\in P_1$. But $s_{6\enn -1}\in P_1$ has no friends and $a_{3\enn}$ is an enemy to her. Hence, she would rather be alone, a contradiction. 

Therefore, no element agents can be in $P_1$.
However, they have two friends in $P_1$ (except $a_1$ has only 1) and no enemies in the whole graph, so in order for them not to envy $P_1$, they must get at least two friends in $\Pi$ (except $a_1$ gets one).
By \cref{obs:ns}, this is only possible if $a_{3\enn},a_{3\enn-1},\dots ,a_{1}$ are all together in a coalition~$P_2$ with $P_2\neq P_1$, where each of them obtains at least one more friends from the set agents.

As no two set agents can be inside $P_2$, whose corresponding sets intersect, but each $a_i$ gets a set agent friend, this implies that there are exactly $\enn$ set agents in $P_2$. Thus, the corresponding sets must form an exact over. 
\end{proof}
\noindent The correctness follows from \cref{claim:ns1}--\ref{claim:ns2}.
\end{proof}
}

\newcommand{\clmfenisconstruction}{%
  The constructed instance satisfies that $\maxdeg=18$, the friendship and enemy graphs are respectively acyclic, and $\fas=1$.
}
\newcommand{\fenisfasdelta}{
  \FEN-\IndS remains \np-complete even if $\maxdeg=18$ and $\fas=1$ such that the friendship graph has one feedback arc and the enemy graph is acyclic.
}
\newcommand{\clmisforward}{
\begin{compactenum}[(i)]
  \item \label{claim:FEN-is-1,i}
  For each~$i\in [3\enn]$, agents~$u_i,u_i^1,u_i^2$ must be together in~$\Pi$. 
\item \label{claim:FEN-is-1,ii} For each $i\in [3\enn]$, agent~$a_i$ must have at least two friends in $\Pi$ and at most one of them can be a set agent~$c_j$.
\item \label{claim:FEN-is-2,i} There is no $i\in [3\enn]$ such that $a_i$ and $U_i$ are together in~$\Pi$.
\item \label{claim:FEN-is-2,ii} There is no $i\in [3\enn]$ s.t.\ $a_i$ and $U_{(i\bmod 3\enn)+1}$ are together~$\Pi$. %
\end{compactenum}
}
\begin{theorem} %
\label{thm:FEN-is-fas+delta}
\fenisfasdelta
\end{theorem}
\looseness=-1
  \begin{proof} %
    We reduce from \pxct; we will not use the planarity property though.
    Let $I=([3\enn], \mathcal{C})$ be an instance of \pxct, where $\mathcal{C}=\{C_1,\dots,C_{\emm}\}$.
    Without loss of generality, assume that $\enn$ is odd. %
    For each element~$i\in [3\enn]$, create an element agent~$a_i$, a leader agent~$u_i$, and two follower agents $u_i^1,u_i^2$; define $U_i=\{u_i,u_i^1,u_i^2\}$.
    For each set $C_j\in \mathcal{C}$, create a set agent $c_j$.
     The friendship (resp.\ enemy) graph contains the following arcs, where $i+1$ is taken as $(i\bmod 3\enn)+1$: %
  $E(\goodG)=\{ (u_i^z,u_i) \mid (i,z)\in [3\enn] \times [2]\} \cup \{ (a_i,a_{i+1})\mid i\in [3\enn] \}$ $\cup$ $\{ (a_i,c_j) \mid i\in [3\enn]$ and $j\in [\emm]$ with $i\in C_j\} \bigcup \{ (a_i,u_i^1),(a_i,u_i^2),(a_i,u_{i+1}^1),(a_i,u_{i+1}^2)\}$.
  $E(\badG)=\{ (u_i,u_{i+1})\mid i\in [3\enn -1]\} \cup \{ (u_1,u_{3\enn})\} \bigcup \{ (a_i,u_{i+1})\mid i\in [3\enn ] \} \cup \{ (c_j,c_t)\mid j,t\in [\emm ]\colon j<t \wedge C_j\cap C_t\ne \emptyset \} \cup \{ (u_i,c_j), (u_{i+1}, c_j)$, $(u_{i+2}, c_j)\mid (i,j)\in [3\enn]  \times [\emm ]\colon i \in C_j \}$. %
The construction is illustrated in \cref{fig:FEN-is-fas+delta}, and satisfies: 
\begin{claim}[\appendixsymb]\label{claim:fen-is-fas=1}
  \clmfenisconstruction
\end{claim}

     \begin{figure}[t!]
     \centering
     \def\xdif{1.6}
     \def\xdiff{1.9}
     \def\ydif{1.5}
       \begin{tikzpicture}[scale=0.8,every node/.style={scale=0.8}, >=stealth', shorten <= 2pt, shorten >= 2pt]

         \begin{scope}[shift={(-3.2,-1.6*\ydif)}]
		\foreach \x / \y / \n / \nn / \typ / \p / \dx / \dy in {
      2.5/1.3/a/a_i/pn/{above left}/-1/-4,
      1/.5/u1/u^1_i/pn/{above left}/-1/-4,
      2/.5/u2/u^2_i/pn/{above left}/-1/-4,
      3/.5/u3/u^1_{i + 1}/pn/{above right}/-1/-4,
      4/.5/u4/u^2_{i + 1}/pn/{above right}/-1/-4,
      1.5/0/u/u_i/pn/{left}/-1/1,
      3.5/0/up/u_{i + 1}/pn/{right}/-1/1}{
        \node[\typ] at (\xdif * \x,\xdif * \y*0.8) (\n) {};
        \node[\p = \dx pt and \dy pt of \n] {$\nn$};
      }
      \foreach \n / \x / \group in {0/0/{(i -2)}, 1/1.8/{(i - 1)}, 2/3.3/{i}, 3/4.8/{(i + 1)}} {
        \node[] at (\xdif * \x + 0.3, -1) (h\n) {};
        \foreach \i / \dir / \lab in {1/{below right}/3,2/{below}/2,3/{below left}/1} {
          \node[pn, below left= 0 and 0.5*(-1 + \i) of h\n] (u\n\i) {};
          \node[\dir = 0 pt and -7 pt of u\n\i] {$c_{\group_\lab}$};
        }		
      }      
      
       \begin{pgfonlayer}{bg}
           \foreach \s / \t / \aa / \type in {
           a/u1/0/fc,
           a/u2/0/fc,
           a/u3/0/fc,
           a/u4/0/fc,
           u1/u/0/fc,
           u2/u/0/fc,
           u3/up/0/fc,
           u4/up/0/fc,
           u/up/0/ec,
           a/up/40/ec} {
             \draw[->, \type] (\s) edge[bend right = \aa] (\t);
           }
           \foreach \i in {1,2,3} {
           		\foreach \j in {0,1,2} {
					 \draw[->, ec] (u) edge[] (u\j\i);        		
           		}
           		\foreach \j in {1,2,3} {
					 \draw[->, ec] (up) edge[] (u\j\i);        		
           		}
           }
      \end{pgfonlayer}
	\end{scope}       

      \foreach \s / \i / \d in {{i+3}/0/{right},{i+2}/1/{right},{i+1}/2/{right},{i}/3/{left}} { %
		\node[sn] at (-1.5 * \xdif * \i + \xdif*1.5,0) (u\i) {};
		\node[\d = 0pt and -1pt of u\i] {$U_{\s}$};
	  }      
	  \foreach \s / \i in {{i+2}/0,{i+1}/1,{i}/2} { %
		\node[pn] at (-1.5 * \xdif *\i + \xdif*.75, \ydif*0.8) (a\i) {};
		\node[above left = -1 pt and -4 pt of a\i] {$a_{\s}$};
	  }
          \foreach \i in {1, 2, 3} {
		\node[pn] at (-3.75*\xdif + 0.75*\xdif * \i,-1*\ydif*0.7) (c\i) {};
		\node[below = 1pt and -2pt of c\i] {$c_{i_\i}$};
		\begin{pgfonlayer}{bg}
		\draw[->, fc] (a2) edge[] (c\i);
		\foreach \j in {1,2,3} {
			\draw[->, ec] (u\j) edge[] (c\i);

		}
		\end{pgfonlayer}
          }	
	  \begin{pgfonlayer}{bg}
	  \foreach \i / \j  in {0/1, 1/2, 2/3} {
	  	\draw[->, fc] (a\i) edge[bend left=15] node[grouplabel,pos=.45] {\small 2f} (u\i);
	  	\draw[->, fc] (a\i) edge[] node[grouplabel,midway] {\small 2f} (u\j);
                \draw[->, ec] (a\i) edge[bend right=10] node[grouplabel,pos=.65] {\small 1e} (u\i);
	  }
	  \foreach \i / \j in {0/1, 1/2} {
	  	\draw[->, fc] (a\j) edge[] (a\i);
	  }
	  \foreach \i / \j / \ang in {1/2/30, 2/3/30, 1/3/40} {
            \draw[->, ec] (c\i) edge[bend right=\ang] (c\j);
	  }
	  \end{pgfonlayer}  
	 
        \end{tikzpicture}\caption{Illustration for the proof of \cref{thm:FEN-is-fas+delta}. Top: An overview of the reduction, where solid blue (resp.\ dashed red) arcs indicate friends (resp.\ enemies).
          For each~$z \in [3\enn]$, let $C_{z_1}, C_{z_2}, C_{z_3}$ denote the sets containing it such that $z_1 < z_2 < z_3$.
          Bottom: Description of the arcs related to $U_i = \{u_i, u^1_i, u^2_i\}$ and $U_{i+1} = \{u_{i + 1}, u^1_{i + 1}, u^2_{i + 1}\}$.}\label{fig:FEN-is-fas+delta}
     \end{figure}
    
\appendixproofwithstatement{claim:fen-is-fas=1}{\clmfenisconstruction}
{\begin{proof}[Proof of claim~\ref{claim:fen-is-fas=1}]
\renewcommand{\qedsymbol}{$\diamond$}
The agents with the highest degree are the set agents~$c_j$.
They can have three incoming friendship arcs from the element agents~$a_i$,
at most six (incoming or outgoing) enemy arc of the type $(c_j,c_t)$ (as each $a_i$ is in at most three sets) and at most nine incoming enemy arcs of the type $(u_i,c_j)$ (if $C_j=\{ j_1,j_2,j_3\}$, then $u_{j_1+z},u_{j_2+z},u_{j_3+z}$ for $z\in \{0,1,2\}$ consider $c_j$ an enemy), thus altogether at most 18 incoming and outgoing arcs.
The follower agents~$u_i^z$, $(i,z)\in [3\enn ] \times [2]$ agents have only two incoming and one outgoing friendship arcs and no adjacent enemy arcs.
The leader agents~$u_i$, $i\in [3\enn ]$, have two incoming friendship arcs (from the two followers) and at most $2+1+9=12$ adjacent enemy arcs:
two of the form $(u_i,u_{i+1})$, one of the form $(a_{i-1},u_i)$, and at most $3\cdot 3 =9$ of the form $(u_i,c_{j})$, with  $C_j\cap \{ i-2,i-1,i\} \ne \emptyset$ as each element is in at most three sets. Thus their degree is at most 14.
The element agents~$a_i$ have at most nine adjacent friendship arcs (four outgoing to the follower agents, three outgoing to the set agents, and one incoming and one outgoing arc from the previous and to the next element agents, respectively). Also, they have no incoming enemy arc and only one outgoing enemy arc. Hence, their degree is at most 10.

To see that only the friendship graph has one feedback arc while the enemy graph is acyclic, remove the arc $(a_{3n},a_1)$.
Then, the order $a_1,\dots,a_{3n},u_1^1,\dots, u_{3n}^1, $ $u_1^2,\dots,u_{3n}^2,$ $u_1,\dots,u_{3n},c_1,\dots,c_{3n}$ is a topological order of the vertices in the remaining graph.
\end{proof}
}

\noindent It remains to show that $I$ admits an exact cover if and only if the constructed instance has an individually stable partition.
For the ``if'' part, let $\Pi$ be an individually stable partition.
We first observe: 
\begin{claim}[\appendixsymb]
\label{claim:FEN-is-1} \clmisforward
\end{claim}
\appendixproofwithstatement{claim:FEN-is-1}{\clmisforward}{
\begin{proof}[Proof of \cref{claim:FEN-is-1}]\renewcommand{\qedsymbol}{$\diamond$}
We start with Statement \eqref{claim:FEN-is-1,i}.
Suppose there is an $i$, such that it does not hold. Then, there is a $z\in [2]$, such that $u_i^z$ is not with $u_i$. As $u_i^z$ is not an enemy to anyone, and her only friend is $u_i$, it would deviate to the coalition containing $u_i$, which would accept it, contradicting individual stability.

Now we prove Statement~\eqref{claim:FEN-is-1,ii}. Suppose $a_i$ has only $1$ friend in $\Pi$. But, by Statement~\eqref{claim:FEN-is-1,i} it has at least $2$ friends in the coalition $P_i$ containing $u_i$. As $a_i$ is not an enemy for any agent, it holds that $a_i$ prefers $P_i$ to $\Pi(a_i)$ and $P_i$ would accept $a_i$, contradiction. Now assume $a_i$ has two set agents as friends in her coalition: $c_j$ and $c_t$ with $j<t$. Then, $c_j$ has no friends, but has an enemy, which contradicts that $\Pi$ is individually rational.

For Statement (\ref{claim:FEN-is-2,i}), suppose that $a_i$ and $U_i$ are in one coalition, call it $P_i$. This means that agent $a_{i-1}$ has at least 3 friends in $P_i$.

Suppose that $a_{i-1}\notin P_i$. As $a_{i-1}$ is not an enemy to anyone, it must hold that she does not envy $P_i$. Hence, $a_{i-1}$ must have at least three friends in her coalition. As she can only have one set agent friend by \cref{claim:FEN-is-1}(\ref{claim:FEN-is-1,ii}), this can only happen if she is together with a set agent~$c_j$ and with $U_{i-1}$. However, this would imply that $u_{i-1}$ is together with a set agent $c_j$ such that $i-1\in C_j$, so $u_{i-1}$ has an enemy, namely $c_j$, which contradicts that $\Pi$ is individually rational. Hence, $a_{i-1}\in P_i$. 

We claim that $a_{i-2}$ is together with $U_{i-2}$.
Suppose the contrary. We distinguish two cases and obtain a contradiction in both of them. First, if she is also not together with $U_{i-1}$, then to have 2 friends by \cref{claim:FEN-is-1}(\ref{claim:FEN-is-1,ii}) , it must hold that $a_{i-2}$ is together with a set agent $c_j$ such that $i-2\in C_j$ and also $a_{i-1}$. As $a_{i-1}\in P_i$, we get that $\{a_{i-2},c_j\} \subset P_i$. However, since $u_i\in P_i$ and $C_j$ contains $i-2$ we get that $u_i$ has an enemy in $P_i$, namely $c_j$, contradiction.
Second, if $a_{i-2}$ is together with $U_{i-1}$, then she can have no other friends: she cannot have her friends in $U_{i-2}$ by our assumption, she cannot have any set agent friends, as they are an enemy to $u_{i-1}$, and she also cannot have $a_{i-1}$, because that would imply that $\{ u_{i-1},a_{i-1},u_i\} \subset P_i$, so $u_{i-1}$ (or $u_i$, if $i=1$) has an enemy, namely $u_i$ (or $u_{3n}$, if $i=1$).
However, since $a_{i-2}$ has an enemy in $U_{i-1}$, she prefers the coalition containing $U_{i-2}$.
Since $a_{i-2}$ is not an enemy to anyone, this would contradict the individual stability of $\Pi$. So finally, we obtain that $a_{i-2}$ must be together with $U_{i-2}$.

By iterating the argument above, we obtain that $a_{i-4}$ and $U_{i-4}$ are together in $\Pi$, and so on.
Since $3\enn$ is odd, we also get that $a_{i-1}$ and $U_{i-1}$ must be together. Hence, $u_{i-1}\in P_i$, but one of $u_{i-1}$ and $u_i$ is an enemy to the other one in $P_i$, contradiction. This concludes the proof of Statement~(\ref{claim:FEN-is-2,i}).

For Statement~\eqref{claim:FEN-is-2,ii}, suppose that $a_i$ and $U_{i+1}$ are together. This implies that $a_i$ can have none of her set agents friends, as they are enemies to $u_{i+1}$. She also cannot have her 2 friends from $U_i$, because that would imply that $u_i$ or $u_{i+1}$ has an enemy in $\Pi$. Since $a_i$ has at least two friends and no enemies in the coalition containing $U_i$ (because her only enemy is already together with her in $U_{i+1}$) it must hold that $a_i$ has at least three friends, so $a_{i+1}$ is together with $a_i$. Hence, $a_{i+1}$ is together with $U_{i+1}$. But this contradicts (\ref{claim:FEN-is-2,i}).
\end{proof}
}
By \cref{claim:FEN-is-1}\eqref{claim:FEN-is-2,i}--\eqref{claim:FEN-is-2,ii}, for each $i\in [3\enn]$, agent~$a_i$ does not have any friends from $U_i$ or $U_{i+1}$. By \cref{claim:FEN-is-1}(\ref{claim:FEN-is-1,ii}) $a_i$ must have at least two friends and at most one of them can be a set agent.
Therefore, we obtain that in $\Pi$, each agent $a_i$, $i\in [3\enn]$ must have exactly one set agent friend and $a_{i+1}$ as friends. This implies that $a_1,\dots,a_{3n}$ are all together and each of them has exactly one set agent friend in the coalition. Hence, the sets corresponding to those set agents must form an exact cover.

For the ``only if'' part, let $\mathcal{K}$ be an exact cover.
Then, define a coalition~$P=\{c_j\mid C_j\in \mathcal{K}\}\cup \{a_i\mid i\in [3\enn]\}$.
We claim that $\Pi =\bigcup_{i\in [3\enn]}\{ U_i\} \cup \{P\} \cup \bigcup_{C_j\in \mathcal{C}\setminus \mathcal{K}}\{\{c_j\}\}$, consisting of $3\enn+\emm+1$ coalitions, is Nash stable, and hence individually stable. 
Since $c_j$ has no friends in $G$ for all~$j\in [\emm ]$, and they have no enemies in $\Pi$, they do not envy any coalition.
Similarly, $u_i$ has no friends in $G$ for any $i\in [3\enn ]$ and she has no enemies in $\Pi$.
For all $i\in [3\enn ]$, $z\in [2]$, agent~$u_i^z$ has all of her friends and no enemies in $\Pi$, so she does not envy any coalition.
Finally, for each $i\in [3\enn ]$, $a_i$ has two friends and no enemies in $\Pi$, and among the other six friends of $a_i$, two of them are alone (the two other set agents~$c_j$ with $i\in C_j$), two of them are in $U_i$, and two of them are in $U_{i+1}$. Therefore, $a_i$ envies none of the coalitions.
\end{proof}

\section{Conclusion and Future Work}\label{sec:conclude}%
We resolved many complexity questions from the literature under the \FE\ as well as the \FEN\ model,
and significantly extended previous work for these two models. 
As an immediate open question, we do not know the complexity of \FE-\verif\ (resp.\ \FE-\sverif) $\maxdeg=6$ for the symmetric case.
For the case of $\maxdeg = 3$ with not necessarily symmetric preferences,
we can show that \FE-\verif\ remains \conp-hard, even if $\fas=1$, between the submission and the publication of the paper.
We conjecture that 
Last but not least, it would be interesting to know how the refined complexity for the variant with enemy aversion~\cite{dimitrov2006simple} behaves. %

\clearpage

\begin{acks}
This work, and Jiehua Chen, Sofia Simola, and Sanjukta Roy have been funded by the Vienna Science and Technology Fund (WWTF) [10.47379/ VRG18012]. Gergely Cs{\'a}ji acknowledges support from  the COST Action CA16228 on European Network for Game Theory, the Momentum Program of the Hungarian Academy of Sciences -- grant number LP2021-1/2021 and the  Hungarian National Research, Development and Innovation Office -- NKFIH, grant number  K143858.
\end{acks}

\bibliographystyle{ACM-Reference-Format} 

\balance 
\bibliography{cit}

\iflong
\clearpage
\begin{table}[t!]
  \centering
  \Large \textbf{\appendixtitle}
\end{table}
\bigskip

\appendix

\appendixtext

\fi

\clearpage

\end{document}